\documentclass[letter,11pt]{article}
\pdfoutput=1

\usepackage{tcs}

\usepackage{listings}
\newcommand{\Brac}[1]{\left [ #1 \right ]}

\newcommand{\Snote}[1]  {}

\begin{document}
\title{Entrywise Low-Rank Approximation and Matrix $p \rightarrow q$ Norms via Global Correlation Rounding}

\author{
Prashanti Anderson\footnote{Supported by NSF award no. 2238080.}\\
\texttt{paanders@csail.mit.edu} \\
MIT
\and
Ainesh Bakshi \\
\texttt{ainesh@nyu.edu} \\
NYU
\and
Samuel B. Hopkins\footnote{Supported by NSF award no. 2238080, MLA@CSAIL, FinTechAI@CSAIL, MIT Research Support Committee, MIT-Google Program for Computing Innovation.} \\
\texttt{samhop@mit.edu} \\
MIT
}

\maketitle

\begin{abstract}
    Given a matrix $A$, the goal of the entrywise low-rank approximation problem is to find $\argmin \|A-B\|_p$ over all rank-$k$ matrices $B$, where $\| \cdot \|_p$ is the entrywise $\ell_p$ norm.
    When $p = 2$ this well-studied problem is solved by the singular value decomposition, but for $p \neq 2$ the problem becomes computationally challenging.
    For every even $p > 2$ and every fixed $k$, we give the first polynomial-time approximation scheme for this problem, improving on the $(3 + \eps)$ approximation of Ban, Bhattiprolu, Bringmann, Kolev, Lee, and Woodruff, the bi-criteria approximation of Woodruff and Yasuda, and the additive approximation scheme of Anderson, Bakshi, and Hopkins.
    Prior algorithmic approaches based on sketching and column selection, which yielded a polynomial-time approximation scheme in the $p < 2$ setting, face concrete barriers when $p > 2$.
    Instead, we use the Sherali-Adams hierarchy of convex programs, and in so doing suggest a methodology for how to use convex hierarchies to design polynomial-time approximation schemes for continuous optimization problems.

    We use the same algorithmic strategy to give a new family of additive approximation algorithms for matrix $p \rightarrow q$ norms, which are intimately related to small-set expansion and quantum information.
    In particular, we give the first nontrivial additive approximation algorithms in the regime $p < 2 < q$.
\end{abstract}

\thispagestyle{empty}

\newpage
\thispagestyle{empty}

\tableofcontents

\newpage

\setcounter{page}{1}

\section{Introduction}
\label{sec:introduction}

We design new approximation algorithms for two fundamental families of high-dimensional continuous optimization problems: \emph{entrywise low-rank approximation} of matrices and matrix $p \rightarrow q$ norms. 
Our algorithms are obtained via a new analysis of the Sherali-Adams hierarchy of convex relaxations for both problems.

\paragraph{Entrywise Low-Rank Approximation.}
The goal of the \emph{low rank approximation} problem (LRA) is to find the closest low-rank matrix to a given $A \in \R^{n \times d}$.
Low-rank approximation is ubiquitous in algorithms, machine learning, statistics, and beyond.

We study low rank approximation with entrywise error, where the approximation error
is measured via an entrywise $\ell_p$ norm
\[
\|A - B\|_p^p = \sum_{i,j} |A_{ij} - B_{ij}|^p.
\]
When $p = 2$, this is the well-studied Frobenius norm, which admits exact low-rank
approximation via singular value decomposition (SVD). The cases $p = 0$ and
$p \in [1,2)$ have received intensive study, generating numerous new algorithmic
insights spanning sketching and subspace embeddings
\cite{sohler_woodruff_stoc2011_l1embed,woodruff_fnt2014_sketching,woodruff_zhang_colt2013_exponential},
randomized linear algebra
\cite{sarlos_focs2006_randomprojections},
robust statistics 
\cite{candes_etal_jacm2011_rpca},
and fine-grained approximation algorithms
\cite{ban_etal_soda2019_ptas_lp_lra},
culminating in polynomial time approximation schemes (PTASes)
achieving $(1+\varepsilon)$-approximation in time polynomial in $n$
for fixed $k$ and $p$
\cite{ban_etal_soda2019_ptas_lp_lra,mahankali_woodruff_arxiv2020_l1css_ptas}.

We focus on $p > 2$, which penalizes any poorly-approximated entry more aggressively than $p=2$ would, thus pushing the approximating matrix $B$ to approximate $A$ uniformly well on all entries.
When $p \rightarrow \infty$, entrywise $\ell_p$ low-rank approximation turns into the problem of computing the \emph{approximate rank} of $A$, an important quantity in communication complexity and beyond, which is NP hard to compute even when $k=1$ \cite{AlonLeeShraibmanVempala2013ApproxRank,LeeShraibman2009LowerBounds,LinialShraibman2009FactorizationNorms,ChattopadhyayMandeSherif2019LogApproxRankFalse,AnshuBenDavidGargJainKothariLee2017SeparatingQuantumApproxRank,gillis2019low}.

The $p > 2$ case remains algorithmically mysterious.
The state of the art algorithms with multiplicative approximation guarantees, draw on the tools that lead to a PTAS in the $p < 2$ case, are a $(3+\eps)$ approximation by \cite{ban_etal_soda2019_ptas_lp_lra} and a bi-criteria approximation of \cite{DBLP:conf/stoc/Woodruff023}, which is allowed to output a matrix of rank larger than $k$.
(See Table~\ref{tab:lra-intro}.)
There are significant barriers to extending these approaches to a PTAS for the $p > 2$ case.\footnote{Notably, in the $p \in [1,2)$ case the PTAS of \cite{ban_etal_soda2019_ptas_lp_lra} relies on the existence of $p$-stable distributions, which do not exist for $p > 2$.}\footnote{Interestingly, the work \cite{cohen2024ptas} which gives a PTAS for the $p=0$ case uses a different set of techniques, borrowing from algorithms for dense constraint satisfaction problems. Our algorithm can also be viewed as a (significant) generalization of a different set of tools for solving dense CSPs, via LP/SDP hierarchies.}
It has therefore been unclear whether there is a PTAS for $\ell_p$ LRA when $p > 2$, which seems to require new techniques, or whether this problem is fundamentally harder than in the $p < 2$ case.

Our first main result resolves this question by giving the first PTAS for $p > 2$.

\begin{theorem}
    \label{thm:lra-intro}
    For every $\eps > 0$ and even $p \in \N$, there is an algorithm running in time $(nd)^{\poly(k/\eps)}$ which takes a matrix $A \in \R^{n \times d}$ whose entries are integers of magnitude at most $\poly(nd)$ and with probability $0.99$ finds $U \in \R^{n \times k}, V \in \R^{d \times k}$ such that
    \[
    \| A - UV^\top \|_p^p \leq (1 + \eps) \cdot \min_{\rank(B) \leq k} \|A - B\|_p^p \, .
    \]
\end{theorem}

The assumption that the entries of $A$ are bounded integers -- or equivalently that they have bounded bit-complexity -- is relatively standard in the low-rank approximation literature, needed also for the only known PTAS in the $p < 2 $ case \cite{ban_etal_soda2019_ptas_lp_lra}.
We conjecture that our algorithm, or a minor variant thereof, obtains an analogous guarantee in the case that $p > 2$ is odd, but our analysis does not currently extend to this setting.

Our algorithm combines the $O(1)$-approximation of \cite{ban_etal_soda2019_ptas_lp_lra} with the Sherali-Adams hierarchy of linear programs. 
The only prior work we are aware of to use convex programming to solve entrywise LRA is \cite{anderson2025additiveapproximationschemeslowdimensional}, which obtains an \emph{additive} approximation scheme in the $p > 2$, $k = 1$ case in quasipolynomial time.
We introduce a range of new technical tools, discussed below, to obtain multiplicative rather than additive error guarantees, polynomial rather than quasi-polynomial running time, and especially to handle $k > 1$.

\begin{table}[t]
\centering
\footnotesize
\setlength{\tabcolsep}{5pt}
\begin{tabular}{@{}lccc@{}}
\toprule
Reference & Approx.\ factor & Output rank & Running time \\
\midrule
\cite{pmlr-v70-chierichetti17a}
& $k+1$
& $k$
& $d^k \poly(nd)$ \\
\cite{ban_etal_soda2019_ptas_lp_lra}
& $3+\varepsilon$
& $k$
& $n^{\poly(k/\varepsilon)}$ \\
\cite{DBLP:conf/stoc/Woodruff023}
& $1+\varepsilon$
& $\widetilde O\!\left(k^{p/2-2/p+1}\log d/\varepsilon^{1+2/p}\right)$
& $\poly(n,d,k,\varepsilon^{-1})$ \\
This work
& $1+\varepsilon$
& $k$
& $(nd)^{(k/\varepsilon)^{O_p(1)}}$ \\
\bottomrule
\end{tabular}
\caption{Selected algorithms for entrywise $\ell_p$ rank-$k$ approximation for fixed finite $p>2$. The target rank is $k$ throughout. The column ``Output rank'' records the rank of the returned matrix, so \cite{DBLP:conf/stoc/Woodruff023} gives a bicriteria approximation. \cite{ban_etal_soda2019_ptas_lp_lra} assumes bounded-bit entries, and in this work we additionally assume that $p$ is an even integer.}
\label{tab:lra-intro}
\end{table}

\paragraph{Matrix $p \rightarrow q$ Norms.}
For $p,q \geq 1$, the $p\rightarrow q$ norm of a matrix $A \in \R^{n \times d}$ is
\[
\|A\|_{p \rightarrow q} = \max_{x \in \R^d} \frac{ \|Ax\|_q}{\|x\|_p} \, ,
\]
where $\|\cdot\|_p$ and $\|\cdot \|_q $ denote the $\ell_p$ and $\ell_q$ norms, respectively. 
The $2 \rightarrow 2$ norm is simply the maximum singular value of $A$, which can be computed in polynomial time.
For most choices of $p,q$ this problem is hard to solve exactly, with varying degrees of hardness of approximation \cite{HendrickxOlshevsky2010MatrixPNormsNPHard,BhaskaraVijayaraghavan2011ApproximatingMatrixPNorms,barak2012hypercontractivity,BhattiproluGhoshGuruswamiLeeTulsiani2023InapproximabilityMatrixPQ,BhattiproluLeeTulsiani2022GrothendieckSeparation}.
And yet, the problem of computing $p \rightarrow q$ norms when $p \neq 2$ and/or $q \neq 2$ is pervasive.
Examples include the $\infty \rightarrow 1$ norm, known as the Grothendieck problem and closely related to the cut norm, communication complexity, and discrepancy theory \cite{alon2006cutnorm,frieze1999quick,linial2008learning,LinialShraibman2009FactorizationNorms,lee2008discrepancy}, the $2 \rightarrow 1$ norm used in high-dimensional statistics \cite{hopkins2020mean} and the $2 \rightarrow q$ norms for $q > 2$, which are closely related to small-set expansion and unique games, quantum information, and algorithmic statistics \cite{barak2012hypercontractivity,brandao2015replacing,diakonikolas2025sos}.

We focus on the ``hypercontractive'' $p < q$ case.
Even for the special case of the $2 \rightarrow q$ norms, they are NP-hard to approximate within inverse polynomial multiplicative error, and any constant factor approximation would refute both the exponential time hypothesis and the small set expansion hypothesis \cite{barak2012hypercontractivity}.
But algorithms with other forms of approximation guarantees -- in particular, \emph{additive} approximation error -- still have broad consequences, for instance in quantum information \cite{brandao2015replacing} and low-rank approximation \cite{ban_etal_soda2019_ptas_lp_lra}, and are of interest in their own right to delineate which instances $A$ have easy-to-compute $p \rightarrow q$ norms.

Our second main contribution is a new additive approximation algorithm for the $p \rightarrow q$ norm.

\begin{theorem}
\label{thm:intro-pq}
  Let $q \geq 2$ be an even integer with H\"older dual $q^* = \tfrac q {q-1}$, and let $p \in [q^*,2]$ have H\"older dual $p^* = \tfrac p {p-1}$ which is an even integer.
  For every $\eps > 0$, there is a $(nd)^{O(\eps^{-4q})}$ time algorithm which takes a matrix $A \in \R^{n \times d}$ and with probability $0.99$ returns $x \in \R^d$ such that
  \[
  \frac{\|Ax\|_q^q}{\|x\|_p^q} \leq \|A\|_{p \rightarrow q}^q \leq \frac{\|Ax\|_q^q}{\|x\|_p^q} + \eps \cdot \sum_{i \leq n} \|A_i\|_{p^*}^q
  \]
  where $A_1,\ldots,A_n$ are the rows of $A$.
\end{theorem}
An illustrative example is the $4/3 \rightarrow 4$ norm, chosen because in this case both $p^*, q > 2$, a setting in which prior additive approximations suffer.
Theorem~\ref{thm:intro-pq} says that there is a polynomial time algorithm to approximate $\|A\|_{4/3 \rightarrow 4}^4$ to additive error $\epsilon \cdot \sum_{ij} A_{ij}^4$.
Note that by H\"older's inequality, $\sum_{ij} A_{ij}^4 \geq \|A\|_{4/3 \rightarrow 4}^4$.
The only prior additive approximation we are aware of for the $p < 2$ case, by \cite{brandao2015replacing}, instead yields incomparable additive error (see Section~\ref{sec:prior-work-pq-additive}).

More generally, to interpret the right-hand side of Theorem~\ref{thm:intro-pq}, note that $\sum_{i \leq n} \|A_i\|_{p^*}^q$ is a natural upper bound on $\|A\|_{p \rightarrow q}^q$ via H\"older's inequality.
The theorem says that there is an efficient algorithm which improves on this upper bound by any constant factor; alternatively, the $p \rightarrow q$ norm is easy to compute up to $(1+\eps)$ multiplicative error when $\|A\|_{p \rightarrow q}$ is within a constant factor of this upper bound.

All prior polynomial-time additive approximation algorithms for $p \rightarrow q$ norm either apply only to the case $p=2$ or incur error which can be larger than $\eps \cdot \sum_{i \leq n} \|A_i\|_{p^*}^q$ by a polynomial factor in $n$ or $d$ \cite{brandao2015replacing}.
Thus, we view Theorem~\ref{thm:intro-pq} as an important step to chart the algorithmic landscape in the thus-far underexplored $p < 2$ case.
We give more detailed comparison to prior algorithms with additive approximation guarantees in \cref{sec:prior-work-pq-additive}.

\paragraph{A Heavy-Tailed Pinning Lemma.}
A key technical piece of our rounding algorithms is a new pinning lemma for heavy-tailed random variables.
Pinning lemmas are measure-decomposition statements: after conditioning on, or ``pinning,'' a small number of coordinates of a high-dimensional random vector, the remaining coordinates have small average pairwise correlations.
They arose in statistical physics \cite{Montanari2008} and in rounding algorithms for semidefinite programs \cite{Barak2011,Raghavendra2012}, especially through the global correlation rounding paradigm.

For real-valued random variables, the standard form controls average covariance by average variance.

\begin{theorem}[Pinning Lemma, Bounded-Variance Version \cite{Montanari2008,Raghavendra2012,Barak2011}]
\label{thm:pinning-intro-variance}
  For all $t \in \N$ there exists $s \leq t$ such that if $S \subseteq [n]$ is a random set of size $s$ and $(X_1,\ldots,X_n)$ are any jointly distributed random variables,
  \[
  \E_{S} \E_{i,j \sim [n]} |\Cov(X_i, X_j \, | \, X_S)|  \leq O\Paren{\frac 1 {\sqrt t}} \cdot \E_{i \sim [n]} \Var(X_i) \, .
  \]
\end{theorem}

This is not enough for our applications.
In both the $p \rightarrow q$ norm algorithm and the low-rank approximation PTAS, the rounding loss contains covariance terms involving powers of the relaxed variables.
The relevant low moments are controlled, but the variances of these powered variables may be too large to use Theorem~\ref{thm:pinning-intro-variance}.
We instead prove the following heavy-tailed version, which bounds the same average covariance using pairwise moments below order $2$.

\begin{theorem}[Heavy-Tailed Pinning Lemma]
\label{thm:pinning-intro-heavy}
Let $1 < q < 2$ and $t \in \N$.
For any jointly distributed real-valued random variables $(X_1,\ldots,X_n)$, there exists $s \leq t$ such that if $S \subseteq [n]$ is a random set of size $s$, then
\[
  \E_S \E_{i,j\sim [n]} | \Cov(X_i,X_j \, | \, X_S) |
  \leq
  O\Paren{t^{-(q-1)/2}} \cdot
  \Brac{\E_{i,j \sim [n]} \E |X_i X_j|^q}^{1/q}.
\]
\end{theorem}

Theorem~\ref{thm:pinning-intro-heavy} is the form of pinning needed for the heavy-tailed covariance terms in our rounding analyses.
We also use weighted and bipartite variants, and versions for pseudodistributions, but the scalar statement above captures the main phenomenon.

\section{Technical Overview}
\label{sec:overview}
We turn to an overview of our techniques.
We start by reviewing the route from rounding convex programming hierarchies to pinning lemmas, which is by now standard \cite{Barak2011,GuruswamiSinop2012FasterSDP,Raghavendra2012,GuruswamiSinop2013RoundingLasserre,YoshidaZhou2014SAForDenseCSPs,BarakKelnerSteurer2014RoundingSoS,Risteski2016PartitionFunctionsHierarchies,JainKoehlerRisteski2018MeanFieldCorrelationRounding,AlevJeronimoTulsiani2019HDXSoS,HopkinsSchrammTrevisan2020SubexpSA,RaghavendraYau2020ListDecodableLearningSoS,bakshi2021list,CohenAddadLeeNewman2022CorrelationClusteringSA,CohenAddadLeeLiNewman2023PreclusteringSA,JeronimoSrivastavaTulsiani2023SoSListDecoding,AnandLeeMazzaliSharma2025MinCSPCompleteII_SA,BakshiCohenAddadJayaramHopkinsLattanzi2025SheraliAdamsEmbeddings,anderson2025additiveapproximationschemeslowdimensional}.
We use a sketch of our result for $p\rightarrow q$ norms to illustrate the importance of the heavy-tailed pinning lemma, and then discuss our PTAS for low-rank approximation (\cref{sec:ptas-intro}) and the proof of the heavy-tailed pinning lemma (\cref{sec:overview-heavy-tailed}).

\subsection{Sum-of-Squares Relaxation and Global Correlation Rounding for $p \rightarrow q$ Norms}
As a warm-up, we start with our additive approximation algorithm for $p \rightarrow q$ norms.
For concreteness, we focus on $4/3 \rightarrow 4$.
Our algorithms for more general $\|A\|_{p \rightarrow q}$ follow the exact same template.
The first step is to describe the convex relaxation and rounding scheme we use, both of which are fairly standard.

\paragraph{Setting Up the Convex Relaxation}
We can phrase $\|A\|_{4/3 \rightarrow 4}$ as the following bilinear optimization problem over the $\ell_{4/3}$ unit ball:
\[
\|A\|_{4/3 \rightarrow 4} = \max_{\|x\|_{4/3} = 1, \|y\|_{4/3} = 1} x^\top A y
\]
It is not too hard to show that there exists a (not too large) discrete subset $\calS$ of $\mathbb{R}$ such that there are $\ell_{4/3}$ unit vectors $x,y$ which satisfy $x^\top A y \geq \|A\|_{4/3 \rightarrow 4} - \eps \cdot \sum_{i \leq n} \|A_i\|_4^4$ and have all their entries in $\calS$.
Thus, we can restrict attention to $x,y$ with entries in $\calS$.
This lets us formulate $\max x^\top A y$ as a discrete assignment problem over $\calS^{n} \times \calS^d$.

The first step of our algorithm is to solve the degree $O(\eps^{-4})$ Sum-of-Squares (SoS) semidefinite programming relaxation of this discrete assignment problem.\footnote{Our algorithms would also work using a variant of the Sherali-Adams linear programming relaxation; we do not need the full power of SoS. But since the SoS notation has become standard, we will follow the SoS conventions.}
A solution to the degree-$\ell$ SoS relaxation is a \emph{pseudodistribution}.
For our purposes, a pseudodistribution is a collection of \emph{local distributions} $\{ \mu_T \}_{T \subseteq [n+d], |T| = \ell}$ on $\calS^\ell$.
These local distributions are \emph{consistent} in that $\mu_T$ and $\mu_{T'}$ must induce the same marginal distribution on $T \cap T'$.
Such a collection of local distributions induces a \emph{pseudoexpectation}, which is a linear operator 
\[
\pE \, : \, (\text{linear combinations of functions depending on $\leq \ell$ variables}) \rightarrow \R \, ,
\]
by defining, for each function $f(x,y)$ depending on at most $\ell$ variables $A \subseteq [n] \cup [d]$,
\[
\pE f(x,y) = \E_{(x_A,y_A) \sim \mu_{A}} f(x_A,y_A) \, ,
\]
and extending $\pE$ linearly.
We impose additional (convex) constraints on the local distributions to ensure that $\pE \|x\|_{4/3}^{4/3} = \pE \|y\|_{4/3}^{4/3} = 1$.
We can write the objective value of a solution to the relaxation as $\pE \sum_{i \leq n,j\leq d} A_{ij} x_i y_j$, which is a linear function of $\{\mu_T\}$, and hence we can maximize it over solutions.
Our next goal is to round $\{ \mu_T \}$ to a pair of vectors $\hat{x},\hat{y}$ without losing too much in objective value\footnote{It suffices to only consider the objective value; in fact, we can slightly modify the rounding procedure to ensure we produce vectors with norm at most $1$.}.

\paragraph{Global Correlation Rounding via Pinning Lemmas.}
Two key algorithmic operations comprise the global correlation rounding scheme.
First is \emph{conditioning:} choosing an index $i \in [n]$ or $j \in [d]$, sampling $x_i$ or $y_j$ from its marginal distribution, and replacing all the local distributions with their conditioned versions, which yields a new degree $\ell-1$ pseudodistribution.
Second is independent sampling: we can always draw each $x_i$ and $y_j$ from their one-wise marginals, though this will not preserve any correlations across coordinates\footnote{Note that independently sampling (in expectation) produces the same objective value as outputting the expectation of each coordinate on bilinear objectives. We utilize expectation rounding in our final algorithm to ensure that we produce vectors with bounded norm.}.
\begin{mdframed}
  \begin{algorithm}[Global Correlation Rounding \cite{Barak2011,Raghavendra2012} for $4/3\rightarrow 4$ Norm]
    \label{algo:lra-additive}\mbox{}
    \begin{description}
    \item[Input:] Pseudo-distribution $\pE$ of degree at least $\ell$ over variables $\{x_i\}_{i \in [n]} \cup \{y_j\}_{j \in[d]}$.
    
    \item[Operations:]\mbox{}
    \begin{enumerate}
        \item For every subset $\calT \subseteq [n] \cup [d]$, $\vert \calT \vert \leq \ell-2$ of variables:
        \begin{enumerate}
        \item Sample the values of $(x_{\calT},y_{\calT})$ from their joint marginal local distribution and condition on their values.
        \item Sample $\hat{x}_i, \hat{y}_j$ from the one-wise conditional marginals.
        \end{enumerate}
        \item Output the $\hat{x},\hat{y}$ achieving the best objective value.
    \end{enumerate}
    \end{description}
  \end{algorithm}
\end{mdframed}

By adding appropriate constraints to the convex relaxation, we can ensure that conditioning preserves objective value.
So the main thing we need to show is that at least one of the conditioned pseudodistributions satisfies the following ``small rounding loss'' inequality.
\[
\Abs{\text{SoS value} - \text{expected rounded value}} = \Abs{\sum_{i \in [n], j\in [d]} A_{ij} \cdot (\pE x_i y_j - \pE x_i \pE y_j)} \leq \eps \cdot \Paren{\sum_i \|A_i\|_4^4}^{1/4} \, .
\]

Now we can see the use of a pinning lemma: $$\sum_{ij} A_{ij} (\pE x_i y_j - \pE x_i \pE y_j) \leq \sum_{ij} |A_{ij}| \cdot |\Cov(x_i,y_j)|\,,$$ which is exactly a weighted sum of pairwise covariances.
(We abuse notation and use $\Cov(x_i,y_j)$ to denote the covariance of the variables $x_i,y_j$ in their $2$-wise local distribution.)
So bounding the rounding loss reduces to bounding a weighted average of pairwise covariances in the conditioned pseudodistribution.
The usual variance pinning lemma would do this if the relevant variables had controlled variances.

The underlying (pseudo)random variables $x_1,\ldots,x_n,y_1,\ldots,y_d$ might not have variances we can control -- our relaxation promises us nothing about $\Var(x_i), \Var(y_j)$.\footnote{Strictly speaking, these variances are bounded in terms of the largest-magnitude element of our discretization $\calS$. But using this bound would cost us a $\poly(n,d)$ factor, which we cannot afford.}
But we have a tight handle on the $4/3$ moments, because $\pE \|x\|_{4/3}^{4/3} \|y\|_{4/3}^{4/3} = \sum_{ij} \pE |x_i|^{4/3} |y_j|^{4/3} = 1$ via the constraints on our relaxation.
Thus, we invoke the heavy-tailed pinning lemma with moment parameter $4/3$.
Our weighted, bipartite\footnote{That is, allowing the two disjoint sets of variables $\{x_i\},\{y_j\}$} version of Theorem~\ref{thm:pinning-intro-heavy}, which also applies to pseudodistributions, shows that the small rounding loss inequality above holds for at least one choice of pinning subset $\calT$.

\subsection{PTAS for Entrywise Low-Rank Approximation}
\label{sec:ptas-intro}
We turn to Theorem~\ref{thm:lra-intro}.
We focus on $p=4$ for simplicity; the ideas easily extend to any even $p$.

\paragraph{Additive Approximation, Rank-One Case.}
We start by describing how to get an additive approximation algorithm to $\ell_4$ LRA with additive error $\eps \cdot \|A\|_4^4$. Even for this guarantee, the best prior algorithm \cite{anderson2025additiveapproximationschemeslowdimensional} runs only in quasipolynomial time.
As with $p \rightarrow q$ norm, we solve a degree $\poly(1/\eps)$ SoS relaxation obtained by discretizing the allowed set of values assumed by $u$ and $v$.
We show that there exists an optimal solution $u,v$ obeying $\|u\|_4^4, \|v\|_4^4 \leq O(\|A\|_4^2)$, and hence our convex program can enforce that $\pE \|u\|_4^4, \pE \|v\|_4^4 \leq O(\|A\|_4^2)$.

We once again deploy global correlation rounding.
To see the ``small rounding loss'' inequality we need in this setting, it helps to expand the objective function as a polynomial:
\[ \sum_{i \in [n], j \in [d]} (A_{ij} - u_i v_j)^4 = A_{ij}^4 - 4 A_{ij}^3 u_i v_j + 6 A_{ij}^2 u_i^2 v_j^2 - 4 A_{ij} u_i^3 v_j^3 + u_i^4 v_j^4\,.\]

We can break the loss in objective value incurred by independent rounding on a conditioned pseudodistribution $\pE$ into four terms:
\begin{align*}
& \Abs{\sum_{ij} 4 A_{ij}^3 (\pE u_i v_j - \pE u_i \pE v_j)} \, , \quad \Abs{\sum_{ij} 6 A_{ij}^2 (\pE u_i^2 v_j^2 - \pE u_i^2 \pE v_j^2)} \, , \quad \\
& \Abs{\sum_{ij} 4 A_{ij} (\pE u_i^3 v_j^3 - \pE u_i^3 \pE v_j^3)} \, , \quad 
\Abs{\sum_{ij} \pE u_i^4 v_j^4 - \pE u_i^4 \pE v_j^4} \, .
\end{align*}
We can show that one of the conditioned pseudodistributions discovered by global correlation rounding after conditioning on $\poly(1/\eps)$ coordinates makes the first three of these terms less than $\eps \cdot \|A\|_4^4$ using pinning lemmas.
(Again, variants of Theorem~\ref{thm:pinning-intro-heavy} and Theorem~\ref{thm:pinning-intro-variance} which allow for weighted averages of covariances, and two sets of variables $u$ and $v$.)
Crucially, the term involving $\pE u_i^3 v_j^3 - \pE u_i^3 \pE v_j^3 = \Cov(u_i^3,v_j^3)$ demands a heavy-tailed pinning lemma. We don't control the variances $\Var(u_i^3),\Var(v_j^3)$, so a variance-style pinning lemma won't work, but we do control the $4/3$ moments, since $\E |u_i^3|^{4/3} = \E u_i^4$.

For the last term, involving $\Cov(u_i^4,v_j^4)$, we are able to include constraints in our relaxation which ensure that no rounding loss is incurred.
This term factors as $\pE \|u\|_4^4 \|v\|_4^4 - \pE \|u\|_4^4 \pE \|v\|_4^4$, and we can ensure that this difference is zero, following \cite{anderson2025additiveapproximationschemeslowdimensional}. Anderson, Bakshi and Hopkins~\cite{anderson2025additiveapproximationschemeslowdimensional} follow the same template as above to obtain their additive approximation, but, lacking a heavy-tailed pinning lemma, they employ a pinning lemma suited to discrete random variables in an \emph{ad hoc} manner; it is unclear how to improve their approach from quasipolynomial to polynomial time except via the heavy-tailed pinning lemma.

\paragraph{From Additive to Multiplicative Approximation, Rank-One Case.}
We now describe the main ideas required to extend the additive approximation described above to a PTAS. 
The key idea is to ``warm start'' our additive approximation algorithm with a constant-factor optimal rank-one approximation.
Using e.g. \cite{ban_etal_soda2019_ptas_lp_lra}, in polynomial time we can find $x,y$ such that $\|A -xy^\top\|_4^4 \leq 3.1 \cdot \min_{u,v} \|A-uv^\top\|_4^4$.\footnote{Here $3.1$ stands in for any constant greater than $3$, since \cite{ban_etal_soda2019_ptas_lp_lra} gives a $3+\eps$-approximation algorithm.}
Then we search for a rank-one approximation to $A$ only in a small $\ell_4$ ball around $(x,y)$.
To enable this, we prove a structure lemma, which shows that some optimal rank-one approximation to $A$ must exist in a small ball around $(x,y)$.

\begin{lemma}[Nearby Optimal Solutions for Rank-One Approximation]
\label{lem:strong-ident-intro}
    Let $p$ be an even integer, $x,u \in \mathbb{R}^n$, and $y,v \in \mathbb{R}^m$. Then
    \[ \min_c \left(\norm{x+cu}_p^p \norm{y-v/c}_p^p + \norm{x-cu}_p^p \norm{y+v/c}_p^p\right) \leq O(1) \cdot \norm{uv^\top - xy^\top}_p^p\,.\]
\end{lemma}

Let us briefly sketch why this kind of inequality is true.
First, one can sketch the two-dimensional subspaces spanned by $x,u$ and by $y,v$ down to constant dimension while preserving the relevant $\ell_p$ norms up to constant factors.
In constant dimension, $\ell_p$ and $\ell_2$ norms are equivalent, so it is enough to understand the corresponding Frobenius-norm statement.
There, the key identity is
\[
  (x+cu)(y-v/c)^\top + (x-cu)(y+v/c)^\top = 2(xy^\top - uv^\top).
\]
Choosing $c$ to balance the $\ell_2$ norms of $x$ and $cu$ makes the two left-hand rank-one matrices orthogonal in the first factor.
Thus the squared Frobenius norm of their sum controls the sum of their squared Frobenius norms, which is exactly the product-of-residual-norms expression appearing in the lemma.

To interpret Lemma~\ref{lem:strong-ident-intro}, note that if $(x,y)$ has objective value $O(\OPT)$ and $(u,v)$ is optimal, then $\|xy^\top - uv^\top\|_p^p \leq O(\OPT)$ by triangle inequality.
Furthermore, the solution $(cu, v/c)$ has the same objective value as $(u,v)$.
So at a high level, Lemma~\ref{lem:strong-ident-intro} tells us there exists an optimal solution $(cu,v/c)$ that is close to $(x,y)$ in $\ell_p$ norm.

The quantitative bound on the distance between $(cu,v/c)$ and $(x,y)$ implied by Lemma~\ref{lem:strong-ident-intro} is stronger than we might expect in a not-obvious way.
A first (but incorrect) interpretation of Lemma~\ref{lem:strong-ident-intro} might be that all of $\|x + cu\|_p^p, \|y - v/c\|_p^p,\|x-cu\|_p^p,\|y + v/c\|_p^p$ are small, perhaps of order $\sqrt{\OPT}$.
But this is actually impossible (assuming the optimal error is sufficiently small) -- $\|x+cu\|_p^p$ and $\|x -cu\|_p^p$ cannot simultaneously be very small, and the same goes for $\|y -v/c\|_p^p,\|y+v/c\|_p^p$.
WLOG, imagine that $\|x+cu\|$ and $\|y+v/c\|$ are the ``large'' ones, much larger than $\sqrt{\OPT}$.
Then $\|y-v/c\|_p^p$ and $\|x - cu\|_p^p$ must be much less than $\sqrt{\OPT}$.
This improvement over the naive interpretation of Lemma~\ref{lem:strong-ident-intro} turns out to be crucial to our PTAS.

With Lemma~\ref{lem:strong-ident-intro} in hand, we can add constraints to our convex program ensuring the pseudodistribution is only supported over $u,v$ such that the norms of the residual vectors, $\|u-x\|_4^4$ and $\|v-y\|_4^4$ are much smaller than $\|u\|_4^4, \|v\|_4^4$, scaling (roughly speaking) with $\OPT$ instead of with $\|A\|_4^4$.
Then, at a high level, we repeat the additive approximation algorithm above, except that when we analyze the objective value loss due to rounding, we expand $\|A - uv^\top\|_4^4$ as a polynomial in the residual variables $u_i - x_i, v_j - y_j$.
We then apply pinning lemmas to these residual variables rather than $u_i,v_j$.

\paragraph{The Rank-$k$ Case.}
Most of the ideas from the rank-one case extend immediately to rank $k$.
We can still set up a convex program whose solutions are $\ell$-wise local distributions, now over $(\R^k)^\ell$ rather than $\R^\ell$, and condition on the values of individual coordinates $u_i,v_j \in \R^k$.
The main new issue in rank $k$ is the non-uniqueness of the factorization.
In rank one, the only ambiguity is scalar rescaling: $uv^\top = (cu)(v/c)^\top$.
For rank $k$, there is a whole change-of-basis symmetry: $UV^\top = (UM)(VM^{-\top})^\top$ for every invertible $M \in \R^{k \times k}$.
Thus, we first have to choose the right coordinates in the latent $k$-dimensional space.
The higher-rank identifiability lemma below does exactly this.
It finds a change of basis under which the two natural ``signed'' residual products are small, generalizing the rank-one identity with the scalar $c$.

\begin{lemma}[Nearby Optimal Solutions for Rank-$k$ Approximation]
\label{lem:high-rank-ident-intro}
    Let $X,U \in \mathbb{R}^{n \times k}$ and $Y,V \in \mathbb{R}^{m \times k}$ be full column rank matrices. Then there exists an invertible matrix $M \in \mathbb{R}^{k \times k}$ such that
    \[ \norm{(U + XM)(V - Y\left(M^{-1}\right)^{\top})^\top}_p^p + \norm{(U - XM)(V + Y\left(M^{-1}\right)^\top)^\top}_p^p \leq k^{O_p(1)} \norm{UV^\top - XY^\top}_p^p\,.\]
    Furthermore, if $X$, $U$, $\left(X^\top X\right)^{-1/2}$, and $\left(U^\top U\right)^{-1/2}$ have entries bounded in magnitude by $\beta$, then $M$ and $M^{-1}$ have entries bounded by $\poly(\beta \cdot nmk)$.
\end{lemma}

Similarly to the rank $1$ case, the change-of-basis symmetry lets us express an optimal rank-$k$ matrix in the latent coordinates chosen by~\cref{lem:high-rank-ident-intro} without changing its objective value.
Thus we can restrict our search to $U,V$ that are ``close'' to $X,Y$ in the sense of~\cref{lem:high-rank-ident-intro} without any loss in objective value.
However, unlike in~\cref{lem:strong-ident-intro}, \cref{lem:high-rank-ident-intro} only bounds the norm of the product matrix of the residuals (rather than the product of the norms).
This does not immediately show that we can reduce our search to a small $\ell_p$ ball around $(X, Y)$, and it is not a priori clear whether~\cref{lem:high-rank-ident-intro} is enough to analyze the rounding error on the rank $k$ objective.

However, the following \emph{residual balancing} lemma turns out to be sufficient.
Intuitively, it says that there exists some affine transformation of a small $\ell_p$ ball such that there is an optimal solution in the transformed ball centered at $(X, Y)$.
The reason such a lemma is needed is that product control is weaker than the moment control required by pinning.
Abstractly, even if $FG^\top$ is small, neither $F$ nor $G$ has to be small in the original coordinates: for any invertible change of basis $N$, we can replace $F$ by $FN$ and $G$ by $G N^{-\top}$ without changing the product.
Residual balancing chooses this change of basis so that the two factors are small at the same time.

Auerbach's lemma gives exactly the finite-dimensional normed-space tool needed to choose such a basis.
It says that in every $k$-dimensional normed space there is a basis $b_1,\ldots,b_k$ and a dual basis $b_1^*,\ldots,b_k^*$ such that $\|b_a\| = 1$, $\|b_a^*\|_* = 1$, and $b_a^*(b_b) = \mathbf{1}\{a=b\}$.
In other words, one can choose coordinates so that each coordinate direction has unit norm and each coordinate can be read off by a norm-one linear functional.
We apply this to the norm $\|z\|_F \vcentcolon= \|Fz\|_p$ on $\R^k$.
This produces a basis of $\R^k$ in which each basis vector has unit $\|\cdot\|_F$ norm, and the associated coordinate functionals have dual norm one.
If $N$ is the matrix whose columns are this basis, then $FN$ is normalized, and the dual condition says that $FN$ cannot hide large coordinates: the coordinates of any vector $z$ are controlled by $\|FNz\|_p$.
Applying this to the rows of $G N^{-\top}$ shows that smallness of $FG^\top = FN(GN^{-\top})^\top$ forces $G N^{-\top}$ to be small as well, up to factors depending only on $k$.
Finally, a scalar rescaling of $N$ balances the two sides, giving both $\|FN\|_p^p$ and $\|G N^{-\top}\|_p^p$ bounded by roughly $\|FG^\top\|_p^{p/2}$.
We apply this balancing twice: once to the product $(X+U)(Y-V)^\top$, producing $N_1$, and once to $(X-U)(Y+V)^\top$, producing $N_2$.

Thus, the product-residual control from higher-rank identifiability becomes individual moment control for the transformed residuals, yielding the residual balancing lemma below.
Those are exactly the moment bounds needed when the objective is expanded, and the heavy-tailed pinning lemma is applied.

\begin{lemma}[Residual Balancing (Informal)]
\label{lem:residual-balancing}
    Let $A \in \mathbb{R}^{n \times m}$ be a matrix and let $X \in \mathbb{R}^{n \times k},Y \in \mathbb{R}^{m \times k}$ be matrices such that $\norm{A- XY^\top}_p^p \leq \OPT$. Let $B$ be a rank $k$ matrix such that $\norm{A - B}_p^p \leq \OPT$. Then there exists a factorization $U \in \mathbb{R}^{n \times k}$, $V \in \mathbb{R}^{m \times k}$ such that $B = UV^\top$ and invertible matrices $N_1, N_2 \in \mathbb{R}^{k \times k}$ such that
    \begin{align*}
        \norm{(X+U)N_1^\top}_p^p &\leq O_{p,k}(1) \cdot \sqrt{OPT}  &\norm{(Y-V)N_1^{-1}}_p^p &\leq O_{p,k}(1) \cdot \sqrt{OPT} \\
        \norm{(X-U)N_2^\top}_p^p &\leq O_{p,k}(1) \cdot \sqrt{OPT} &\norm{(Y+V)N_2^{-1}}_p^p &\leq O_{p,k}(1) \cdot \sqrt{OPT}\,.
    \end{align*}
\end{lemma}

There are still several issues left to resolve; for example, we are not handed such $N_1$ and $N_2$ that balance the moments of the residuals, and a priori, such transformations might not even have bounded bit complexity. Without such $N_1$ and $N_2$ in hand, we cannot enforce the same constraints we did in the rank-$1$ case, which required that the pseudo-distribution only be supported on solutions within a ball around $(X, Y)$. But it turns out that we can bound the entries of $N_1$ and $N_2$ and thus it suffices to brute force search over a sufficiently fine net since they only have $k^2$ entries\footnote{In the technical sections of the paper, we execute this via including variables for these transformations and conditioning on their values, but brute force search also suffices.}. Once given $N_1$ and $N_2$ such that the inequalities in~\cref{lem:residual-balancing} are feasible for an (approximately) optimal solution, we can simply enforce the moment bounds on our pseudo-distribution solutions. We defer a detailed discussion of these issues to~\cref{sec:higher-rank-lra}.

\subsection{Heavy Tailed Pinning Lemma}
\label{sec:overview-heavy-tailed}

We now explain the proof of the heavy-tailed pinning lemma, and
for the sake of exposition, we assume we are working with random
variables $X_1,\ldots,X_n$ with bounded $4/3$-rd moments, and very large variance. In particular, we show that for any sufficiently large 
$t \in \mathbb{N}$, there exists a random subset of coordinates
 $S \subseteq [n]$ of size $s\leq t$ such that the average covariance of the random variables 
conditioned on $X_S$ is small:
\begin{equation*}
  \E_S \E_{i,j \sim [n]} |\Cov(X_i,X_j \mid X_S)|
  \leq O(t^{-1/6}) \cdot M^{3/4},
\end{equation*}
where 
\begin{equation*}
  M = 
  \E_{i,j \sim [n]} \E |X_iX_j|^{4/3} = \E \left(\E_{i \sim [n]} |X_i|^{4/3}\right)^2\,.
\end{equation*}
Equivalently, conditioning on $t = O(\eps^{-6})$ coordinates makes the average covariance at most $\eps\; M^{3/4}$. In our low-rank approximation and matrix norm applications, we would like to apply such a statement to variables like $u_i^3$, 
since we have bounds on $\sum_i \E |u_i^3|^{4/3}$ but no bounds on $\sum_i \E |u_i^3|^2$. 

Our starting point to prove the heavy-tailed pinning lemma is the usual variance pinning lemma, which states that if $X_1, \ldots , X_n$
had bounded variance, then after $t$ rounds of conditioning, the average covariance would be at most $O(t^{-1/2} \cdot \E_i \Var(X_i)$. 
The key insight is that we can apply the variance pinning lemma to a truncated version of the random variables. More formally,
fix a threshold $L$ and decompose 
\begin{equation*}
  X_i = Y_i + R_i, \qquad Y_i = \sgn(X_i)\min\{|X_i|,L\}, \qquad R_i = X_i-Y_i.
\end{equation*}
Using that $\min\Set{a,b}^2 \leq a^{2-\zeta} b^{\zeta}$ for any non-negative $a,b$ and $\zeta \in [0,2]$, 
we can bound the second moment of the clipped variable $Y_i$ in terms of the $4/3$-rd moment of $X_i$:
\begin{equation*}
  Y_i^2  \leq L^{2/3} \cdot |X_i|^{4/3}\,, \quad \text{and hence} \quad \E_i \Var(Y_i) \leq L^{2/3} \E_i \E |X_i|^4/3\leq L^{2/3} M^{1/2}\,.
\end{equation*}
Therefore, applying the variance pinning lemma to $Y_i$, we have 
\begin{equation}
\label{eqn:intro-truncated-covariance}
  \E_S \; \E_{i,j} |\Cov(Y_i, Y_j \mid X_S)| \leq O(t^{-1/2}) \cdot L^{2/3} M^{1/2}\,.
\end{equation}
This enables us to control the average covariance of the truncated random variables but we still have to
account for the tails. Using the bi-linearity of covariance, we have the following decomposition:
\begin{equation}
\label{eqn:intro-covariance-decomposition}
\Cov(X_i,X_j \mid X_S)
  = \underbrace{ \Cov(Y_i,Y_j \mid X_S) }_{\eqref{eqn:intro-covariance-decomposition}.(1)} + \underbrace{ \Cov(Y_i,R_j \mid X_S) }_{\eqref{eqn:intro-covariance-decomposition}.(2)} + \underbrace{ \Cov(R_i,Y_j \mid X_S) }_{\eqref{eqn:intro-covariance-decomposition}.(3)} + \underbrace{ \Cov(R_i,R_j \mid X_S) }_{\eqref{eqn:intro-covariance-decomposition}.(4)}.
\end{equation}
We handled term \eqref{eqn:intro-covariance-decomposition}.(1) in \cref{eqn:intro-truncated-covariance} using the variance pinning lemma. 
The remaining terms are controlled by showing that, on average, the residuals $R_i$ have small mass. Let $U = \E_i |R_i|$ and let $V = \E_i |Y_i|$. 
Observe, 
\begin{equation*}
  |R_i| \leq  |X_i| \cdot \mathbf{1}\{|X_i| > L\} \leq |X_i| \cdot \Paren{  \frac{|X_i|}{L}}^{1/3} = L^{-1/3} |X_i|^{4/3} \quad \text{and hence} \quad U \leq L^{-1/3} \E_i |X_i|^{4/3}\,.
\end{equation*}
and therefore $\E U^2 \leq L^{-2/3} M$. Further, since $|Y_i|\leq |X_i|$, we have $V \leq \E_i |X_i| \leq \Paren{\E_i |X_i|^{4/3}}^{3/4}$. Therefore, 
\begin{equation*}
  \E V^2 \leq \E \Paren{\E_i |X_i|^{4/3}}^{3/2} \leq  \Paren{\Paren{\E \E_i |X_i|^{4/3} }^2 }^{3/4}  = M^{3/4}\,.
\end{equation*}
These estimates suffice to bound the remaining terms: 
\begin{equation*}
  \begin{split}
  \eqref{eqn:intro-covariance-decomposition}.(2) \leq E_{i,j}\E[|Y_iR_j| \mid X_S] +
  \E_i\E[|Y_i| \mid X_S] \, \E_j\E[|R_j| \mid X_S] = \E[UV \mid X_S] + \E[V \mid X_S]\E[U \mid X_S]\,.
  \end{split}
\end{equation*}
Therefore, after taking
expectations and applying Cauchy-Schwarz and Jensen's inequality,
\begin{equation*}
  \E_S \E_{i,j} |\eqref{eqn:intro-covariance-decomposition}.(2)| \leq \E[UV] +
  \E_S[ \E[U \mid X_S]\E[V \mid X_S] ]  \leq O(1)\cdot L^{-1/3}M^{7/8}\,.
\end{equation*}
Term \eqref{eqn:intro-covariance-decomposition}.(3) is symmetric and can be handled in the exact same way. 
Further, using a similar argument, we have 
\begin{equation*}
  \begin{split}
  \E_{i,j} |\eqref{eqn:intro-covariance-decomposition}.(4)| & \leq \E_{i,j} \E[ |R_i R_j| \mid X_S ] + \E_i [ E[|R_i| \mid X_S] ] \E_j [ E[|R_j| \mid X_S] ] = \E[U^2 \mid X_S] + \E[U \mid X_S]^2\,.
  \end{split}
\end{equation*}
Taking expectations and using Jensen's inequality gives
\begin{equation*}
  \E_S \E_{i,j}|\Cov(R_i,R_j \mid X_S)| \leq O(1)\cdot \E U^2 \leq O(1)\cdot L^{-2/3}M.
\end{equation*}
Plugging all the estimates back into \cref{eqn:intro-covariance-decomposition},
we have 
\begin{equation*}
  \E_S \E_{i,j} |\Cov(X_i,X_j \mid X_S)| \leq O\Big(
    t^{-1/2}L^{2/3}M^{1/2} \
    + L^{-1/3}M^{7/8} \
    + L^{-2/3}M
  \Big).
\end{equation*}
Choosing $L = M^{3/8} t^{1/2}$ balances the first two terms, which dominate the third, and yields the desired bound.

\paragraph{AI Disclosure.}
Our results for the PTAS for entry-wise rank-one approximation, as well as our additive approximation for the $p \rightarrow q$ norm, were produced with minimal AI assistance. Given a complete manuscript for the rank-one approximation, GPT 5.4 Pro formulated and proved the higher-rank extensions of the matrix norm structure statements. These were verified by the authors and used to generate the proof of the rank-$k$ low-rank approximation result. GPT 5.4 Pro also simplified the author's initial proof of the heavy-tailed pinning lemma. The simplified proof is presented in~\cref{sec:overview-heavy-tailed}, while~\cref{sec:potential-fn} contains the original proof (we also give an overview of the original proof in~\cref{sec:mirror-descent-overview}). GPT 5.4 Pro and GPT 5.6 Sol were also used to audit the manuscript for correctness and fix a minor issue in an earlier version of~\cref{sec:potential-fn}.

\section{Preliminaries}

Throughout the following sections, we follow the convention that all vectors are column vectors.

\subsection{Sum-of-Squares Background}

We now provide an overview of the sum-of-squares proof system, taken from the corresponding section in~\cite{10.1145/3717823.3718218}.
We closely follow the exposition as it appears in the lecture notes of Barak~\cite{barak2016proofs}.   

\paragraph{Pseudo-Distributions.}
A discrete probability distribution over $\R^m$ is defined by its probability mass function, $D\from \R^m \to \R$, which must satisfy $\sum_{x \in \mathrm{supp}(D)} D(x) = 1$ and $D \geq 0$.
We extend this definition by relaxing the non-negativity constraint to merely requiring that $D$ passes certain low-degree non-negativity tests.
We call the resulting object a pseudo-distribution.

\begin{definition}[Pseudo-distribution]
A \emph{degree-$\ell$ pseudo-distribution} is a finitely-supported function $D:\R^m \rightarrow \R$ such that $\sum_{x} D(x) = 1$ and $\sum_{x} D(x) p(x)^2 \geq 0$ for every polynomial $p$ of degree at most $\ell/2$, where the summation is over all $x$ in the support of $D$.
\end{definition}
Next, we define the related notion of pseudo-expectation.
\begin{definition}[Pseudo-expectation]
The \emph{pseudo-expectation} of a function $f$ on $\R^m$ with respect to a pseudo-distribution $\mu$, denoted by $\pexpecf{\mu(x)}{f(x)}$,  is defined as
\begin{equation*}
    \pexpecf{\mu(x)}{f(x)} = \sum_{x} \mu(x) f(x).
\end{equation*}
\end{definition}
We use the notation $\pexpecf{\mu(x)}{(1,x_1, x_2,\ldots, x_m)^{\otimes \ell}}$ to denote the degree-$\ell$ moment tensor of the pseudo-distribution $\mu$.
In particular, each entry in the moment tensor corresponds to the pseudo-expectation of a monomial of degree at most $\ell$ in $x$. 

\begin{definition}[Constrained pseudo-distributions]
\label{def:constrained-pseudo-distributions}
Let $\calA = \Set{ p_1\geq 0 , p_2\geq0 , \dots, p_r\geq 0}$ be a system of $r$ polynomial inequality constraints of degree at most $d$ in $m$ variables.
Let $\mu$ be a degree-$\ell$ pseudo-distribution over $\mathbb{R}^m$.
We say that $\mu$ \emph{satisfies} $\calA$ at degree $\ell \ge1$ if for every subset $\calS \subset [r]$ and every sum-of-squares polynomial $q$ such that $\deg(q) + \sum_{i \in \calS } \max\Paren{ \deg(p_i), d} \leq \ell$, $\pexpecf{\mu}{ q \prod_{i \in \calS} p_i } \geq 0$.
Further, we say that $\mu$ \emph{approximately satisfies} the system of constraints $\calA$ if the above inequalities are satisfied up to additive error $\pexpecf{\mu}{ q \prod_{i \in \calS} p_i } \geq -2^{-n^{\ell} } \norm{q} \prod_{i \in \calS} \norm{p_i}$, where $\norm{\cdot}$ denotes the Euclidean norm of the coefficients of the polynomial, represented in the monomial basis.  
\end{definition}

Crucially, there's an efficient separation oracle for moment tensors of constrained pseudo-distributions. 
Below gives the unconstrained statement; the constraint statement follows analogously.

\begin{fact}[\cite{shor1987approach, nesterov2000squared, parrilo2000structured, grigoriev2001complexity}]
    \label{fact:sos-separation-efficient}
    For any $m,\ell \in \N$, the following convex set has a $m^{\bigO{\ell}}$-time weak separation oracle, in the sense of \cite{grotschel1981ellipsoid}:\footnote{
        A separation oracle of a convex set $S \subset \R^M$ is an algorithm that can decide whether a vector $v \in \R^M$ is in the set, and if not, provide a hyperplane between $v$ and $S$.
        Roughly, a weak separation oracle is a separation oracle that allows for some $\eta$ slack in this decision.
    }:
    \begin{equation*}
        \Set{  \pexpecf{\mu(x)} { (1,x_1, x_2, \ldots, x_m)^{\otimes \ell } } \Big\vert \text{ $\mu$ is a degree-$\ell$ pseudo-distribution over $\R^m$}}
    \end{equation*}
\end{fact}
This fact, together with the equivalence of weak separation and optimization \cite{grotschel1981ellipsoid} forms the basis of the sum-of-squares algorithm, as it allows us to efficiently approximately optimize over pseudo-distributions. 

Given a system of polynomial constraints, denoted by $ \calA$, we say that it is \emph{explicitly bounded} if it contains a constraint of the form $\{ \|x\|^2 \leq 1\}$. Then, the following fact follows from  \cref{fact:sos-separation-efficient} and \cite{grotschel1981ellipsoid}:

\begin{theorem}[Efficient optimization over pseudo-distributions]
    \label{fact:eff-pseudo-distribution}
There exists an $(m+r)^{O(\ell)} $-time algorithm that, given any explicitly bounded and satisfiable system $ \calA$ of $r$ polynomial constraints in $m$ variables, outputs a degree-$\ell$ pseudo-distribution that satisfies $ \calA$ approximately, in the sense of~\cref{def:constrained-pseudo-distributions}.\footnote{
    Here, we assume that the bit complexity of the constraints in $ \calA$ is $(m)^{O(1)}$.
}
\end{theorem}

We now state the following standard fact for pseudo-distributions, which extends a fact that holds for standard probability distributions.
These can be found in the prior works listed above.

\begin{fact}[Hölder's inequality for pseudo-distributions] \label{fact:pseudo-expectation-holder}
Let $f,g$ be polynomials of degree at most $d$ in the variables $x \in \R^m$. 
Fix $t \in \N$. Then, for any degree-$dt$ pseudo-distribution $\mu$,
\begin{equation*}
    \pexpecf{\mu}{ f^{t-1}  g} \leq \Paren{ \pexpecf{\mu}{ f^t }  }^{\frac{t-1}{t}} \cdot  \Paren{  \pexpecf{\mu }{ g^t }  }^{\frac{1}{t}}.  
\end{equation*}
In particular, when $t$ is even,
$\pexpecf{\mu}{f}^t \leq \pexpecf{\mu}{ f^t }$.
\end{fact}

\subsection{Potential Aligned Functions}

In this subsection, we restate some definitions and basic facts about potential functions used in the analysis of global correlation rounding, taken from Section 4 in~\cite{anderson2025additiveapproximationschemeslowdimensional}. 

\begin{definition}[Pseudo-distribution potentials]
\label{def:pseudo-distribution-potential}
    A function $\Phi$ from pseudo-distributions over $\mathbb{R}^n$ to $\mathbb{R}$ is a pseudo-distribution potential if
    \begin{enumerate}
        \item For all pseudo-distributions $\mu$ we have that $\Phi(\mu) \geq 0$.
        \item For a fixed $i \in [n]$, let $\mu'$ be the pseudo-distribution given by sampling a random value $\hat{x}_i$ from the marginal of $x_i$ and conditioning on the event $\Set{ x_i = \hat{x}_i}$. Then,
        \[ \E_{x_i} \left[\Phi(\mu) - \Phi(\mu')\right] \geq 0\,.\]
    \end{enumerate}
\end{definition}

The typical potential functions used in global correlation rounding by \cite{Raghavendra2012,Barak2011} are $\E_{i \sim [n]} H(\{x_i\})$ and $\E_{i \sim [n]} \Var(x_i)$.
Both are pseudo-distribution potentials in the sense of Definition~\ref {def:pseudo-distribution-potential}.

\begin{definition}[Potential aligned functions]
\label{def:potential-aligned}
    A function $f: (\mathbb{R}^k)^n \rightarrow \mathbb{R}$ is $(\delta, \delta')$-potential aligned with a pseudo-distribution potential $\Phi$, initial pseudo-distribution $\mu_0$, and family of polynomials $p_1,\ldots,p_m$ which have $\pE_\mu p_i = 0$ if for all distributions over pseudo-distributions $\mu$ such that 
    \begin{enumerate}
        \item $\E_{\{\mu\}} \mu = \mu_0$, 
        \item all $\mu'$ in the support of $\{\mu\}$ satisfy $\pE_{\mu'} p_i = 0$ for all $i \leq m$, and 
        \item $ \left\vert \E_{\{\mu\}} \pE_{\mu} f(x_1,\ldots, x_n) - \pE_{\mu^{\otimes}} f(x_1,\ldots, x_n) \right\vert \geq \delta\,,$
    \end{enumerate}
    we have that there exists an index $i \in [n]$ such that
    \[ \E_{ \{ \mu \}} \E_{\hat{x_i} \sim \{x_i \} } \left[\Phi(\mu) - \Phi(\mu') \right] \geq \delta'\,,\]
    where $\mu'$ is the pseudo-distribution given by sampling a random value $\hat{x}_i$ from the marginal of $x_i$ and conditioning on $x_i = \hat{x}_i$.
\end{definition}

In this language, prior works like \cite{Barak2011, Raghavendra2012} can be interpreted as showing that any  function $f$ of the form $\E_{i,j \sim [n]} f_{ij}(x_i,x_j)$ where $f_{ij}$'s are $[-1,1]$-valued is potential-aligned with respect to $\Phi = \E_{i} H(\{x_i\})$. 
And, functions of the form $f = \E_{i,j \sim [n]} A_{ij} \langle x_i, x_j\rangle$ are potential-aligned with respect to $\Phi = \E_{i \sim [n]} \Var(x_i)$; this follows almost directly from prior work on global correlation rounding (see~\cite{Barak2011}). 

We can also combine multiple potential functions to analyze correlation rounding on a sum of objective functions.

\begin{lemma}[Handling several potential aligned functions (see Lemma 4.4 in~\cite{anderson2025additiveapproximationschemeslowdimensional})]
\label{lem:gcr_linear_combo}
    Let $\epsilon \geq 0$. Let $\{f_{j}\}_{j \in [\ell]}$ be a collection of functions from $(\mathbb{R}^k)^n$ such that $f_j$ is $(\epsilon/\ell, \delta_j)$-potential aligned for some pseudo-distribution potential $\Phi_j$ and some initial pseudo-distribution $\mu$. Then there exists a set $\calT \subseteq [n]$ of size at most $\sum_{j \in [\ell]} \Phi_j(\mu)/\delta_j$ such that after conditioning on the values of $x_i \in \calT$,
    \[ \left\vert \E_{\{\mu_\calT\}} \pE_{\mu_{\calT}} \sum_{j \in [\ell]} f_j(x_1,\ldots, x_n) - \pE_{\mu^{\otimes}_{\calT}} \sum_{j \in [\ell]} f_j(x_1,\ldots, x_n) \right\vert \leq \epsilon \,.\]
\end{lemma}

\subsection{Miscellaneous}

\begin{fact}[Young’s Inequality]
\label{fact:young}
    For $a,b \ge 0$ and $p,q > 1$ with $\frac{1}{p} + \frac{1}{q} = 1$, we have that
    \[ ab \le \frac{a^{p}}{p} + \frac{b^{q}}{q}\,.\]
\end{fact}

\section{Designing Potential Functions for Global Correlation Rounding}
\label{sec:potential-fn}

In this section, we will give a framework for analyzing rounding error, which can be used to prove the \emph{heavy-tailed} pinning lemma. The proofs in this section utilize a connection to mirror descent and a unifying framework for proving pinning lemmas. An overview of our approach can be found in~\cref{sec:mirror-descent-overview}. 

The key challenge in our setting is to quantify the decrease in the uncertainty of a random variable $u$ when conditioning on some other random variable $v$. In prior work, uncertainty has often been measured by either entropy or variance of the associated random variable. We give extensions of these potential decrease inequalities based on a general framework inspired by mirror descent.

In all of our applications, our potential functions will be defined using Bregman divergences. 

\begin{definition}[Bregman Divergence]
    For a differentiable function $f$, we have that the Bregman divergence associated with $f$ is
    \[ B_f(x, y) = f(x) - f(y) - \langle \nabla f(y), x - y \rangle\,.\]
\end{definition}

At a high level, our proof of both of our new potential decrease lemmas will revolve around using the decrease in Bregman divergence given by mirror descent to bound the potential decrease when there are large correlations among variables.

\subsection{Low-moment potential decay based on covariance}
One common approach to understanding how uncertainty decreases when conditioning uses the fact that if $u$ and $v$ have large covariance and $v$ has small variance, then the variance of $u$ significantly drops when conditioning on $v$. We will generalize this approach to work in settings where the variance of the associated random variable might be very large, but lower moments are still bounded. Thus, it is necessary to understand how the $q$ moments (or related functions of them) for $q \in (1,2)$ behave when conditioning on other random variables. We give a bound that characterizes this and generalizes the following fact, which has been used to analyze GCR rounding schemes in prior work:
\begin{align}
    \Var(u) - \E_v \Var(u \vert v) \geq \frac{\Cov^2(u, v)}{\Var(v)}\,. \label{eq:var-decrease-bound}
\end{align} 
Instead of working with variance as our measure of uncertainty, we will instead work with a more general measure of uncertainty. These will be defined using Bregman divergences associated with polynomial functions $f$. Note that for polynomial functions of degree strictly less than $q$, the associated Bregman divergences only depend on smaller moments of the distribution than the $q$th-moment.
\begin{definition}[Bregman Uncertainty]
\label{def:breg-uncertainty}
    Let $p \in \mathbb{N}^{> 0}$ and let $B_p$ be the Bregman divergence associated with the function $f(x) = \frac{p-1} p \vert x\vert^{\frac p {p-1}}$. Then the $p$-Bregman uncertainty $\phi^{(p)}_u$ is 
    \[\phi^{(p)}_u = \min_c \E_u B_p(u, c)\,.\]
\end{definition}

Note that the following fact gives us an explicit formula for $\phi^p_u$:
\begin{fact}
\label{fact:bregman_min_expecation}
    For $p \in \mathbb{N}^{>0}$ and any random variable $x$, the minimizer over all \emph{constants} $c$ of the function $c \mapsto \E_x B_p(x,c)$ is $c = \E x$, and the minimum value is $\frac{p-1} p \left[\E \vert x\vert^{\frac p {p-1}} - \left\vert\E x\right\vert^{\frac p {p-1}}\right]$.
\end{fact}

This also implies the following easy bound on the potential:
\begin{fact}
\label{fact:gen-initial-potential-bound}
    For $p \in \mathbb{N}^{>0}$ and any random variable $x$ we have that 
    \[ \phi^{(p)}_x \leq O\left(\E \vert x\vert^{\frac p {p-1}}\right)\,.\]
\end{fact}

Finally, we will need the following equality in the analysis of potential decreases:
\begin{fact}[Bregman Divergence Three Point Equality]
\label{fact:bregman_three_point}
    For any Bregman Divergence $B_h$ we have that
    \[ B_h(x, y) + B_h(y, z) = B_h(x, z) + \langle \nabla h (z) - \nabla h (y), x-y \rangle\,. \]
\end{fact}

\subsection{Bregman uncertainty generalizes variance}

Note that when $p=2$ the associated Bregman divergence is simply $B_2(x,y) = \frac 1 2 (x-y)^2$ and since the minimizing constant $c$ is the expectation we have that $\phi^{(2)}_u = \frac 1 2 \Var(u)$. Furthermore, the guarantee of~\cref{eq:var-decrease-bound} (up to constant factors) can be recovered using mirror descent progress lemmas as follows:

\begin{theorem}[Conditioning reduces Variance]
\label{thm:variance-reduction-via-mirror}
Given scalar random variables $u, v$ drawn from the joint distribution $\calD$, the decrease in variance of $u$ when conditioning on $v$ can be lower bounded as follows:
\begin{equation*}
    \Var(u) - \E_v \Var(u \vert v) \geq \frac{\Cov^2(u, v)}{4\cdot \Var(v)}\,.
\end{equation*}
\end{theorem}

In the proof of~\cref{thm:variance-reduction-via-mirror}, we will need the following standard mirror descent progress lemma:
\begin{lemma}[Mirror Descent Progress Lemma]
\label{lem:mirror-descent-progress}
    Let $f: \mathbb{R}^d \to \mathbb{R}$ be a differentiable function that is $1$-strongly convex with respect to a norm $\|\cdot\|$ and let $B$ be the Bregman divergence induced by $f$. Let $x, x' \in \mathbb{R}^d$ be such that
    \[x' = \arg\min_{z \in \mathbb{R}^d} \left\{ \eta \langle y, z \rangle + B(z, x) \right\},\]
    for some cost vector $y \in \mathbb{R}^d$ and step size $\eta > 0$. Then for any $u \in \mathbb{R}^d$, the following inequality holds:
    \[ \eta \langle y, u - x \rangle \leq B(u, x) - B(u, x') + \eta^2 \|y\|_*^2\,,\]
    where $\|\cdot\|_*$ denotes the dual norm of $\|\cdot\|$. 
\end{lemma}

\begin{proof}[Proof of~\cref{thm:variance-reduction-via-mirror}]
Using \cref{def:breg-uncertainty}, we note that for any $c$ which is just a function of $v$ and the joint distribution of $u, v$ we have that
\[ \Var(u) - \E_v \Var(u\vert v) = 2 \left( \phi^{(2)}_u - \E_v \phi^{(2)}_{u \vert v}\right) \geq \E_u B_2(u, \E u) - \E_v \E_{u \vert v} B_2(u, c)\,,\]
using the fact that the divergence is minimized at $c = \E_u$. Thus, it suffices to exhibit a $c$ which causes a sufficiently large decrease in the Bregman divergence. We can compute such a $c$ by performing a mirror descent update step on $x= \E u$ using the cost vector $y= v - \E v$. Note that in our setting $f(x) = \frac 1 2 x^2$, which is $1$-strongly convex with respect to the $\ell_2$ norm and is self-dual. Thus, we can apply~\cref{lem:mirror-descent-progress} with $x = \E u$ and $y = v - \E v$ to get that
\[ \eta \langle v - \E v, u - \E u \rangle \leq B_2(u, \E u) - B_2(u, x') + \eta^2 \norm{v - \E v}_2^2\,.\]
When we take an expectation over both $u, v$ on both sides, we get that this inequality gives us that
\begin{align*}
    \eta \Cov(u, v) = \E_{u,v} \eta \langle v - \E v, u - \E u \rangle &\leq \E_u B_2(u, \E u) - \E_v \E_{u \vert v} B_2(u, x') + \eta^2 \E_v \norm{v - \E v}_2^2 \\
    &= \E_u B_2(u, \E u) - \E_v \E_{u \vert v} B_2(u, x') + \eta^2 \Var(v)\,.
\end{align*}

Now setting $\eta = \frac{\Cov(u, v)}{2 \Var(v)}$ and rearranging, we have 
\begin{equation*}
    \begin{split}
    \Paren{\frac{1}{2} - \frac{1}{4}} \frac{\Cov^2(u, v)}{ \Var(v)}  & \leq \E_u B_2(u, \E u) - \E_v \E_{u \vert v} B_2(u, x') \\
    &\leq \Var(u) - \E_v \Var(u \vert v) 
    \end{split}
\end{equation*}
which yields the desired variance reduction inequality.
\end{proof}

Thus, Bregman uncertainty generalizes variance as a metric of uncertainty. Furthermore, the associated mirror descent progress bounds and analysis can be used to understand how Bregman uncertainty decreases when conditioning on other random variables. 

\subsection{Fractional Bregman uncertainty decrease}

We now extend the proof of the case where $p$ is an even integer larger than $2$. Specifically, we prove the following bound on the decrease in Bregman uncertainty:
\begin{lemma}
\label{lem:general-bregman-p-uncertainty-decrease}
    Let $p$ be an even integer and let $B_p$ be the Bregman divergence associated with $\frac {p-1} p \vert x\vert^{\frac p {p-1}}$. Let $y,z$ be jointly distributed random variables over $\mathbb{R}$. Then we have that 
    \[ \E_{y,z} \left[B(y, \E y) - B(y, \E [y \vert z])\right] \geq \Omega \left(\frac{\Cov^{2(p-1)} (y,z)}{\left( \E \vert z\vert^{\frac p {p-1}} \right)\left(\E_{y,z} \left[\vert y\vert^{\frac p {p-1}} \vert z\vert^{\frac p {p-1}}\right] + \E_y \left[\vert y\vert^{\frac p {p-1}}\right] \E_z \left[\vert z\vert^\frac p {p-1}\right]\right)^{\frac{(p-2)(2p-1)}{p}}}\right) \,.\]
\end{lemma}
Before we proceed to the proof of~\cref{lem:general-bregman-p-uncertainty-decrease}, we first prove a few helper lemmas which will be useful to us in the proof.

\begin{lemma}
\label{lem:small-fact-non-negative-rv}
    Let $c \in (0,1)$ and let $X, Y$ be non-negative jointly distributed random variables. Then we have that
    \[ \E_{X, Y} [X Y] \E_Y [Y]^{c} \leq O\left( \E_{X,Y} [X Y^{c+1}] + \E_X [X] \E_Y [Y^{c+1}]\right)\]
\end{lemma}
\begin{proof}
    Note that by H\"olders with $1+c, (1+c)/c$ we have that 
    \[ \E[XY] \E[Y]^{c} = \E\left[\left(X^{\frac 1 {1+c}} Y\right)\left(X^{\frac c {1+c}}\right)\right] \E[Y]^{c} \leq \E[XY^{1+c}]^{\frac 1 {1+c}} \E[X]^{\frac c {1+c}} \E[Y]^{c}\,.\]
    Using that $\E[XY^{1+c}]^{\frac 1 {1+c}} > 0, \E[X]^{\frac c {1+c}} \E[Y]^{c} > 0$ and~\cref{fact:young} we have that 
    \[ \E[XY] \E[Y]^{c}  \leq \E[XY^{1+c}]^{\frac 1 {1+c}} \E[X]^{\frac c {1+c}} \E[Y]^{c} \leq O\left( \E[XY^{1+c}] + \E[X] \E[Y]^{1+c}\right)\,.\]
    Applying convexity to $\E[Y]^{1+c}$ and using that $\E[X] > 0$ completes the proof.
\end{proof}

Note that~\cref{lem:small-fact-non-negative-rv} also implies the following bound:
\begin{lemma}
\label{lem:separate-v-u-decrease}
    Let $y,z$ be jointly distributed random variables and let $B$ be the Bregman divergence associated with $\frac {p-1} p \vert x\vert^{\frac p {p-1}}$. Then we have that
    \begin{align*}
        &\E_{y,z} \left[ \vert z \vert^{\frac{p(p-2)}{(p-1)^2}} (B(y, \E [y]) - B(y, \E [y \vert z])) \right] \\
        &\quad\leq O\left( \left(\E_{y,z} \left[\vert y \vert^{\frac p {p-1}} \vert z\vert^{\frac p {p-1}}\right] + \E_y \left[\vert y\vert^{\frac p {p-1}}\right] \E_z \left[\vert z\vert^{\frac p {p-1}}\right] \right)^{\frac {p-2}{p-1}} \left( \E_{y,z} (B(y, \E [y]) - B(y, \E [y \vert z])) \right)^{\frac 1 {p-1}} \right)\,.
    \end{align*}
\end{lemma}
\begin{proof}
    For every fixed $z$ we have
    \[
      \E_{y | z} B(y, \E[y]) - B(y, \E[y | z]) \geq 0.
    \]
    Hence, by H\"olders with $\frac{p-1}{p-2}, p-1$ we have that 
    \begin{align*}
        &\E_{y,z} \left[\vert z\vert^{\frac{p(p-2)}{(p-1)^2}} (B(y, \E [y]) - B(y, \E [y \vert z]))\right] \\
        &\quad = \E_{z} \left[\vert z\vert^{\frac{p(p-2)}{(p-1)^2}} \E_{y \vert z} (B(y, \E [y]) - B(y, \E [y \vert z])) \right]\\
        &\quad\leq O\left( \left(\E_{z} \left[\vert z \vert^{\frac p {p-1}} \cdot (\E_{y | z} B(y, \E [y]) - B(y, \E [y \vert z]))  \right]\right)^{\frac{p-2}{p-1}} \left( \E_{y,z} (B(y, \E [y]) - B(y, \E [y \vert z])) \right)^{\frac 1 {p-1}} \right)\,.
    \end{align*}
    Thus, it remains to bound $\E_{z} \left[\vert z\vert^{\frac p {p-1}} \cdot (\E_{y | z} B(y, \E [y]) - B(y, \E [y \vert z]))\right]$. We have that 
    \begin{align*}
        \E_{z} \left[\vert z\vert^{\frac p {p-1}} \cdot (\E_{y | z} B(y, \E [y]) - B(y, \E [y \vert z]))\right] \leq \E_{y,z} \left[\vert z\vert^{\frac p {p-1}} B(y, \E [y])\right]\,,
    \end{align*}
    since $\vert z\vert^{\frac p {p-1}}$ and $\E_{y | z} B(y, \E [y \vert z])$ are both non-negative. Expanding the definition of $B$, we have that 
    \begin{align*}
        \E_{y,z} \left[ \vert z \vert^{\frac p {p-1}} B(y, \E [y])\right] &= \E_{y,z} \left[\vert z\vert^{\frac p {p-1}} \left( \vert y\vert^{\frac p {p-1}} - \left\vert\E [y]\right\vert^{\frac p {p-1}} - \frac p {p-1} \cdot \sgn(\E [y]) \cdot \left\vert \E [y] \right\vert^{\frac 1 {p-1}} \left( y - \E [y] \right) \right) \right]\\
        &\leq \E_{y,z} [\vert y\vert^{\frac p {p-1}} \vert z\vert^{\frac p {p-1}}] + \frac p {p-1} \left( \E [ \vert y \vert]\right)^{\frac 1 {p-1}} \left(\E [\vert z \vert^{\frac p {p-1}} \cdot \vert y \vert] \right)  + \frac p {p-1} \E [\vert z\vert^{\frac p {p-1}}] \E[\vert y\vert]^{\frac p {p-1}}\,.
    \end{align*}
    Using~\cref{lem:small-fact-non-negative-rv} to bound the second term in the sum, we have that 
    \[ \E_{y,z} \left[\vert z\vert^{\frac p {p-1}} B(y, \E [y])\right] \leq O\left( \E_{y,z} [\vert y\vert^{\frac p {p-1}} \vert z \vert^{\frac p {p-1}}] + \E [\vert z \vert^{\frac p {p-1}}] \E[\vert y\vert^{\frac p {p-1}}] + \E [\vert z\vert^{\frac p {p-1}}] \E[\vert y \vert]^{\frac p {p-1}}\right)\,.\]
    Applying convexity to the last term in the bound completes the proof of the lemma.
\end{proof}

We will also need a simple bound on the divergence incurred by the mirror descent update.

\begin{lemma}
\label{lem:bregman-update-divergence-bound}
    Let $p \in \mathbb{Z}^{\geq 2}$ and let $B$ be the Bregman divergence associated with $\frac {p-1}{p}\vert x\vert^{\frac{p} {p-1}}$. Fix any $x_1 \in \R$ and $\Delta \in \R$, and define $x_2 = \Paren{\sgn(x_1) \cdot \vert x_1\vert^{\frac 1 {p-1}} + \Delta}^{p-1}$.
    Then we have that
    \[
        B(x_1, x_2) \leq O\left(\Delta^2 \vert x_1 \vert^{\frac{p-2}{p-1}} + \Delta^p\right)\,.
    \]
    In particular, setting $\Delta = \eta \cdot \sgn(z) \cdot \vert z\vert^{\frac 1 {p-1}}$ yields
    \[
        B(x_1, x_2) \leq O\left( \eta^2 \vert z\vert^{\frac 2 {p-1}} \vert x_1 \vert^{\frac{p-2}{p-1}} + \eta^p \vert z \vert^{\frac p {p-1}}\right)\,.
    \]
\end{lemma}
\begin{proof}
    Let $a = \sgn(x_1) \cdot \vert x_1\vert^{\frac 1 {p-1}}$ so that $x_1 = a^{p-1}$ and $x_2 = (a+\Delta)^{p-1}$.
    Expanding the definition of $B$ gives
    \begin{align*}
        B(x_1, x_2)
        &= \frac {p-1}{p} \left(a^{p} - (a+\Delta)^p - \frac{p}{p-1}(a+\Delta)\left(a^{p-1} - (a+\Delta)^{p-1}\right)\right) \\
        &= \frac {p-1}{p} \left(- \sum_{i=0}^{p-1} \binom{p}{i} a^i \Delta^{p-i} + \frac{p}{p-1} \left(a+ \Delta\right) \sum_{i=0}^{p-2} \binom{p-1} i a^i \Delta^{p-1-i}\right)
    \end{align*}
    Note that
    \[ - \binom{p}{p-1} a^{p-1} \Delta + \left(\frac{p}{p-1} a\right) \left(\binom{p-1} {p-2} a^{p-2} \Delta \right) \]
    is identically $0$, so we have that the only monomials with non-negative coefficients are of the form $a^i \Delta^{p-i}$ for $i \in [0, p-2]$ and
    \[ B(x_1, x_2) \leq O \left( \sum_{i=1}^{p-2} \vert a\vert^i \cdot \vert \Delta \vert^{p-i} \right)\,.\]
    By~\cref{fact:young} we have that $\vert a\vert^i \cdot \vert \Delta \vert^{p-i} \leq O\left(\Delta^2 \vert a \vert^{p-2} + \vert \Delta \vert^p\right)$ for $i \in [0,p-2]$, and therefore $B(x_1, x_2) \leq O\left(\Delta^2 a^{p-2} + \Delta^p\right)$.
    Substituting $a^{p-2} = \vert x_1\vert^{\frac{p-2}{p-1}}$ completes the proof.
\end{proof}

We now return to the proof of~\cref{lem:general-bregman-p-uncertainty-decrease}.
\begin{proof}[Proof of~\cref{lem:general-bregman-p-uncertainty-decrease}]
    WLOG let $\Cov(y,z) \geq 0$. We consider the following setup: let $x_0  = \E y$ be the predictor for $y$, minimizing $B(y, x)$. We will show that we can construct a good predictor $x(z)$ which is a function of only $z$ and the joint distribution of $y,z$ such that $B(y, x(z))$ is sufficiently smaller than $B(y, x_0)$ in expectation. Note that this means that the potential function we are considering decreases since it picks the best such function. We can describe the predictor as follows:
    \begin{enumerate}
        \item Sample $z$ and let $x_1$ = $\vert z \vert^{\frac {p-2}{p-1}} x_0 = \vert z\vert^{\frac {p-2}{p-1}} \E y$.
        \item Let $x_2$ be the result of a mirror descent update step on $x_1$ with cost function $\sgn(z) \cdot \vert z\vert^{\frac 1 {p-1}}$ and some step size $\eta$. 
        \item Output $x_2/ \vert z\vert^{\frac {p-2}{p-1}}$.
    \end{enumerate}
    We will analyze this using a mirror descent-inspired analysis on the evolution between $x_1$ and $x_2$, which can be thought of as trying to predict $z^{\frac {p-2}{p-1}} y$.
    
    Note that the update in step 2 has the following explicit form:
    \[ \left( \sgn\left( \vert z \vert^{\frac {p-2}{p-1}} \E y\right) \cdot \left\vert \vert z \vert^{\frac {p-2}{p-1}} \E y \right\vert^{\frac 1 {p-1}} + \eta \sgn(z) \cdot \vert z\vert^{\frac 1 {p-1}}\right)^{p-1}\]
    Consider~\cref{fact:bregman_three_point} applied to $\vert z\vert^{\frac {p-2}{p-1}} y, \vert z\vert^{\frac {p-2}{p-1}} \E y$, and the mirror-descent update $x_2$.
    Since $x_2$ is defined by a mirror-descent update with gradient $\eta \cdot \sgn(z) \cdot \vert z\vert^{\frac 1 {p-1}}$, we have that
    \[ B(\vert z\vert^{\frac {p-2}{p-1}} y, \vert z\vert^{\frac {p-2}{p-1}} \E y) - B(\vert z\vert^{\frac {p-2}{p-1}} y, x_2) = \langle \eta \cdot \sgn(z) \cdot \vert z\vert^{\frac 1 {p-1}}, \vert z \vert^{\frac {p-2}{p-1}} y - \vert z\vert^{\frac {p-2}{p-1}} \E y \rangle - B(\vert z\vert^{\frac {p-2}{p-1}} \E y, x_2) \,.\]
    If we take the expectation over $y,z$ on both sides, note that we have that
    \[ \E_{y,z} B(\vert z\vert^{\frac {p-2}{p-1}} y, \vert z\vert^{\frac {p-2}{p-1}} \E y) - B(\vert z\vert^{\frac {p-2}{p-1}} y, x_2) \leq \E_{y,z} B(\vert z\vert^{\frac {p-2}{p-1}} y, \vert z\vert^{\frac {p-2}{p-1}} \E y) - B(\vert z\vert^{\frac {p-2}{p-1}} y, \vert z\vert^{\frac {p-2}{p-1}} \E [y \vert z]) \,,\]
    since Bregman divergences are minimized at the expectation by~\cref{fact:bregman_min_expecation}. Furthermore, since the Bregman divergence for the $\frac p {p-1}$ power is a homogeneous function of degree $\frac p {p-1}$ we can move the terms involving $z$ out of the divergence. Using this as well as~\cref{lem:separate-v-u-decrease} we have that
    \begin{align*}
        \E_{y,z} \Biggl[&B(\vert z\vert^{\frac {p-2}{p-1}} y, \vert z\vert^{\frac {p-2}{p-1}} \E y) - B(\vert z\vert^{\frac {p-2}{p-1}} y, \vert z\vert^{\frac {p-2}{p-1}} \E [y \vert z])\Biggr] = \E_{y,z} \vert z\vert^{\frac{p(p-2)}{(p-1)^2}} \left[B(y, \E y) - B(y, \E [y \vert z])\right] \\
        &\leq O\left( \left(\E_{y,z} [\vert y\vert^{\frac p {p-1}} \vert z\vert^{\frac p {p-1}}] + \E_y [\vert y\vert^{\frac p {p-1}}] \E_z [\vert z\vert^{\frac p {p-1}}] \right)^{\frac {p-2}{p-1}} \left( \E_{y,z} (B(y, \E [y]) - B(y, \E [y \vert z])) \right)^{\frac 1 {p-1}} \right)\,.
    \end{align*}
    Combining this with our initial expression, we have that
    \begin{align}
        \label{eq:bregman-progress-1}
        &\left(\E_{y,z} \left[ \vert y\vert^{\frac p {p-1}} \vert z\vert^{\frac p {p-1}}\right] + \E_y \left[\vert y\vert^{\frac p {p-1}}\right] \E_z \left[\vert z\vert^{\frac p {p-1}}\right] \right)^{\frac {p-2}{p-1}} \left( \E_{y,z} (B(y, \E [y]) - B(y, \E [y \vert z])) \right)^{\frac 1 {p-1}} \\
        &\quad \geq \E_{y,z} \langle \eta \sgn(z) \cdot \vert z\vert^{\frac 1 {p-1}}, \vert z\vert^{\frac {p-2}{p-1}} y - \vert z\vert^{\frac {p-2}{p-1}} \E y \rangle - \E_{y,z} B(x_1, x_2) \,.
    \end{align}
    Note that $\E_{y,z} \langle \eta \sgn(z) \cdot \vert z\vert^{\frac 1 {p-1}}, \vert z\vert^{\frac {p-2}{p-1}} y - \vert z\vert^{\frac {p-2}{p-1}} \E y \rangle$ is simply $\eta \Cov(y,z)$ so it suffices to bound the last divergence. By~\cref{lem:bregman-update-divergence-bound} applied with $\Delta = \eta \sgn(z) \cdot \vert z\vert^{\frac 1 {p-1}}$, we have that
    \[
        B(x_1, x_2) \leq O\left( \eta^2 \vert z\vert^{\frac 2 {p-1}} \cdot \vert x_1\vert^{\frac{p-2}{p-1}} + \eta^p \vert z\vert^{\frac p {p-1}}\right) = O \Paren{ \eta^2 \vert z\vert^{\frac {p^2 - 2p + 2}{(p-1)^2}} (\E [\vert y\vert])^{\frac{p-2}{p-1}} + \eta^p \vert z\vert^{\frac p {p-1}}}\,.
    \]
    When we take the expectation over $z$ on both sides, by H\"olders and convexity we have that
    \[
        \E_z B(x_1, x_2) \leq O\left( \eta^2 \left(\E \vert z\vert^{\frac p {p-1}} \right)^{\frac {p^2 - 2p + 2}{p(p-1)}} \left( \E \left[ \vert y\vert^{\frac p {p-1}} \right]\right)^{\frac {p-2} p} + \eta^p \left(\E \left[ \vert z\vert^{\frac p {p-1}} \right] \right)\right) \,.
    \]
    Now we choose $\eta$ so that $\eta \Cov(y,z)$ dominates $\E_v B(x_1,x_2)$.
    It suffices to ensure
    \[
    \eta \ll \frac{\Cov(y,z)}{\left(\E \vert z\vert^{\frac p {p-1}} \right)^{\frac {p^2 - 2p + 2}{p(p-1)}} \left( \E \left[ \vert y\vert^{\frac p {p-1}} \right]\right)^{\frac {p-2} p}}
    \]
    and
    \[
    \eta \ll \Paren{ \frac{ \Cov(y,z)}{\Paren{\E \vert z\vert^{\frac p {p-1}}}}   }^{\frac 1 {p-1}} \, ,
    \]
    so for a small enough constant $c$ we choose
    \[
    \eta = c \min \Paren{ \frac{\Cov(y,z)}{\left(\E \vert z\vert^{\frac p {p-1}} \right)^{\frac {p^2 - 2p + 2}{p(p-1)}} \left( \E \left[ \vert y\vert^{\frac p {p-1}} \right]\right)^{\frac {p-2} p}}, \Paren{ \frac{ \Cov(y,z)}{\Paren{\E \vert z\vert^{\frac p {p-1}}}}   }^{\frac 1 {p-1}}} \, .
    \]
    
    Putting this together with \eqref{eq:bregman-progress-1}, we get
    \[
        \E_{y,z} (B(y, \E [y]) - B(y, \E [y \vert z])) \geq \Omega(1) \cdot \frac{\min \Paren{ \frac{\Cov(y,z)^{2(p-1)}}{\left(\E \vert z\vert^{\frac p {p-1}} \right)^{\frac {p^2 - 2p + 2}{p}} \left( \E \left[ \vert y\vert^{\frac p {p-1}} \right]\right)^{\frac {(p-2)(p-1)} p}},  \frac{ \Cov^p(y,z)}{\Paren{\E \vert z\vert^{\frac p {p-1}}}} }}
        {\left(\E_{y,z} [\vert y\vert^{\frac p {p-1}} \vert z\vert^{\frac p {p-1}}] + \E_y [\vert y\vert^{\frac p {p-1}}] \E_z [\vert z\vert^{\frac p {p-1}}] \right)^{p-2}}
    \]

    We now simplify the minimum term.
    Using
    \[
        \min\{A,B\} \geq \frac{AB}{A+B} = \frac{1}{\frac{1}{A} + \frac{1}{B}}
    \]
    with
    \[
        A = \frac{\Cov(y,z)^{2(p-1)}}{\left(\E \vert z\vert^{\frac p {p-1}} \right)^{\frac {p^2 - 2p + 2}{p}} \left( \E \left[ \vert y\vert^{\frac p {p-1}} \right]\right)^{\frac {(p-2)(p-1)} p}}
        \qquad\text{and}\qquad
        B = \frac{ \Cov^p(y,z)}{\Paren{\E \vert z\vert^{\frac p {p-1}}}} \,,
    \]
    we get
    \begin{align*}
        \min\{A,B\}
        &\geq \frac{1}{\frac{\left(\E \vert z\vert^{\frac p {p-1}} \right)^{\frac {p^2 - 2p + 2}{p}} \left( \E \left[ \vert y\vert^{\frac p {p-1}} \right]\right)^{\frac {(p-2)(p-1)} p}}{\Cov(y,z)^{2(p-1)}}
        + \frac{\left(\E \vert z\vert^{\frac p {p-1}}\right)}{\Cov(y,z)^p}} \\
        &= \frac{\Cov(y,z)^{2(p-1)}}{\left(\E \vert z\vert^{\frac p {p-1}} \right)^{\frac {p^2 - 2p + 2}{p}} \left( \E \left[ \vert y\vert^{\frac p {p-1}} \right]\right)^{\frac {(p-2)(p-1)} p}
        + \Cov(y,z)^{p-2}\left(\E \vert z\vert^{\frac p {p-1}}\right)}\,.
    \end{align*}
    Plugging this lower bound into the previous display yields
    \begin{align*}
        & \E_{y,z} (B(y, \E [y]) - B(y, \E [y \vert z]))\\
        & \quad \geq
        \Omega(1) \cdot 
        \left(\Cov(y,z)^{2(p-1)}\right)\cdot
        \left(\E_{y,z} \left[\vert y\vert^{\frac p {p-1}} \vert z\vert^{\frac p {p-1}}\right] + \E_y \left[\vert y\vert^{\frac p {p-1}}\right] \E_z \left[\vert z\vert^\frac p {p-1}\right] \right)^{-(p-2)} \\
        & \qquad \left(\left(\E \vert z\vert^{\frac p {p-1}} \right)^{\frac {p^2 - 2p + 2}{p}} \left( \E \left[ \vert y\vert^{\frac p {p-1}} \right]\right)^{\frac {(p-2)(p-1)} p}
        + \Cov(y,z)^{p-2}\left(\E \vert z\vert^{\frac p {p-1}}\right)\right)^{-1}\,.
    \end{align*}

    We now lower bound the last two factors using only $\frac{p}{p-1}$-th moments.
    Using triangle, Holder, and Jensen's inequality,
    \[
        \abs{\Cov(y,z)} \leq \Paren{\E[\vert y\vert^{\frac p {p-1}} \vert z\vert^{\frac p {p-1}}]}^{\frac {p-1}{p}} + \Paren{\E \vert y\vert^{\frac p {p-1}} \, \E \vert z\vert^{\frac p {p-1}}}^{\frac {p-1} {p}}\,.
    \]
    Hence,
    \begin{align*}
        \Cov(y,z)^{p-2} &\leq O\left( \Paren{\E[\vert y\vert^{\frac p {p-1}} \vert z\vert^{\frac p {p-1}}]}^{\frac {p-1}{p}} + \Paren{\E \vert y\vert^{\frac p {p-1}} \, \E \vert z\vert^{\frac p {p-1}}}^{\frac {p-1} {p}}\right)\\
        &\leq O\left(\Paren{\E[\vert y\vert^{\frac p {p-1}} \vert z\vert^{\frac p {p-1}}] + \E \vert y\vert^{\frac p {p-1}} \, \E \vert z\vert^{\frac p {p-1}}}^{\frac {(p-1)(p-2)}{p}}\right)\,.
    \end{align*}
    We thus have that
    \begin{align*}
        &\left(\E_{y,z} \left[\vert y\vert^{\frac p {p-1}} \vert z\vert^{\frac p {p-1}}\right] + \E_y \left[\vert y\vert^{\frac p {p-1}}\right] \E_z \left[\vert z\vert^\frac p {p-1}\right] \right)^{p-2} \\
        &\quad\left(\left(\E \vert z\vert^{\frac p {p-1}} \right)^{\frac {p^2 - 2p + 2}{p}} \left( \E \left[ \vert y\vert^{\frac p {p-1}} \right]\right)^{\frac {(p-2)(p-1)} p}
        + \Cov(y,z)^{p-2}\left(\E \vert z\vert^{\frac p {p-1}}\right)\right) \\
        &\quad \quad\leq O \Biggl(
        \left(\E_{y,z} \left[\vert y\vert^{\frac p {p-1}} \vert z\vert^{\frac p {p-1}}\right] + \E_y \left[\vert y\vert^{\frac p {p-1}}\right] \E_z \left[ \vert z\vert^\frac p {p-1}\right] \right)^{p-2}\left(\E \vert z\vert^{\frac p {p-1}} \right)^{\frac {p^2 - 2p + 2}{p}} \left( \E \left[ \vert y\vert^{\frac p {p-1}} \right]\right)^{\frac {(p-2)(p-1)} p}
        \\
        &\quad \qquad \qquad+ \left(\E \vert z\vert^{\frac p {p-1}}\right)\left(\E_{y,z} \left[\vert y\vert^{\frac p {p-1}} \vert z\vert^{\frac p {p-1}}\right] + \E_y \left[\vert y\vert^{\frac p {p-1}}\right] \E_z \left[\vert z\vert^\frac p {p-1}\right]\right)^{\frac{(p-2)(2p-1)}{p}}\Biggr) \\
        &\quad \quad \leq O\Paren{        \left(\E \vert z\vert^{\frac p {p-1}}\right)\left(\E_{y,z} \left[\vert y\vert^{\frac p {p-1}} \vert z\vert^{\frac p {p-1}}\right] + \E_y \left[\vert y\vert^{\frac p {p-1}}\right] \E_z \left[\vert z\vert^\frac p {p-1}\right]\right)^{\frac{(p-2)(2p-1)}{p}}}
    \end{align*}
\end{proof}

\subsection{Rounding error to potential decrease}

In this subsection, we will give a general lemma that can be used to analyze how rounding error on typical objectives is related to a potential decrease when the potential function is based on Bregman uncertainty. 
\begin{lemma}
\label{lem:general-p-rounding-err}
    Let $p$ be an even integer and let $q$ be an integer in $[1, p-1]$. Let $\mu_{y,z}$ be a random pseudodistribution with indeterminates $y,z$ and let $\err$ denote 
    \[ \E_{\{\mu\}} \left[ \pE_{y,z \sim \mu_y \otimes \mu_z} \sum_{i,j} A_{ij}^q y_i z_j - \pE_{y,z \sim \mu_{y,z}} \sum_{i,j} A_{ij}^q y_i z_j\right]\,,\]
    where $\E_{\{\mu\}}$ denotes the expectation over the randomness in the choice of $\mu$. Let $y_i^{(0)}$ denote the expected distribution of $y_i$ over the randomness in $\mu$ and define the potential function
    \[ \Phi(y) = \sum_i \left(\phi_{y_i^{(0)}}^{(p)}\right)^{\frac{p-1}{p-q} - 1} \cdot \left( \phi_{y_i}^{(p)}\right)\,,\]
    where $\phi_{y_i}^{(p)}$ is the $p$-Bregman uncertainty of $y_i$ and $\phi_{y_i^{(0)}}^{(p)}$ is the $p$-Bregman uncertainty of $y_i^{(0)}$. Let $\psi$ be a distribution over $[n]$ where $\psi(j) \propto \left( \E_{\{\mu\}} \pE [\vert z_j\vert^{\frac p {p-1}}] \right)^{\frac{p-1}{p-q}}$. Then we have that
    \begin{align*}
        \err \leq &O(1) \cdot \left(\E_{\{\mu\}} \left[ \sum_j \pE \left[\vert z_j\vert^{\frac p {p-q}}\right] \right]\right)^{\frac{p-q}{2(p-1)^2}} \\
        &\qquad\left( \E_{\{\mu\}} \left[ \sum_{ij} \left( \pE \left[\vert y_i\vert^{\frac p {p-q}} \vert z_j\vert^{\frac p {p-q}}\right] + \pE \left[\vert y_i\vert^{\frac p {p-q}}\right] \pE \left[\vert z_j\vert^{\frac p {p-q}}\right] \right) \right] \right)^{\frac{(p-q)(p-2)(2p-1)}{2p(p-1)^2}} \\ 
        &\qquad\norm{A}_p^q \left( \E_{\{\mu\}} \left[\E_{j \sim \psi} \left( \Phi(y) - \E_{z_j} \Phi(y \vert z_j) \right) \right]\right)^{\frac{p-q}{2(p-1)^2}}\,.
    \end{align*}
\end{lemma}

\begin{remark}
    Note that a sum-of-squares pseudodistribution is consistent with a real distribution on subsets of size at most the degree. Thus, we abuse notation and write $\pE$ to denote the expectation over these local distributions for functions of one or two indeterminates, even when they are not (integer) polynomials.
\end{remark}

\begin{proof}[Proof of~\cref{lem:general-p-rounding-err}]
    Note that we have that
    \[ \err = \E_{\{\mu\}} \left[ \pE_{y,z \sim \mu_y \otimes \mu_z} \sum_{i,j} A_{ij}^q y_i z_j - \pE_{y,z \sim \mu_{y,z}} \sum_{i,j} A_{ij}^q y_i z_j  \right] \leq \E_{\{\mu\}}\sum_{i,j} \vert A_{ij} \vert^q \, \vert \Cov(y_i, z_j) \vert \,.\]
    Let $\Delta(i \vert j) = \E_{z_j} \left[ \phi_{y_i}^{(p)} - \phi_{y_i \vert z_j}^{(p)}\right]$ and define
    \[ \alpha_{ij} = \pE_{y_i,z_j} \left[\vert y_i\vert^{\frac p {p-1}} \vert z_j\vert^{\frac p {p-1}}\right] + \pE \left[\vert y_i\vert^{\frac p {p-1}}\right] \pE \left[\vert z_j\vert^\frac p {p-1}\right]\,.\]
    Then by~\cref{lem:general-bregman-p-uncertainty-decrease} we have that
    \begin{align*}
        \err &\leq \E_{\{\mu\}} \left[ \sum_{i,j} \vert A_{ij} \vert^q \, \vert \Cov(y_i, z_j) \vert \right] \\
        &\leq O(1) \cdot \sum_{i,j} \vert A_{ij} \vert^q \cdot \E_{\{\mu\}} \left[ \left(\pE \left[\vert z_j\vert^\frac p {p-1}\right] \right)^{\frac 1 {2(p-1)}}  \alpha_{ij}^{\frac{(p-2)(2p-1)}{2p(p-1)}} \Delta(i \vert j)^{\frac 1 {2(p-1)}}\right]\,.
    \end{align*}
    Note that $\frac 1 {2(p-1)} + \frac{(p-2)(2p-1)}{2p(p-1)} + \frac 1 {2(p-1)} \leq 1$ for $p \geq 2$ and thus we have by H\"older and convexity that
    \begin{align*}
        \err &\leq O(1) \cdot \sum_{i,j} \vert A_{ij} \vert^q \cdot \left( \E_{\{\mu\}} \pE \left[\vert z_j\vert^\frac p {p-1}\right] \right)^{\frac 1 {2(p-1)}}  \left( \E_{\{\mu\}} \alpha_{ij}\right)^{\frac{(p-2)(2p-1)}{2p(p-1)}} \left(\E_{\{\mu\}} \Delta(i \vert j)\right)^{\frac 1 {2(p-1)}}\,.
    \end{align*}
    Applying H\"older with $p/q$ and $p/(p-q)$ we have that
    \begin{align*}
        \err &\leq O(1) \cdot \sum_{i,j} \vert A_{ij} \vert^q \cdot \left( \E_{\{\mu\}} \pE \left[\vert z_j\vert^\frac p {p-1}\right] \right)^{\frac 1 {2(p-1)}}  \left( \E_{\{\mu\}} \alpha_{ij}\right)^{\frac{(p-2)(2p-1)}{2p(p-1)}} \left(\E_{\{\mu\}} \Delta(i \vert j)\right)^{\frac 1 {2(p-1)}} \\
        &\leq O(1) \cdot \norm{A}_p^q \left( \sum_{i,j} \left(\E_{\{\mu\}} \pE \left[\vert z_j\vert^\frac p {p-1}\right] \right)^{\frac {p}{2(p-1)(p-q)}} \left( \E_{\{\mu\}} \alpha_{ij}\right)^{\frac{(p-2)(2p-1)}{2(p-q)(p-1)}} \left( \E_{\{\mu\}} \Delta(i \vert j)\right)^{\frac p {2(p-q)(p-1)}} \right)^{\frac{p-q}{p}}\,.
    \end{align*}
    Applying H\"olders again with $\frac{2(p-1)^2}{(p-2)(2p-1)}$ and $\frac{2(p-1)^2}{p}$ gives us that
    \begin{align*}
        \err &\leq O(1) \cdot \norm{A}_p^q \left( \sum_{i,j} \left(\E_{\{\mu\}} \pE \left[\vert z_j\vert^\frac p {p-1}\right] \right)^{\frac {p}{2(p-1)(p-q)}} \left( \E_{\{\mu\}} \alpha_{ij}\right)^{\frac{(p-2)(2p-1)}{2(p-q)(p-1)}} \left( \E_{\{\mu\}} \Delta(i \vert j)\right)^{\frac p {2(p-q)(p-1)}} \right)^{\frac{p-q}{p}} \\
        &\leq O(1) \cdot \norm{A}_p^q \left( \sum_{i,j} \left( \E_{\{\mu\}} \alpha_{ij}\right)^{\frac {p-1}{p-q}} \right)^{\frac{(p-q)(p-2)(2p-1)}{2p(p-1)^2}} \left( \sum_{i,j} \left(\E_{\{\mu\}} \E \left[\vert z_j\vert^\frac p {p-1}\right] \right)^{\frac{p-1}{p-q}} \left( \E_{\{\mu\}} \Delta(i \vert j) \right)^{\frac{p-1}{p-q}}\right)^{\frac{p-q}{2(p-1)^2}} \,.
    \end{align*}
    Let by $\psi(j) \propto \left( \E_{\{\mu\}} \left[ \pE\left[\vert z_j\vert^{\frac p {p-1}}\right] \right]\right)^{\frac{p-1}{p-q}}$. Then we have that
    \begin{align*}
        \err &\leq O(1) \cdot \norm{A}_p^q \left( \sum_{i,j} \left( \E_{\{\mu\}} \alpha_{ij}\right)^{\frac {p-1}{p-q}} \right)^{\frac{(p-q)(p-2)(2p-1)}{2p(p-1)^2}} \left( \sum_{i,j} \left(\E_{\{\mu\}} \pE \left[\vert z_j\vert^\frac p {p-1}\right] \right)^{\frac{p-1}{p-q}} \left( \E_{\{\mu\}} \Delta(i \vert j) \right)^{\frac{p-1}{p-q}}\right)^{\frac{p-q}{2(p-1)^2}} \\
        &= O(1) \cdot \norm{A}_p^q \left( \sum_{i,j} \left( \E_{\{\mu\}} \alpha_{ij}\right)^{\frac {p-1}{p-q}} \right)^{\frac{(p-q)(p-2)(2p-1)}{2p(p-1)^2}} \left(  \sum_j \left( \E_{\{\mu\}} \left[ \pE\left[\vert z_j\vert^{\frac p {p-1}}\right] \right]\right)^{\frac{p-1}{p-q}}\right)^{\frac{p-q}{2(p-1)^2}} \\
        &\quad\left( \E_{j \sim \psi} \sum_i \left( \E_{\{\mu\}} \Delta(i \vert j) \right)^{\frac{p-1}{p-q}} \right)^{\frac{p-q}{2(p-1)^2}}\,.
    \end{align*}
    We now note that $0 \leq \E_{\{\mu\}} \Delta(i \vert j) \leq \E_{\{\mu\}} \phi_{y_i}^{(p)} \leq \phi_{y_i^{(0)}}^{(p)}$, where the last inequality follows by Jensen's inequality on the expectation over $\{\mu\}$. Therefore,
    \begin{align*}
        \E_{j \sim \psi} \sum_i \left( \E_{\{\mu\}} \Delta(i \vert j) \right)^{\frac{p-1}{p-q}}
        &\leq \E_{j \sim \psi} \sum_i \left(\phi_{y_i^{(0)}}^{(p)}\right)^{\frac{p-1}{p-q}-1} \E_{\{\mu\}} \Delta(i \vert j) = \E_{\{\mu\}} \left[\E_{j \sim \psi} \left( \Phi(y) - \E_{z_j} \Phi(y \vert z_j) \right)\right]\,.
    \end{align*}
    Plugging this back into the previous inequality and using that by convexity and Almost Triangle Inequality we have that
    \[  \left(\E_{\{\mu\}} \pE [\vert z_j\vert^\frac p {p-1}] \right)^{\frac{p-1}{p-q}} \leq \E_{\{\mu\}} \pE [\vert z_j\vert^{\frac p {p-q}}] \quad \text{and} \quad \alpha_{ij}^{\frac {p-1}{p-q}} \leq O\left(  \pE \left[\vert y_i\vert^{\frac p {p-q}} \vert z_j\vert^{\frac p {p-q}}\right] + \pE \left[\vert y_i\vert^{\frac p {p-q}}\right] \pE \left[\vert z_j\vert^{\frac p {p-q}}\right]  \right)\]
    completes the proof.
\end{proof}

\subsection{Generalizing fractional Bregman uncertainty to asymmetric correlations}

In this section, we generalize the potential decrease lemma from the previous sections to cases where we have correlations of the form $\Cov(u^p, v^q)$ for $p\neq q$. This will be necessary for our application to $p\rightarrow q$ norms.

\begin{lemma}
\label{lem:gen-p-q-norm-potential-decrease}
    Let $p \leq q$ be even integers and let $B$ be the Bregman divergence associated with $\vert x \vert^{\frac p {p-1}}$. Then we have that
    \[\E_{u,v} \left( B( u^{p-1}, \E u^{p-1}) - B ( u^{p-1},  \E [u^{p-1} \vert v])  \right) \geq \frac{\Cov^{\frac{2q(p-1)}{p}}(u^{p-1}, v^{q-1})}{\left( \E v^q \right)^{\frac{(q-p)(p-2)+pq}{p^2}} \left( \E [u^p v^q] + \E u^p \E v^q \right)^{\frac{2 pq (p-2) - (p^2 - 2q)}{p^2}}}\,.\]
\end{lemma}

\begin{proof}
    Consider the problem of predicting $v^{\frac{pq-p-q}{p}} u^{p-1}$ given $v$. Let $x_1 = \vert v \vert^{\frac{pq-p-q}{p}} \E u^{p-1}$ and let $x_2$ be the result of a mirror descent update step on $x_1$ with cost function $\sgn(v) \cdot \vert v \vert^{\frac{q}{p}}$. Applying~\cref{fact:bregman_three_point} to $\vert v \vert^{\frac{pq-p-q}{p}} u^{p-1}$, $x_1$, $x_2$ we have that
    \begin{align*}
         &\langle \eta \sgn(v) \cdot \vert v \vert^{\frac{q}{p}}, \vert v \vert^{\frac{pq-p-q}{p}} u^{p-1} - \vert v \vert^{\frac{pq-p-q}{p}} \E u^{p-1} \rangle  \\
         &\qquad= B(\vert v \vert^{\frac{pq-p-q}{p}} u^{p-1}, \vert v \vert^{\frac{pq-p-q}{p}} \E u^{p-1}) - B (\vert v \vert^{\frac{pq-p-q}{p}} u^{p-1}, x_2) + B(x_1, x_2)\,.
    \end{align*}
    When we take the expectation over $u,v$ on both sides, we have that the left-hand side is $\eta \Cov(u^{p-1}, v^{q-1})$, while the right-hand side is at most
    \[ \E_{u,v}  \left[ B(\vert v \vert^{\frac{pq-p-q}{p}} u^{p-1}, \vert v \vert^{\frac{pq-p-q}{p}} \E u^{p-1}) - B (\vert v \vert^{\frac{pq-p-q}{p}} u^{p-1}, \vert v \vert^{\frac{pq-p-q}{p}} \E [u^{p-1} \vert v]) + B(x_1, x_2) \right]\,,\]
    since expectations minimize the expected Bregman divergence to a fixed point (\cref{fact:bregman_min_expecation}). Thus,
    \begin{align*}
        \eta \Cov(u^{p-1}, v^{q-1}) \leq  \E_{u,v}  \Biggl[ &B(\vert v \vert^{\frac{pq-p-q}{p}} u^{p-1}, \vert v \vert^{\frac{pq-p-q}{p}} \E u^{p-1}) \\
        &\quad- B (\vert v \vert^{\frac{pq-p-q}{p}} u^{p-1}, \vert v \vert^{\frac{pq-p-q}{p}} \E [u^{p-1} \vert v]) + B(x_1, x_2) \Biggr] \,.
    \end{align*} 
    Note that by a similar argument to~\cref{lem:separate-v-u-decrease} we have that
    \begin{align*}
        &\E_{u,v}  \Biggl[ B(\vert v \vert^{\frac{pq-p-q}{p}} u^{p-1}, \vert v \vert^{\frac{pq-p-q}{p}} \E u^{p-1}) - B (\vert v \vert^{\frac{pq-p-q}{p}} u^{p-1}, \vert v \vert^{\frac{pq-p-q}{p}} \E [u^{p-1} \vert v]) \Biggr] \\
        &\quad= \E_{u,v}  \Biggl[\vert v \vert^{\frac{pq-p-q}{p-1}} \left( B( u^{p-1}, \E u^{p-1}) - B ( u^{p-1},  \E [u^{p-1} \vert v])  \right)\Biggr] \\
        &\quad\leq \left( \E_{u,v} v^q \left( B( u^{p-1}, \E u^{p-1}) - B ( u^{p-1},  \E [u^{p-1} \vert v])  \right) \right)^{\frac{pq-p-q}{q(p-1)}} \\
        &\qquad\qquad\left( \E_{u,v} \left( B( u^{p-1}, \E u^{p-1}) - B ( u^{p-1},  \E [u^{p-1} \vert v])  \right)\right)^{\frac{p}{q(p-1)}} \\
        &\quad\leq \left( \E [u^p v^q] + \E u^p \E v^q \right)^{\frac{pq-p-q}{q(p-1)}} \left( \E_{u,v} \left( B( u^{p-1}, \E u^{p-1}) - B ( u^{p-1},  \E [u^{p-1} \vert v])  \right)\right)^{\frac{p}{q(p-1)}} \,.
    \end{align*}
    We now turn our attention to the $\E_{u,v} B(x_1, x_2)$ error term. By~\cref{lem:bregman-update-divergence-bound} we have that
    \begin{align*}
        \E_{u,v} B(x_1, x_2) \leq \E_{u,v} \left[ \eta^2 \vert v \vert^{\frac{2q}{p}} \left( \vert v \vert^{\frac{pq-p-q}{p}}\E u^{p-1} \right)^{\frac{p-2}{p-1}} + \eta^p \vert v \vert^{q}\right] \leq \E_{u,v} \left[ \eta^2 \vert v \vert^{\frac{pq-p-q+2}{p-1}} \left( \E u^{p} \right)^{\frac{p-2}{p}} + \eta^p \vert v \vert^{q}\right] \,.
    \end{align*}
    Thus, combining this we have that
    \begin{align*}
        \eta \Cov(u^{p-1}, v^{q-1}) \leq &\left( \E [u^p v^q] + \E u^p \E v^q \right)^{\frac{pq-p-q}{q(p-1)}} \left( \E_{u,v} \left( B( u^{p-1}, \E u^{p-1}) - B ( u^{p-1},  \E [u^{p-1} \vert v])  \right)\right)^{\frac{p}{q(p-1)}} \\
        &\quad+  \E_{u,v} \left[ \eta^2 \vert v \vert^{\frac{pq-p-q+2}{p-1}} \left( \E u^{p} \right)^{\frac{p-2}{p}} + \eta^p v^{q}\right]\,.
    \end{align*}    
    We aim to pick $\eta$ such that the error term $\E_{u,v} \left[ \eta^2 \vert v \vert^{\frac{pq-p-q+2}{p-1}} \left( \E u^{p} \right)^{\frac{p-2}{p}} + \eta^p v^{q}\right]$ is smaller than the left hand side. Thus, it suffices to take
    \[ \eta \leq \min \left( \frac{\Cov(u^{p-1}, v^{q-1})}{\left(\E v^q \right)^{\frac{pq-p-q+2}{q(p-1)}} \left( \E u^{p} \right)^{\frac{p-2}{p}}},\, \left( \frac{\Cov(u^{p-1}, v^{q-1})}{\E v^q} \right)^{\frac 1 {p-1}}\right)\,.\]
    Furthermore, note that since 
    \[ \min(A, B) \geq \frac{AB}{A+B} = \frac{1}{\frac 1 A + \frac 1 B}\,,\]
    we have that 
    \begin{align*}
        \eta \geq \frac 1 {\frac {\left(\E v^q \right)^{\frac{pq-p-q+2}{q(p-1)}} \left( \E u^{p} \right)^{\frac{p-2}{p}}} {\Cov(u^{p-1}, v^{q-1})} + \frac {\left( \E v^q\right)^{\frac 1 p}} {\Cov^{\frac 1 p}(u^{p-1}, v^{q-1})}} = \frac {\Cov (u^{p-1}, v^{q-1})} {\left(\E v^q \right)^{\frac{pq-p-q+2}{q(p-1)}} \left( \E u^{p} \right)^{\frac{p-2}{p}} + \Cov^{\frac{p-2}{p-1}}(u^{p-1}, v^{q-1}) \cdot \left(E v^q\right)^{\frac 1 {p-1}}} \,.
    \end{align*}
    Observe that by Triangle Inequality, we have that
    \[ \Cov^{\frac{p-2}{p-1}}(u^{p-1}, v^{q-1}) \leq O(1) \cdot \left( \E_{u,v} u^{p-1} v^{q-1} + \E u^{p-1} \E v^{q-1}\right)^{\frac{p-2}{p-1}}\]
    and by H\"olders and convexity we have that
    \[ \E_{u,v} u^{p-1} v^{q-1} \leq \left( \E_{u,v} u^p v^q \right)^{\frac{p-1}{p}} \left( \E v^q \right)^{\frac{q-p}{pq}}\,.\]
    Thus we have that
    \[\Cov^{\frac{p-2}{p-1}}(u^{p-1}, v^{q-1}) \leq \left( \E_{u,v} u^p v^q \right)^{\frac{p-2}{p}} \left( \E v^q \right)^{\frac{(q-p)(p-2)}{pq(p-1)}} + \left(\E u^{p}\right)^{\frac{p-2}{p}} \left(\E v^{q}\right)^{\frac{(p-2)(q-1)}{(p-1)q}}\,,\]
    and thus
    \[ \eta \geq \frac{\Cov(u^{p-1}, v^{q-1})}{\left( \E v^q \right)^{\frac{(q-p)(p-2) + pq}{pq(p-1)}} \left( \left(\E_{u,v} u^p v^q\right)^{\frac {p-2} p} + \left(\E u^p \E v^q\right)^{\frac {p-2}{p}} \right)}\,.\]
    Returning to the original inequality, we have that
    \[ \eta \Cov(u^{p-1}, v^{q-1}) \leq \left( \E [u^p v^q] + \E u^p \E v^q \right)^{\frac{pq-p-q}{q(p-1)}} \left( \E_{u,v} \left( B( u^{p-1}, \E u^{p-1}) - B ( u^{p-1},  \E [u^{p-1} \vert v])  \right)\right)^{\frac{p}{q(p-1)}}\,.\]
    Note that
    \[ \left(\E_{u,v} u^p v^q\right)^{\frac {p-2} p} + \left(\E u^p \E v^q\right)^{\frac {p-2}{p}} \leq O(1) \cdot \left(\E_{u,v} u^p v^q + \E u^p \E v^q\right)^{\frac {p-2}{p}}\]
    so plugging in the lower bound on $\eta$ and rearranging yields
    \[ \E_{u,v} \left( B( u^{p-1}, \E u^{p-1}) - B ( u^{p-1},  \E [u^{p-1} \vert v])  \right) \geq \frac{\Cov^{\frac{2q(p-1)}{p}}(u^{p-1}, v^{q-1})}{\left( \E v^q \right)^{\frac{(q-p)(p-2)+pq}{p^2}} \left( \E [u^p v^q] + \E u^p \E v^q \right)^{\frac{2 pq (p-2) - (p^2 - 2q)}{p^2}}}\,.\]
\end{proof}

\section{Additive Approximation for LRA}

In this section, we will give an analysis of the additive approximation for $\ell_p$ LRA. 

\begin{theorem}
\label{thm:lra-additive}
    Given a matrix $A \in \R^{n \times m}$ with entries with bit complexity $B$, even $p \in N$, and $0<\eps<1$, there exists an algorithm that runs in $(nmB)^{(1/\eps)^{O_p(1)}}$ time and outputs matrices $\hat{u} \in \R^{n}, \hat{v} \in \R^{m}$ such that with probability at least $0.99$, 
    \begin{equation*}
        \norm{ A - \hat{u} \hat{v}^\top }_p^p \leq \OPT + \eps \cdot \Norm{A}_p^p\,. 
    \end{equation*}
\end{theorem}

Our algorithm uses the same approach in the work of~\cite{anderson2025additiveapproximationschemeslowdimensional}, but we give a tighter analysis that avoids their quasipolynomial dependence on $n$.

\begin{mdframed}
  \begin{algorithm}[$\ell_p$ LRA Additive Approximation]
    \label{algo:lra-additive}\mbox{}
    \begin{description}
    \item[Input:] Matrix $A \in \mathbb{R}^{n \times m}$ and target accuracy $0<\epsilon<1$.
    
    \item[Operations:]\mbox{}
    \begin{enumerate}
        \item Let $\calS$ be a $\frac{\epsilon^p}{n^3} \left(\sum_{i \in [n], j \in [m]} A_{ij}^p\right)^{1/2}$ grid of \[\left[-2 \left(\sum_{i \in [n], j \in [m]} A_{ij}^p\right)^{1/2}, 2 \left(\sum_{i \in [n], j \in [m]} A_{ij}^p\right)^{1/2}\right]\,.\]
        \item For each $s_u \in \calS, s_v \in \calS$ and set $\calT$ of size at most $O_p\left(\eps^{-6}\right)$:
        \begin{enumerate}
            \item Let $\Sigma$ be a $\frac{\epsilon^p}{n^2} \cdot \left(\sum_{i \in [n], j \in [m]} A_{ij}^p\right)^{1/(2p)}$ grid of 
            \[\left[-\left(\sum_{i \in [n], j \in [m]} A_{ij}^p\right)^{1/(2p)}, \left(\sum_{i \in [n], j \in [m]} A_{ij}^p\right)^{1/(2p)}\right]\,.\]
            \item Let $\mu$ be a $O_p\left(\eps^{-6} \right)$-degree sum-of-squares pseudo-distribution over $\left(\Sigma\right)^{n+m}$ such that $\pE_{\mu}$ optimizes 
            \begin{align}
                &\min_{\pE_{\mu}} \quad \pE_{\mu} \sum_{i\in[n],j\in[m]} (A_{ij} - u_i v_j)^p \notag \\
                &\text{s.t.\ } \pE_{\mu} \text{ satisfies} \notag \\
                &\quad \sum_{i \in [n]} u_i^p = s_u \label{eq:constraint1} \\
                &\quad \sum_{j \in [m]} v_j^p = s_v \label{eq:constraint2}
            \end{align}
            \item Let $\hat{v}_\calT$ be a draw from the local distribution $\{v_{\calT}\}$. Let $\mu_\calT$ be the pseudodistribution obtained by conditioning on $\{v_j = \hat{v}_j\}_{j \in \calT}$.
            \item Let $\hat{u}, \hat{v}$ be the solution where $\hat{u}_i, \hat{v}_j$ are sampled independently from their local distributions in $\mu_\calT$.
        \end{enumerate}
    \end{enumerate}
    \item[Output:] The embedding $\hat{u}, \hat{v}$ with the lowest LRA objective value over all iterations of the loop.
    \end{description}
  \end{algorithm}
\end{mdframed}
\begin{remark}
    Although the above algorithm utilizes Sum-of-Squares pseudodistributions, we note that the Sherali-Adams linear programming hierarchy alongside global linear constraints would suffice. We use sum-of-squares pseudodistributions and SDPs for ease of exposition.
\end{remark}

Note that as in~\cite{anderson2025additiveapproximationschemeslowdimensional}, one of the SDPs with constraints and a discretized domain has optimal objective value at most $\OPT + \eps \cdot \norm{A}_p^p$. It thus suffices to show that our rounding algorithm produces $\hat{u}, \hat{v}$ with objective value at most $\eps \cdot \norm{A}_p^p$ worse than the SDP objective value. 

\begin{lemma}
\label{lem:additive-approx-rounding-error}
    Let $\OPT_{\mathrm{SDP}}$ be the objective value of the SDP in~\cref{algo:lra-additive} for some $s_u, s_v \in \mathcal{S}$. Then global correlation rounding with $(1/\eps)^{O_p(1)}$ rounds of conditioning produces a solution of value at most $\OPT_{\mathrm{SDP}} + \eps \norm{A}_p^p$.
\end{lemma}
\begin{proof}
    Note that since $p$ is an even integer, we can decompose the objective function as follows:
    \[ \sum_{i\in[n],j\in[m]} (A_{ij} - u_i v_j)^p = \sum_{q =0}^p \left(-1\right)^{p-q} \cdot \binom{p}{q} \sum_{i\in[n],j\in[m]} A_{ij}^q u_i^{p-q} v_j^{p-q}\,.\]
    Note that similarly to~\cite{anderson2025additiveapproximationschemeslowdimensional} the term $\sum_{i\in[n],j\in[m]} u_i^{p} v_j^{p}$ incurs no rounding error since
    \[ \sum_{i \in [n], j \in [m]} \pE_{\mu} u_i^p v_j^p = \sum_{i \in [n]} \pE_\mu u_i^p \left( \norm{v}_p^p\right) = \sum_{i \in [n]} s_v \pE_\mu u_i^p  = s_v \cdot s_u\,.\]    
    Thus, it remains to bound the rounding error for $1 \leq q \leq p-1$ (observe that when $q=p$ the corresponding polynomial is simply a constant). We will show that each such function $\sum_{i\in[n],j\in[m]} A_{ij}^q u_i^{p-q} v_j^{p-q}$ for $1 \leq q \leq p-1$ is potential aligned and then apply~\cref{lem:gcr_linear_combo}. Assume that the rounding error on this function is at least $\eps \norm{A}_p^p$, let $\left(u_i^{p-q}\right)^{(0)}$ denote the expected distribution of $u_i^{p-q}$ over the randomness in $\mu$, and let
    \[ \Phi(u) = \sum_i \left(\phi_{\left(u_i^{p-q}\right)^{(0)}}^{(p)}\right)^{\frac{p-1}{p-q}-1} \phi_{u_i^{p-q}}^{(p)}\,.\]
    Then by applying~\cref{lem:general-p-rounding-err} with $y_i = u_i^{p-q}$ and $z_j = v_j^{p-q}$ we have that\footnote{Note that while~\cref{lem:general-p-rounding-err} considers conditioning on $v_j^{p-q}$ conditioning on $v_j$ only further reduces the potential function since $v_j^{p-q}$ is a deterministic function of $v_j$.} 
    \begin{align*}
        \eps \norm{A}_p^p \leq &O(1) \cdot \left(\sum_j \E_{\{\mu\}}\E [v_j^p] \right)^{\frac{p-q}{2(p-1)^2}} \left( \sum_{ij} \E_{\{\mu\}} \left[\E [u_i^p v_j^p] + \E [u_i^p] \E [v_j^p] \right] \right)^{\frac{(p-q)(p-2)(2p-1)}{2p(p-1)^2}} \norm{A}_p^q \\
        &\qquad\left( \E_{\{\mu\}}\E_{j \sim \psi} \left( \Phi(u) - \E_{v_j} \Phi(u \vert v_j)\right)\right)^{\frac{p-q}{2(p-1)^2}}\,.
    \end{align*}
    By our sum-of-squares equality constraints\footnote{Note that conditioning preserves these equality constraints.}, we have that 
    \[ \sum_j \E_{\{\mu\}} \E [v_j^p] = s_v \quad \text{and} \quad \sum_{ij} \E_{\{\mu\}} \left[\E [u_i^p v_j^p] + \E [u_i^p] \E [v_j^p] \right] = 2 s_v s_u\,,\]
    and $s_v,s_u \leq \norm{A}_p^{p/2}$. Plugging this into the inequality above yields that
    \[ \eps \norm{A}_p^p \leq O(1) \cdot \norm{A}_p^{\frac{p(p-q)}{4(p-1)^2}} \cdot \norm{A}_p^{\frac{(p-q)(p-2)(2p-1)}{2(p-1)^2}}\cdot \norm{A}_p^q \cdot \left( \E_{\{\mu\}}\E_{j \sim \psi} \left( \Phi(u) - \E_{v_j} \Phi(u \vert v_j)\right)\right)^{\frac{p-q}{2(p-1)^2}}\,.\]
    Simplifying and cancelling terms we have that 
    \[ \eps^{O_{p}(1)} \norm{A}_p^{p(p-q)} \leq \left( \E_{\{\mu\}} \E_{j \sim \psi} \left( \Phi(u) - \E_{v_j} \Phi(u \vert v_j)\right)\right)^{2(p-q)},\]
    so we have that $\Phi$ for each $q$ is $(\eps \norm{A}_p^p, \eps^{O_p(1)} \norm{A}_p^{p/2})$-potential aligned. Furthermore, for each $q$ we have by~\cref{fact:gen-initial-potential-bound} and convexity that the initial potential satisfies
    \[ \Phi(u^{(0)}) = \sum_i \left\vert\phi_{\left(u_i^{p-q}\right)^{(0)}}^{(p)}\right\vert^{\frac{p-1}{p-q}} \leq \sum_i \left( \E_{u_i} \left\vert u_i^{p-q} \right\vert^{\frac p {p-1}} \right)^{\frac{p-1}{p-q}} \leq \sum_i \E_{u_i} u_i^p\,.\]
    Since for the initial pseudodistribution we have that $\sum_i \E_{u_i} u_i^p = s_u \leq \norm{A}_p^{p/2}$ we have by~\cref{lem:gcr_linear_combo} that the rounding error on the LRA objective is at most $\eps \norm{A}_p^p$ after $1/\eps^{O_p(1)}$ rounds of conditioning.
\end{proof}

Combining~\cref{lem:additive-approx-rounding-error} with the discretization error and runtime analysis in~\cite{anderson2025additiveapproximationschemeslowdimensional} (see section 6.2) gives the proof of~\cref{thm:lra-additive}.

\section{Additive Error to Relative Error for LRA}
\label{sec:ptas}

\begin{theorem}
\label{thm:lra-ptas}
    Given a matrix $A \in \R^{n \times m}$ with entries with bit complexity $B$, even $p \in N$, and $0<\eps<1$, there exists an algorithm that runs in $(nm\cdot 2^B)^{(1/\eps)^{O_p(1)}}$ time and outputs matrices $\hat{u} \in \R^{n}, \hat{v} \in \R^{m}$ such that with probability at least $0.99$, 
    \begin{equation*}
        \norm{ A - \hat{u} \hat{v}^\top }_p^p \leq (1+ \eps) \cdot \OPT\,. 
    \end{equation*}
\end{theorem}

Our algorithm follows the same high-level approach as the additive approximation, but with some key differences. We write down a series of SDPs, and then round each via global correlation rounding. However, unlike in the additive approximation, we also have to enforce some new constraints regarding the solutions' relation to an $O(\OPT)$ valued solution. 

We first define a set of constraints that we will require our pseudodistributions satisfy:

\begin{align*}
\mathcal{A} (s, d) = \left\{ 
\begin{gathered}
    \begin{aligned}
        \forall \, 0 \leq k \leq p: \quad &\sum_{i \in [n]} x_i^{p-k} u_i^k = s_{u,k}, \quad \sum_{j \in [m]} y_j^{p-k} v_j^k = s_{v,k} \\
        &\sum_{i \in [n]} (u_i-x_i)^p = d_u, \quad \sum_{j \in [m]} (v_j-y_j)^p = d_v
    \end{aligned} \\[1ex]
    \begin{aligned}
        (x_i - u_i)^2 &= \left(|x_i - u_i|^{\frac{1}{\prod_{q=1}^{p-1} (p-q)}}\right)^{2 \prod_{q=1}^{p-1} (p-q)} \\
        (y_j - v_j)^2 &= \left(|y_j - v_j|^{\frac{1}{\prod_{q=1}^{p-1} (p-q)}}\right)^{2 \prod_{q=1}^{p-1} (p-q)}
    \end{aligned} \\[1ex]
    |x_i - u_i|^{\frac{1}{\prod_{q=1}^{p-1} (p-q)}} \geq 0 \,,\quad |y_j - v_j|^{\frac{1}{\prod_{q=1}^{p-1} (p-q)}} \geq 0 \\
    \forall \, 1\leq q < p, \, a+b \leq p-q: \\
    \left(\sum_i |x_i|^{\frac{pa}{p-q}} \cdot |x_i - u_i|^{\frac{p(p-q-a)}{p-q}}\right)^{2(p-q)} \leq \|x\|_p^{2pa} \|x-u\|_p^{2p(p-q-a)} \\[1ex]
    \left(\sum_j |y_j|^{\frac{pb}{p-q}} \cdot |y_j - v_j|^{\frac{p(p-q-b)}{p-q}}\right)^{2(p-q)} \leq \|y\|_p^{2pb} \|y-v\|_p^{2p(p-q-b)} \\[1ex]
    \begin{aligned}
        & \left(\sum_{i,j} \vert x_i\vert^{\frac{pa}{p-q}}\vert x_i - u_i\vert^{\frac{p(p-q-a)}{p-q}} \vert y_j\vert^{\frac{pb}{p-q}}\vert y_j - v_j\vert^{\frac{p(p-q-b)}{p-q}} \right)^{2(p-q)} \\
        &\leq \|x\|_p^{2pa} \|x-u\|_p^{2p(p-q-a)} \|y\|_p^{2pb} \|y-v\|_p^{2p(p-q-b)}
    \end{aligned}
\end{gathered}
\right\}
\end{align*}

\begin{remark}
    Note that these constraints also include terms of the form $|x_i - u_i|^{\frac{p(p-q-a)}{p-q}}, |y_j - v_j|^{\frac{p(p-q-b)}{p-q}}$, and these will also appear in the analysis. These are all low degree polynomials in terms of $\vert x_i - u_i \vert^{\frac 1 {\prod_{q=1}^{p-1} (p-q)}},\, \vert y_j - v_j\vert^{\frac 1 {\prod_{q=1}^{p-1} (p-q)}}$ with degree at most $O_p(1)$. Formally, we write the pseudoexpectation with variables $\vert x_i - u_i\vert^{\frac 1 {\prod_{q=1}^{p-1} (p-q)}}, \vert y_j - v_j\vert^{\frac 1 {\prod_{q=1}^{p-1} (p-q)}}$ so these expressions are well defined.
\end{remark}

\begin{mdframed}
  \begin{algorithm}[$\ell_p$ LRA Rank $1$ PTAS]
    \label{algo:lra-ptas}\mbox{}
    \begin{description}
    \item[Input:] Matrix $A \in \mathbb{R}^{n \times m}$ and target accuracy $0<\epsilon<1$.
    
    \item[Operations:]\mbox{}
    \begin{enumerate}
        \item If $A$ is a rank-$1$ matrix output the rank-$1$ factorization $\hat{u}, \hat{v}$.
        \item Let $x,y$ be an $O(\OPT)$ LRA solution for $A$ from Theorem 14 in~\cite{ban_etal_soda2019_ptas_lp_lra} and let the objective value be $\widetilde{\OPT}$. 
        \item Rescale $x,y$ such that $\norm{x} = \norm{y}$ and then round the entries to the nearest value in a $\frac{\eps \widetilde{\OPT}}{\left(nm \cdot 2^B\right)^{O_p(1)}}$-fineness grid.
        \item Let $\calS_1$ be a $\frac{\eps^p \widetilde{\OPT}^p}{\left(nm \cdot 2^B\right)^{O_p(1)}}$-fineness grid of $[0, nm \cdot 2^{O_p(B)}]$.
        \item Let $\calS_2$ be a $\frac{\eps^p \widetilde{\OPT}^p}{\left(nm \cdot 2^B\right)^{O_p(1)}}$-fineness grid of $\left[0, O_p(1) \cdot \frac{\widetilde{\OPT}}{\norm{x}_p^p}\right]$
        \item For all $0 \leq k \leq p$ and each $s_{u,k} , s_{v,k} \in \calS_1$, $d_u, d_v \in \calS_2$, and set $\calT$ of size at most $O_p((1/\eps)^{O(1)})$:
        \begin{enumerate}
            \item Let $\Sigma$ be a $\frac{\eps \cdot \widetilde{\OPT}}{\left(nm \cdot 2^B\right)^{O_p(1)}}$-fineness grid of $\left[ - O(1) \cdot \norm{A}_p^{1/2}, O(1) \cdot \norm{A}_p^{1/2}\right]$.
            \item Let $\mu$ be a $O_p((1/\eps)^{O(1)})$-degree sum-of-squares pseudo-distribution over $\left(\Sigma\right)^{n+m}$ with variables $u_i$ and $v_j$, along with additional variables $\vert u_i - x_i\vert^{\frac 1 {\prod_{q=1}^{p-1} (p-q)}}$ and $ \vert v_j - y_j\vert^{\frac 1 {\prod_{q=1}^{p-1} (p-q)}}$, such that $\pE_{\mu}$ optimizes 
            \begin{align}
                &\min_{\pE_{\mu}} \quad \pE_{\mu} \sum_{i\in[n],j\in[m]} (A_{ij} - u_i v_j)^p \notag \\
                &\text{s.t.\ } \pE_{\mu} \text{ satisfies } \calA(s,d) \notag
            \end{align}
            \item Let $\hat{v}_\calT$ be a draw from the local distribution $\{v_{\calT}\}$. Let $\mu_\calT$ be the pseudodistribution obtained by conditioning on $\{v_j = \hat{v}_j\}_{j \in \calT}$.
            \item Let $\hat{u}, \hat{v}$ be the solution where $\hat{u}_i, \hat{v}_j$ are sampled independently from their local distributions in $\mu_\calT$.        \end{enumerate}
        \item Let $\hat{u}, \hat{v}$ be the LRA solution from the additive approximation in~\cref{thm:lra-additive} with $\eps' = \eps /C_p$ for an absolute constant $C_p$ depending only on $p$.
    \end{enumerate}

    \item[Output:] The embedding $\hat{u}, \hat{v}$ with the lowest LRA objective value over all iterations of the loop and the additive approximation solution.
    \end{description}
  \end{algorithm}
\end{mdframed}

\begin{remark}
    Although the above algorithm utilizes sum-of-squares pseudodistributions, we note that the Sherali-Adams linear programming hierarchy alongside global linear constraints would suffice. We use sum-of-squares pseudodistributions and SDPs for ease of exposition.
\end{remark}

Note that throughout the algorithm, we use a constant approximation to the optimal objective value as a proxy for the true value $\OPT$. In the rest of this section, we will assume that we have the true objective value here, and it is easy to verify that a constant factor approximation in fact suffices.

\paragraph{Discretization.}
We show that restricting LRA solutions to have entries in $\Sigma$ does not increase the cost of the optimal solution by too much. In particular, we show the following lemma, which bounds the error introduced by rounding LRA solutions to the nearest entries in $\Sigma$. The proof is straightforward, and we defer it to the appendix. 

\begin{lemma}
\label{lem:discretization-error-bound}
    Let $u,v$ be a constant factor optimal solution to the $\ell_p$ LRA approximation problem on $A$ such that $\norm{u}, \norm{v} \leq O(1) \cdot \norm{A}_p^{1/2}$. Let $\gamma \leq \norm{A}_p^{1/2}$ and let$\Sigma$ be a $\gamma$-fineness grid of 
    \[\left[ - O(1) \cdot \norm{A}_p^{1/2}, O(1) \cdot \norm{A}_p^{1/2}\right]\,.\] 
    Let $\tilde{u}, \tilde{v}$ denote the result of rounding the entries of $u,v$ to the nearest entry in $\Sigma$. Then we have that
    \[ \sum_{i \in [n], j \in [m]} \left(A_{ij} - u_i v_j\right)^p \leq \sum_{i \in [n], j \in [m]} \left(A_{ij} - \tilde{u}_i \tilde{v}_j\right)^p +  O_p(1) \cdot nm \cdot \gamma \cdot \norm{A}_p^{\frac{2p - 1}{2}}\,.\]
\end{lemma}

Note that by~\cref{lem:discretization-error-bound} if $A$ has entries with bit complexity at most $B$, then it suffices to let $\Sigma$ be a $\frac{\OPT}{\left(nm \cdot 2^B\right)^{O_p(1)}}$-fineness grid of \[\left[ - O(1) \cdot \norm{A}_p^{1/2}, O(1) \cdot \norm{A}_p^{1/2}\right]\,.\] 
If $\OPT = 0$ then we have already outputted a cost $0$ solution in step $1$ of~\cref{algo:lra-ptas}, and thus in order to show bounds on the size of $\Sigma$ it suffices to show lower bounds on $\OPT$ when $\OPT \neq 0$. We prove the following lower bound on the value of $\OPT$ when it is non-zero:
\begin{lemma}
\label{lem:lb-opt}
    Let $p$ be an even integer and let $A$ be a $\mathbb{R}^{n \times m}$ matrix with entries with bit complexity $B$. Let $\OPT$ be the optimal objective value of the $\ell_p$ LRA value and let $\OPT \neq 0$. Then we have that 
    \[ \OPT \geq (nm)^{-O_p(1)} \cdot 2^{-O(B)}\,.\]
\end{lemma}
\begin{proof}[Proof of~\cref{lem:lb-opt}]
    Let $\OPT \neq 0$. Note that $\OPT = 0$ for the $\ell_p$ LRA problem if and only if $\OPT = 0$ for the $\ell_2$ LRA problem; in both cases this occurs if and only $A$ is rank-$1$. Thus, we can assume that the optimal value for the $\ell_2$ LRA problem, is also non-zero. Furthermore, since $\ell_2$ and $\ell_p$ norms are equivalent up to multiplicative factor of $n^{O_p(1)}$, it suffices to show the statement for $p=2$. 

    Let $\hat{A}$ be the best rank-$1$ approximation to $A$, and we want to show that $\norm{A - \hat{A}}_F^2 \geq 2^{-O(B)}$ when $A$ is not rank $1$. Note that since $A$ is not rank $1$, there must exist a $2 \times 2$ minor of $A$ with non-zero determinant. Let the associated submatrix be $M$. Note that $\det(M) \neq 0$, so we must have that $\abs{\det(M)} \geq 2^{-O(B)}$, since the entries of $M$ have bit complexity at most $B$. 

    Consider $\norm{M - \hat{M}}_F^2$, or the error in the best rank $1$ approximation to $M$. Note that this is just $\sigma_2(M)^2$ and we have that $\abs{\sigma_1(M)} \cdot \abs{\sigma_2(M)} = \abs{\det(M)} \geq 2^{-O(B)}$. Furthermore, note that $\abs{\sigma_1(M)} \leq \norm{M}_F$, which is at most $nm\cdot 2^{O(B)}$. Thus, we can conclude that $\abs{\sigma_2(M)} \geq \frac 1 {nm} \cdot 2^{-O(B)}$. Finally, to conclude we note that any rank $1$ approximation to $A$ gives a rank $1$ approximation to $M$ with $\norm{M - \hat{M}}_F^2 \leq \norm{A - \hat{A}}_F^2$. Thus, we have that $\norm{A - \hat{A}}_F^2 \geq \frac 1 {n^2m^2} 2^{-O(B)}$.
\end{proof}

\subsection{LRA Identifiability}

We prove the following ``strong identifiability'' inequality for any two LRA solutions. These identifiability inequalities show that we can compute an SDP solution value with $(1+\eps)\cdot \OPT$\footnote{This is true in the case where $\OPT$ is sufficiently small; otherwise, the additive approximation in~\cref{thm:lra-additive} already gives $\eps \cdot \OPT$ additive error.} where $\norm{x-u}_p^p + \norm{y-v}_p^p\leq \frac \OPT {\norm{x}_p^p}$, which will be crucial in analyzing our rounding algorithm. 

\begin{lemma}[Strong Identifiability for Rank $1$ $\ell_p$ LRA]
\label{lem:gen-strong-identifiability}
    Let $p$ be an even integer, $x,u \in \mathbb{R}^n$, and $y,v \in \mathbb{R}^m$. Then
    \[ \min_c \left(\norm{u+cx}_p^p \norm{v-y/c}_p^p + \norm{u-cx}_p^p \norm{v+y/c}_p^p\right) \leq O(1) \cdot \norm{uv^\top - xy^\top}_p^p\,.\]
\end{lemma}

Note that when $\OPT$ is small we also get stronger bounds via the following lemma:
\begin{lemma}
\label{cor:strong-identifiability}
    Let $x,y$ be an $O\left(\OPT\right)$ value $\ell_p$ LRA solution for $A$ such that $\norm{x} = \norm{y}$ and let $\OPT \leq C_p \norm{A}_p^p$ for some sufficiently small constant $C_p$ depending only on $p$. Then there exists $u,v$ an optimal LRA solution for $A$ such that 
    \[ \max\{\norm{u-x}_p^p, \norm{v-y}_p^p\} \leq O\left( \frac{\OPT}{\norm{x}_p^p} \right)\,.\]
\end{lemma}

We first prove~\cref{cor:strong-identifiability} and then return to the proof of~\cref{lem:gen-strong-identifiability}.

\begin{proof}[Proof of~\cref{cor:strong-identifiability}]
    Applying~\cref{lem:gen-strong-identifiability} we have that for an optimal LRA solution $u,v$ 
    \[ \min_c \left(\norm{u+cx}_p^p \norm{v-y/c}_p^p + \norm{u-cx}_p^p \norm{v+y/c}_p^p\right) \leq O(1) \cdot \norm{uv^\top - xy^\top}_p^p\,.\]
    Note that $cu, v/c$ are also an optimal LRA solution themselves, so WLOG let the optimizing $c$ be $1$. Furthermore, since both $x,y$ and $u,v$ have objective value $O(\OPT)$ we have that 
    \[ \norm{uv^\top - xy^\top}_p^p \leq \norm{(uv^\top- A) + (A - xy^\top)}_p^p \leq O_p(1) \left( \norm{A - uv^\top}_p^p + \norm{A - xy^\top}_p^p \right) \leq O_p(\OPT)\,.\]
    Combining this with the identifiability inequality, we thus have that
    \begin{align}
        \norm{u+x}_p^p \norm{v-y}_p^p + \norm{u-x}_p^p \norm{v+y}_p^p \leq O_p(\OPT) \label{eq:opt-identifiability}
    \end{align}
    Note that this implies that either $\norm{u+x}_p^p \leq O_p \left(\sqrt{\OPT}\right)$ or $\norm{v-y}_p^p \leq O_p \left(\sqrt{\OPT}\right)$. WLOG let $\norm{y-v}_p^p \leq O_p \left(\sqrt{\OPT}\right)$ by symmetry between $x,y$ and $u,v$ and the fact that $-u,-v$ are also an optimal LRA solution. We will show that this implies that $\norm{y+v} \geq \Omega_p(1) \norm{y}$ and thus $\norm{u-x}_p^p \leq O_p(1) \frac \OPT {\norm{y}_p^p}$.

    We first observe that by the Almost Triangle Inequality, we have that 
    \[ \norm{A}_p^p \leq 2^{2p-1} \norm{A-xy^\top}_p^p + 2^{2p-1} \norm{x}_p^{p}\norm{y}_p^p = 2^{2p-1} \norm{A-xy^\top}_p^p + 2^{2p-1} \norm{y}_p^{2p}\,.\]
    Rearranging, we have that 
    \[ \norm{y}_p^{2p} \geq 2^{-2p+1} \norm{A}_p^p - \norm{A-xy^\top}_p^p = 2^{-2p+1} \norm{A}_p^p - O\left(\OPT\right)\,.\]
    Observe that by assumption we have that $\OPT \leq C_p \norm{A}_p^p$ so $\norm{y}_p^{2p} \geq \left(2^{-2p+1} - C'\right) \norm{A}_p^p$ for a new constant $C'$ depending only on $p$. Furthermore, again by Almost Triangle Inequality we have that
    \begin{align*}
        \norm{y+v}_p^p \geq 2^{-2p+1} \norm{2y}_p^p - \norm{y-v}_p^p\,.
    \end{align*}
    Since $\norm{y-v}_p^p \leq O_p(\sqrt{\OPT}) \leq O_p\left( \sqrt{C_p \norm{A}_p^p}\right)$ and $\norm{v}_p^{p} \geq \sqrt{2^{-2p+1} - C'} \cdot \sqrt{\norm{A}_p^p}$, substituting gives us that
    \[ \norm{y+v}_p^p \geq \left(2^{-p+1} \sqrt{2^{-2p+1} - C'} - \sqrt{C} \right)\cdot \norm{A}_p^{p/2}\,.\]
    To conclude, we note that $\norm{A}_p^{p} \geq \Omega_p(1) \norm{x}_p^p \norm{y}_p^p = O_p(1) \norm{y}_p^{2p}$ and thus for sufficiently small $C_p$, $C'$ we have that
    \[ \norm{y+v}_p^p \geq \Omega_p(1) \norm{y}_p^p\,.\]
    Combining this lower bound on $\norm{y+v}_p$ with~\cref{eq:opt-identifiability} gives us that 
    \[\norm{u-x}_p^p \leq O_p(1) \cdot \frac \OPT {\norm{y}_p^p}\,.\]
    Using the above upper bound on $\norm{y-v}$, the same line of reasoning also gives that $\norm{u+x}_p^p \geq \Omega_p(1) \norm{x}_p^p$ and thus
    \[\norm{v-y}_p^p \leq O_p(1) \cdot \frac \OPT {\norm{x}_p^p}\,.\]

\end{proof}

We now return to the proof of~\cref{lem:gen-strong-identifiability}. We will need the following fact about sketching.

\begin{fact}[Sketching]
\label{fact:sketch}
  Let $L \subseteq \R^n$ be a subspace of dimension $k$.
  And let $p > 2$.
  There exists a random matrix $S \in \R^{m \times n}$ with $m = \poly(k,1/\eps)$ such that with probability at least $2/3$, for all $v \in L$,
  \[
  (1-\eps) \|v\|_p \leq \|Sv \|_p \leq (1+\eps) \|v\|_p \, .
  \]
\end{fact}

\begin{proof}[Proof of~\cref{lem:gen-strong-identifiability}]
    Note that by taking $p=O(1)$ and $\eps = 1/2$ we have that there exists random matrices $S_1 \in \mathbb{R}^{O(1) \cdot n}$ and $S_2 \in \mathbb{R}^{O(1) \times m}$ such that with probability at least $2/3$
    \begin{align*}
        \frac{1}{2} \norm{S_1u+cS_1x}_p&\leq \norm{u+cx}_p \leq \frac 3 2 \norm{S_1u+cS_1x}_p \\
        \frac{1}{2} \norm{S_1u-cS_1x}_p&\leq \norm{u-cx}_p \leq \frac 3 2 \norm{S_1u-cS_1x}_p \\
        \frac{1}{2} \norm{S_2v+S_2y/c}_p&\leq \norm{v+y/c}_p \leq \frac 3 2 \norm{S_2v+S_2y/c}_p \\
        \frac{1}{2} \norm{S_2v-S_2y/c}_p&\leq \norm{v-y/c}_p \leq \frac 3 2 \norm{S_2v-S_2y/c}_p 
    \end{align*}
    and $\norm{uv^\top - xy^\top}_p = \Theta \left( \norm{(S_1 u)(S_2 v)^\top - (S_1 x)(S_2 y)^\top}_p \right)$. Thus, it suffices to prove the inequality for $S_1u, c S_1 x, S_2 v, S_2 y/c$ or in the setting where $u, cx, v, y/c$ are constant dimensional vectors. In this setting, up to multiplicative constants, we have that $\ell_p$ and $\ell_2$ norms are equivalent, and thus the inequality reduces to showing that 
    \[\min_c \left(\norm{u+cx}_2^2 \norm{v-y/c}_2^2 + \norm{u-cx}_2^2 \norm{v+y/c}_2^2\right) \leq O(1) \cdot \norm{uv^\top - xy^\top}_F^2\,.\]
    In particular, we will show that when $c = \norm{u}/\norm{x}$ then
    \[ \norm{u+cx}_2^2 \norm{v-y/c}_2^2 + \norm{u-cx}_2^2 \norm{v+y/c}_2^2 \leq O(1) \cdot \norm{uv^\top - xy^\top}_F^2\,.\]
    Observe that 
    \[ \left(u+cx\right)\left(v-y/c\right)^\top + \left(u-cx\right)\left(v+y/c\right)^\top = 2 \left( uv^\top - xy^\top \right)\,.\]
    Thus, $\frac{1}{2}\norm{\left(u+cx\right)\left(v-y/c\right)^\top + \left(u-cx\right)\left(v+y/c\right)^\top} = \norm{uv^\top - xy^\top}$. On the other hand,
    \begin{align*}
        &\norm{\left(u+cx\right)\left(v-y/c\right)^\top + \left(u-cx\right)\left(v+y/c\right)^\top}_F^2 \\
        &\qquad= \norm{u+cx}_2^2 \norm{v-y/c}_2^2 + \norm{u-cx}_2^2 \norm{v+y/c}_2^2 + 2 \left( \langle u+cx, u-cx\rangle \cdot \langle v-y/c, v+y/c\rangle\right) \\
        &\qquad= \norm{u+cx}_2^2 \norm{v-y/c}_2^2 + \norm{u-cx}_2^2 \norm{v+y/c}_2^2 + 2 \left(\norm{u}_2^2 - c^2 \norm{x}_2^2 \right) \left( \norm{v}_2^2 - \frac{1}{c^2} \norm{y}_2^2\right)\,.
    \end{align*}
    Picking $c = \norm{u}/\norm{x}$ makes the last term $0$ and thus 
    \[ \norm{\left(u+cx\right)\left(v-y/c\right)^\top + \left(u-cx\right)\left(v+y/c\right)^\top}_F^2 = \norm{u+cx}_2^2 \norm{v-y/c}_2^2 + \norm{u-cx}_2^2 \norm{v+y/c}_2^2\,,\]
    which completes the proof.
\end{proof}

Finally, we note that~\cref{cor:strong-identifiability} implies that the following bounds hold, which mean that our pseudodistribution constraints will be satisfied by an optimal LRA solution.
\begin{lemma}
\label{lem:addtl-sdp-constraints-satisfied}
    Let $p$ be an even integer and let $a,b,q$ be integers such that $1\leq q < p$ and $a+b \leq p-q$. Let $x,y$ be an $O\left(\OPT\right)$ value $\ell_p$ LRA solution for $A$ such that $\norm{x} = \norm{y}$ and let $\OPT \leq C_p \norm{A}_p^p$ for some sufficiently small constant $C_p$ depending only on $p$. Then there exists $u,v$ an optimal LRA solution for $A$ such that
    \begin{align*}
        \left(\sum_i |x_i|^{\frac{pa}{p-q}} \cdot |x_i - u_i|^{\frac{p(p-q-a)}{p-q}}\right)^{2(p-q)} &\leq \|x\|_p^{2pa} \|x-u\|_p^{2p(p-q-a)} \\[1ex]
        \left(\sum_j |y_j|^{\frac{pb}{p-q}} \cdot |y_j - v_j|^{\frac{p(p-q-b)}{p-q}}\right)^{2(p-q)} &\leq \|y\|_p^{2pb} \|y-v\|_p^{2p(p-q-b)} \\[1ex]
        \left(\sum_{i,j} \vert x_i \vert^{\frac{pa}{p-q}}\cdot \vert x_i - u_i\vert^{\frac{p(p-q-a)}{p-q}} \cdot \vert y_j \vert^{\frac{pb}{p-q}}\cdot \vert y_j - v_j\vert^{\frac{p(p-q-b)}{p-q}} \right)^{2(p-q)} &\leq \|x\|_p^{2pa} \|x-u\|_p^{2p(p-q-a)} \|y\|_p^{2pb} \|y-v\|_p^{2p(p-q-b)}\,.
    \end{align*}
\end{lemma}
\begin{proof}
    All inequalities follow via an application of H\"olders inequality with the applicable conjugates. Using H\"olders inequality with $\frac{p-q}{a}$ and $\frac{p-q}{p-q-a}$ yields the first inequality, and the second follows via H\"olders with $\frac{p-q}{b}$ and $\frac{p-q}{p-q-b}$. Finally, we note that the third inequality is simply the product of the first two.
\end{proof}

\subsection{Rounding Error}

In this section, we will show that global correlation rounding produces a solution with value at most $\eps\cdot \OPT$ larger than the SDP solution. We focus on the case where $\OPT$ is sufficiently small relative to $\norm{A}_p^p$. Otherwise, the additive approximation in~\cref{thm:lra-additive} already produces a solution $\hat{u}, \hat{v}$ which achieves $\eps \cdot \OPT$ relative error.

\begin{lemma}[Relative Error for SDP Rounding]
\label{lem:ptas-rounding}
    Let $p$ be a positive even integer and let $A$ be a matrix such that $\OPT \leq C_p \norm{A}_p^p$ for a sufficiently small constant $C_p$ depending only on $p$. Let $\OPT_{\mathrm{SDP}}$ be the objective value of the SDP in~\cref{algo:lra-ptas} for some $s_{u,k}, s_{v,k} \in \mathcal{S}_1$ and $d_u, d_v \in \calS_2$. Then global correlation rounding with $(1/\eps)^{O_p(1)}$ rounds of conditioning produces a solution with expected value at most $\OPT_{\mathrm{SDP}} + \eps \cdot \OPT$.
\end{lemma}

In the proof of~\cref{lem:ptas-rounding}, we will need the following lemma regarding the potential alignment of parts of our objective.
\begin{lemma}
\label{lem:lra-decomp-potential aligned}
    Let $p$ be a positive even integer and $q \in \mathbb{N}$ be such that $1 \leq q < p$ and let $A$ be a matrix such that $\OPT \leq C_p \norm{A}_p^p$ for a sufficiently small constant $C_p$ depending only on $p$. Furthermore, let $a,b \in \mathbb{N}$ be such that $a, b < p-q$ and $a+b \leq p-q$. Let $r_i = x_i^a (x_i-u_i)^{p-q-a}$, and let $r_i^{(0)}$ denote its distribution under the initial pseudodistribution. Then pseudodistribution solutions to~\cref{algo:lra-ptas} are $\left(\eps\cdot \OPT,\, \eps^{O_p(1)} \cdot \norm{x}^{\frac {p(2a - (p-q))} {p-q}} \cdot \OPT^{\frac{p-q-a}{p-q}} \right)$-potential aligned with the function
    \[ \sum_{i,j} \left(A_{ij} - x_i y_j\right)^q x_i^a \left(x_i - u_i\right)^{p-q-a} y_j^b \left(y_j - v_j\right)^{p-q-b}\]
    and potential function
    \[ \Phi_{a,q}(u) \vcentcolon= \sum_{i \in [n]} \left( \phi_{r_i^{(0)}}^{(p)} \right)^{\frac{p-1}{p-q}-1} \phi_{r_i}^{(p)}\,.\]
\end{lemma}
We will first prove~\cref{lem:ptas-rounding} and then return to the proof of~\cref{lem:lra-decomp-potential aligned}.
\begin{proof}[Proof of~\cref{lem:ptas-rounding}]
    Let $x,y$ be the constant factor approximation LRA solutions. We have that 
    \begin{align*}
        \sum_{i \in [n], j \in [m]} (A_{ij} - u_i v_j)^p &= \sum_{i \in [n], j \in [m]} (A_{ij} \pm x_i y_j - u_i v_j)^p \\
        &= \sum_{q =0}^p \binom{p}{q} \sum_{i \in [n], j \in [m]} \left(A_{ij} - x_i y_j\right)^q \left(x_i y_j - u_i v_j\right)^{p-q}. 
    \end{align*}
    Note that the term where $q=p$ incurs no rounding error because it is a constant function independent of $u,v$ while we can rewrite the term where $q=0$ as 
    \[\sum_{i \in [n], j \in [m]} \left(x_i y_j - u_i v_j\right)^{p} = \sum_{k=0}^p \binom{p}{k} \left(-1\right)^k \left(\sum_i x_i^{p-k} u_i^k \right)\left(\sum_j y_j^{p-k} v_j^k \right)\,.\]
    This function has the same expectation under joint and product distributions due to the SDP constraints, and thus this term incurs no error. It remains to analyze the remaining terms for $1 \leq q < p$. Note that we can decompose the polynomial $\left(x_i y_j - u_i v_j\right)^{p-q}$ using that
    \[ x_i y_j - u_i v_j = -(x_i - u_i)(y_j - v_j) + x_i(y_j - v_j) + y_j (x_i - u_i)\,,\]
    and thus,
    \begin{align*}
        \left(x_i y_j - u_i v_j\right)^{p-q} = \sum_{a=0}^{p-q} \sum_{b=0}^{p-q-a} \left(-1\right)^{p-q-a-b} \binom{p-q}{a} \binom{p-q-a}{b} x_i^a \left(x_i - u_i\right)^{p-q-a} y_j^b \left(y_j - v_j\right)^{p-q-b}\,.
    \end{align*}
    Note that the polynomial $x_i^{p-q} \left(y_j - v_j\right)^{p-q}$ incurs no rounding error since it is only a function of the variable $v_j$ and constants $x_i, y_j$, and similarly $\left(x_i - u_i\right)^{p-q} y_j^{p-q}$ also has no rounding error. 

    By~\cref{lem:lra-decomp-potential aligned} we have that the remaining polynomials 
    \[\sum_{i\in[n], j \in [m]} \left(-1\right)^{p-q-a-b} \binom{p-q}{a} \binom{p-q-a}{b} \left(A_{ij} - x_i y_j\right)^q x_i^a \left(x_i - u_i\right)^{p-q-a} y_j^b \left(y_j - v_j\right)^{p-q-b}\]
    are $\left(\eps\cdot \OPT,\, \eps^{O_p(1)} \cdot \norm{x}^{\frac {p(2a - (p-q))} {p-q}} \cdot \OPT^{\frac{p-q-a}{p-q}} \right)$-potential aligned for the potentials
    \[ \Phi_{a,q}(u) \vcentcolon= \sum_{i \in [n]} \left( \phi_{r_i^{(0)}}^{(p)} \right)^{\frac{p-1}{p-q}-1} \phi_{r_i}^{(p)}\,.\]
    Note that by~\cref{fact:gen-initial-potential-bound} and convexity the initial potential is bounded by 
    \[ \Phi_{a,q}(u^{(0)}) = \sum_{i \in [n]} \left( \phi_{r_i^{(0)}}^{(p)} \right)^{\frac{p-1}{p-q}} \leq \sum_{i \in [n]} \left( \pE \left[ \left\vert x_i^a (u_i - x_i)^{p-q-a}\right\vert^{\frac p {p-1}}\right] \right)^{\frac {p-1}{p-q}} \leq \sum_{i \in [n]} \pE \left[ \left( \vert x_i\vert^a \cdot \vert u_i - x_i\vert^{p-q-a}\right)^{\frac{p} {p-q}} \right]\,.\]
    Note that by applying~\cref{fact:pseudo-expectation-holder} and our SDP constraints we have that
    \begin{align*}
        \pE \left[\sum_i \vert x_i\vert^{\frac{pa}{p-q}}\cdot \vert x_i - u_i\vert^{\frac{p(p-q-a)}{p-q}} \right] &\leq \left(\pE \left[\left(\sum_i \vert x_i\vert^{\frac{pa}{p-q}}\cdot \vert x_i - u_i\vert^{\frac{p(p-q-a)}{p-q}}\right)^{2(p-q)} \right]\right)^{\frac 1 {2(p-q)}} \\
        &\leq \norm{x}_p^{\frac{pa}{p-q}} \left( \pE \left[ \norm{x-u}^{2p(p-q-a)}\right]\right)^{\frac 1 {2(p-q)}}\,.
    \end{align*}
    By our SDP constraints we also have that 
    \[\pE \left[ \norm{x-u}^{2p(p-q-a)}\right] = \pE \left[ \norm{x-u}^p \right]^{2(p-q-a)} \leq \frac{\OPT^{2(p-q-a)}}{\norm{x}^{2p(p-q-a)}}\]
    and thus the initial potential satisfies 
    \[ \Phi_{a,q}(u^{(0)}) \leq \norm{x}_p^{\frac{p(2a-(p-q))}{p-q}}\cdot \OPT^{\frac{p-q-a}{p-q}}\,.\]
    Thus, applying~\cref{lem:gcr_linear_combo} we have that there exists a set $\calT$ of size 
    \[ \sum_{q = 1}^{p-1} \sum_{a=0}^{p-q} \sum_{b=0}^{p-q-a} \frac{\norm{x}_p^{\frac{p(2a-(p-q))}{p-q}}\cdot \OPT^{\frac{p-q-a}{p-q}}}{\eps^{O_p(1)} \cdot \norm{x}^{\frac {p(2a - (p-q))} {p-q}} \cdot \OPT^{\frac{p-q-a}{p-q}}} = O_p(1)\]
    such that after conditioning on $\calT$ the expected rounding error is at most $\eps \cdot \OPT$.
\end{proof}
We now return to the proof of~\cref{lem:lra-decomp-potential aligned}.
\begin{proof}[Proof of~\cref{lem:lra-decomp-potential aligned}]
    Applying~\cref{lem:general-p-rounding-err} with $x_i^a \left(x_i - u_i\right)^{p-q-a}$ and $y_j^b \left(y_j - v_j\right)^{p-q-b}$ we have that if the expected rounding error on this function is at least $\eps\cdot \OPT$ then\footnote{Note that while~\cref{lem:general-p-rounding-err} considers conditioning on $y_j^b \left(y_j - v_j\right)^{p-q-b}$ conditioning on $v_j$ only further reduces the potential function since $y_j^b \left(y_j - v_j\right)^{p-q-b}$ is a deterministic function of $v_j$.} 
    \begin{align*}
        \eps\cdot \OPT &\leq O(1) \cdot \left(\sum_j \E_{\{\mu\}}\pE [\vert y_j\vert^{\frac{pb}{p-q}}\vert y_j - v_j\vert^{\frac{p(p-q-b)}{p-q}}] \right)^{\frac{p-q}{2(p-1)^2}}  \\
        &\qquad\Bigg( \sum_{ij} \E_{\{\mu\}} \Big[\pE [\vert x_i\vert^{\frac{pa}{p-q}}\vert x_i - u_i\vert^{\frac{p(p-q-a)}{p-q}} \vert y_j\vert^{\frac{pb}{p-q}}\vert y_j - v_j\vert^{\frac{p(p-q-b)}{p-q}}]  \\
        &\qquad \qquad+ \pE [\vert x_i\vert^{\frac{pa}{p-q}}\vert x_i - u_i\vert^{\frac{p(p-q-a)}{p-q}}] \pE [\vert y_j\vert^{\frac{pb}{p-q}}\vert y_j - v_j\vert^{\frac{p(p-q-b)}{p-q}}] \Big] \Bigg)^{\frac{(p-q)(p-2)(2p-1)}{2p(p-1)^2}}  \\
        &\qquad \norm{A - xy^\top}_p^q \cdot \left( \E_{\{\mu\}}\E_{j \sim \psi} \left( \Phi_{a,q}(u) - \E_{v_j} \Phi_{a,q}(u \vert v_j)\right)\right)^{\frac{p-q}{2(p-1)^2}}
    \end{align*}
    Note that $\norm{A - xy^\top}_p^q \leq O\left(\OPT^{\frac{q}{p}}\right)$ since $x,y$ are a constant factor approximation LRA solution. Thus, it suffices to bound terms of the following form in expectation over the randomness in $\{\mu\}$:
    \begin{gather*}
        \sum_i \pE \left[\vert x_i\vert^{\frac{pa}{p-q}}\vert x_i - u_i\vert^{\frac{p(p-q-a)}{p-q}}\right]\,, \\
        \sum_j \pE \left[\vert y_j\vert^{\frac{pb}{p-q}}\vert y_j - v_j\vert^{\frac{p(p-q-b)}{p-q}} \right]\,, \\
        \sum_{ij} \pE \left[\vert x_i\vert^{\frac{pa}{p-q}}\vert x_i - u_i\vert^{\frac{p(p-q-a)}{p-q}} \vert y_j\vert^{\frac{pb}{p-q}} \vert y_j - v_j\vert^{\frac{p(p-q-b)}{p-q}}\right] + \pE \left[\vert x_i\vert^{\frac{pa}{p-q}}\vert x_i - u_i\vert^{\frac{p(p-q-a)}{p-q}}\right] \pE \left[\vert y_j\vert^{\frac{pb}{p-q}}\vert y_j - v_j\vert^{\frac{p(p-q-b)}{p-q}}\right] \,.
    \end{gather*}
    Note that $\E_{\{\mu\}}\pE = \pE_0$, or the pseudoexpectation at time $0$, so it suffices just to consider expectations over the initial pseudodistribution on all but the last term. Note that by applying~\cref{fact:pseudo-expectation-holder} and our SDP constraints we have that
    \begin{align*}
        \pE_0 \left[\sum_i \vert x_i\vert^{\frac{pa}{p-q}}\vert x_i - u_i\vert^{\frac{p(p-q-a)}{p-q}} \right] &\leq \left(\pE_0 \left[\left(\sum_i \vert x_i\vert^{\frac{pa}{p-q}} \vert x_i - u_i\vert^{\frac{p(p-q-a)}{p-q}}\right)^{2(p-q)} \right]\right)^{\frac 1 {2(p-q)}} \\
        &\leq \norm{x}_p^{\frac{pa}{p-q}} \left( \pE_0 \left[ \norm{x-u}^{2p(p-q-a)}\right]\right)^{\frac 1 {2(p-q)}}\,.
    \end{align*}
    Similarly,
    \begin{align*}
        \sum_j \E_{\{\mu\}}\pE [\vert y_j\vert^{\frac{pb}{p-q}} \vert y_j - v_j\vert^{\frac{p(p-q-b)}{p-q}}] \leq \norm{y}_p^{\frac{pb}{p-q}} \left( \pE_0 \left[ \norm{y-v}^{2p(p-q-b)} \right] \right)^{\frac 1 {2(p-q)}}\,.
    \end{align*}
    Finally, for the last term, note that by our constraints we have that
    \begin{align*}
        &\sum_{ij} \E_{\{\mu\}} \pE [\vert x_i\vert^{\frac{pa}{p-q}} \vert x_i - u_i\vert^{\frac{p(p-q-a)}{p-q}} \vert y_j\vert^{\frac{pb}{p-q}} \vert y_j - v_j\vert^{\frac{p(p-q-b)}{p-q}}] \\
        &\qquad\leq \norm{x}_p^{\frac{pa}{p-q}} \norm{y}^{\frac{pb}{p-q}} \left( \pE_0 \norm{x-u}^{2p(p-q-a)} \norm{y-v}^{2p(p-q-b)} \right)^{\frac 1 {2(p-q)}}\,,
    \end{align*}
    and
    \begin{align*}
        &\sum_{ij} \pE [\vert x_i\vert^{\frac{pa}{p-q}} \vert x_i - u_i\vert^{\frac{p(p-q-a)}{p-q}}] \pE [\vert y_j\vert^{\frac{pb}{p-q}}\vert y_j - v_j\vert^{\frac{p(p-q-b)}{p-q}}] \\
        &\qquad \leq \norm{x}_p^{\frac{pa}{p-q}} \norm{y}^{\frac{pb}{p-q}} \left( \pE \norm{x-u}^{2p(p-q-a)} \pE \norm{y-v}^{2p(p-q-b)} \right)^{\frac 1 {2(p-q)}}
    \end{align*}
    Using that we imposed equality constraints on $\norm{x-u}, \norm{y-v}$ we have that
    \begin{align*}
        \pE \norm{x-u}^{p(p-q-a)} \pE \norm{y-v}^{p(p-q-b)} = \pE \norm{x-u}^{p(p-q-a)} \norm{y-v}^{p(p-q-b)}\,,
    \end{align*}
    and thus
    \begin{align*}
        \sum_{ij} \E_{\{\mu\}} \Bigg[&\pE [\vert x_i\vert^{\frac{pa}{p-q}} \vert x_i - u_i\vert^{\frac{p(p-q-a)}{p-q}} \vert y_j \vert^{\frac{pb}{p-q}} \vert y_j - v_j\vert^{\frac{p(p-q-b)}{p-q}}] + \pE [\vert x_i\vert^{\frac{pa}{p-q}}\vert x_i - u_i\vert^{\frac{p(p-q-a)}{p-q}}] \pE [\vert y_j\vert^{\frac{pb}{p-q}}\vert y_j - v_j\vert^{\frac{p(p-q-b)}{p-q}}] \Bigg] \\
        &\leq 2 \norm{x}_p^{\frac{pa}{p-q}} \norm{y}^{\frac{pb}{p-q}} \left( \pE_0 \norm{x-u}^{2p(p-q-a)} \norm{y-v}^{2p(p-q-b)} \right)^{\frac 1 {2(p-q)}}
    \end{align*}
    Plugging these bounds in, along with the fact that for any non-negative integers $\alpha, \beta$, we have that 
    \[ \pE_0 \norm{x-u}^{p\alpha} \norm{y-v}^{p\beta} = \left( \pE_0 \norm{x-u}^p \right)^\alpha \left( \pE_0 \norm{y-v}^p \right)^\beta\]
    from our equality constraints, and we get that
    \begin{align*}
        \eps \cdot \OPT \leq &O_p(1) \cdot \left( \pE_0 \norm{y-v}_p^p\right)^{\frac {p-q-b}{p}} \norm{x}_p^{a\left(1-\frac{p(p-1)}{2(p-1)^3}\right)} \left( \pE_0 \norm{x-u}_p^p \right)^{\frac {p-q-a}{p} \left(1-\frac{p(p-1)}{2(p-1)^3}\right)} \\
        &\OPT^{\frac q p} \cdot \left( \E_{\{\mu\}}\E_{j \sim \psi} \left( \Phi_{a,q}(u) - \E_{v_j} \Phi_{a,q}(u \vert v_j)\right)\right)^{\frac{p-q}{2(p-1)^2}}\,.
    \end{align*}
    Note that from our SDP constraints\footnote{Note that $d_u,d_v$ always are at most $O_p(1) \cdot \frac \OPT {\norm{x}_p^p}$.}, we have that
    \[ \max \left(\left( \pE_0 \norm{x-u}_p^p \right) , \left( \pE_0 \norm{y-v}_p^p \right) \right) \leq \frac{\OPT}{\norm{x}^{p}}\,.\]
    Furthermore, $\OPT \leq O_p(1) \cdot \norm{x}_p^p$ since $\OPT \leq C_p \norm{A}_p^p$ and thus $\left( \frac{\OPT}{\norm{x}^p}\right)^{\frac{p-q-(a+b)}{p-q}} \leq O_p(1)$ since $a+b \leq p-q$. Plugging this in, along with the fact that $\norm{x} = \norm{y}$ and simplifying and rearranging yields that
    \begin{align*}
        \eps^{O_p(1)} \cdot \norm{x}^{\frac {p(2a - (p-q))} {p-q}} \cdot \OPT^{\frac{p-q-a}{p-q}} \leq  \E_{\{\mu\}}\E_{j \sim \psi} \left( \Phi_{a,q}(u) - \E_{v_j} \Phi_{a,q}(u \vert v_j)\right)\,,
    \end{align*}
    which completes the proof of potential alignment.
\end{proof}

\subsection{Proof of~\cref{thm:lra-ptas}}

We now return to the proof of the main theorem of this section.
\begin{proof}[Proof of~\cref{thm:lra-ptas}]
    We first observe that if $\OPT = 0$ then we return a solution of cost $0$ in the first step of~\cref{algo:lra-ptas}. Furthermore, if $\OPT \geq C_p \norm{A}_p^p$ for a fixed constant $C_p$ only depending on $p$, then by applying~\cref{thm:lra-additive} with $\eps' = \eps / C_p$ we have that the solution produced by step 7 of~\cref{algo:lra-ptas} has value at most $\OPT + \left( \eps / C_p \right) \norm{A}_p^p \leq \OPT + \eps \cdot \OPT$. Thus, in this case, we will output a $(1+\eps)$ optimal solution from the candidate solution in step $7$. 

    In order to show the approximation guarantee of the algorithm, it suffices to show that one of the $\hat{u},\hat{v}$ solutions produced in the loop is a $(1+\eps)$ multiplicative approximation to $\OPT$ in the case where both $\OPT \neq 0$ and $\OPT \leq C_p \norm{A}_p^p$. Let $x,y$ be the constant factor LRA solution computed. By~\cref{cor:strong-identifiability} we have that after rescaling $x,y$ such that $\norm{x} = \norm{y}$ there exists an optimal solution $u,v$ such that 
    \[ \max\{\norm{u-x}_p^p, \norm{v-y}_p^p\} \leq O\left( \frac{\OPT}{\norm{x}_p^p} \right)\,.\]
    Furthermore, note that the same solution satisfies the last three inequalities in $\calA$ by~\cref{lem:addtl-sdp-constraints-satisfied}. Note that both these facts are still true after rounding the entries of $x,y$, since we round them to a net of fineness $\frac{\OPT}{\left(nm \cdot 2^B\right)^{O_p(1)}}$\footnote{Note that after the rounding the norms of $x,y$ might not be exactly equal, but we note that everywhere we use that $\norm{x} = \norm{y}$ it suffices that they are instead equal up to multiplicative factors.}. Furthermore, by~\cref{lem:discretization-error-bound} we have that after rounding the entries $x,y$ are still a constant factor approximation LRA solution.

    By~\cref{lem:discretization-error-bound} we have that we can round $u,v$ to values in $\Sigma$ while still maintaining that the rounded vectors are a $(1+\eps) \cdot \OPT$ value solution. Let $\tilde{u}, \tilde{v}$ be the rounded solution. Furthermore, we still have that
    \[ \max\{\norm{\tilde{u}-x}_p^p, \norm{\tilde{v}-y}_p^p\} \leq O\left( \frac{\OPT}{\norm{x}_p^p} \right)\,,\]
    and that the last three inequalities in $\calA$ are satisfied by $\tilde{u}, \tilde{v}$. We now consider the potential values of $\sum_{i \in [n]} x_i^{p-k} \tilde{u}_i^k$, $\sum_{j \in [m]} y_j^{p-k} \tilde{v}_j^k$, $\norm{\tilde{u} - x}$, and $\norm{\tilde{v} - y}$. Note that since the entries of $x,y,u,v$ all live on a grid of fineness $\frac{\eps \OPT}{\left(nm \cdot 2^B\right)^{O_p(1)}}$, these values all live in a grid of fineness $\frac{\eps^p \OPT^p}{\left(nm \cdot 2^B\right)^{O_p(1)}}$. For $\norm{\tilde{u} - x}$ and $\norm{\tilde{v} - y}$ we also know that the value must be in the interval $\left[ 0,\, O\left( \frac{\OPT}{\norm{x}_p^p} \right)\right]$, so it suffices to search a grid of this interval. For $\sum_{i \in [n]} x_i^{p-k} \tilde{u}_i^k$ and $\sum_{j \in [m]} y_j^{p-k} \tilde{v}_j^k$ we know that these have value at most $nm \cdot 2^{O_p(B)}$ so it suffices to search the interval $[0, nm \cdot 2^{O_p(B)}]$. Thus, for one iteration of the loop, we will correctly guess all of these quantities, and for this iteration we will have that the optimal value of the SDP is at most $(1+\eps) \cdot \OPT$.
    
    Furthermore, note that we only incur an additional expected $\eps \cdot \OPT$ cost after rounding the SDP solution via~\cref{lem:ptas-rounding}. By Markov's, we produce a solution of value at most $(1+O(\eps))\cdot \OPT$ with probability $0.99$. Thus, in this step of the loop we have a SDP with $\OPT_{SDP} \leq (1+O(\eps)) \cdot \OPT$ and furthermore the rounded solution has value $\OPT_{SDP} + O(\eps) \cdot \OPT$ so we have computed a solution $\hat{u}, \hat{v}$ with value at most $(1+O(\eps)) \cdot \OPT$. 

    Finally, it remains to analyze the runtime. The runtime cost is dominated by the cost of the loop. There are $\left(nm \cdot 2^B/\OPT\right)^{O_p(1)}$ iterations of the loop. Each iteration is dominated by the cost of computing an SDP solution, which can be done in time $\left(nm \cdot 2^B/\OPT\right)^{(1/\eps)^{O_p(1)}}$ by~\cref{fact:eff-pseudo-distribution}. Since $\OPT \geq (nm)^{-O_p(1)} \cdot 2^{-O(B)}$ by~\cref{lem:lb-opt}, we have that the total runtime is at most $\left(nm \cdot 2^B\right)^{(1/\eps)^{O_p(1)}}$ as desired.
\end{proof}

\section{Higher Rank LRA}
\label{sec:higher-rank-lra}

In this section, we extend the PTAS for rank-$1$ LRA to rank-$k$ LRA for all $k$. Some of the key ideas in this section (in particular the decomposition of the objective function and residual balancing lemma) were suggested by ChatGPT Pro when given our draft of the previous sections. We then verified the proofs below independently.

\begin{theorem}
\label{thm:rank-k-lra-main}
    Given a matrix $A \in \R^{n \times m}$ with bit complexity $B$, even $p \in N$, and $0<\eps<1$, there exists an algorithm that runs in $(nm\cdot 2^B)^{(k/\eps)^{O_p(1)}}$ time and outputs matrices $\hat{U} \in \R^{n\times k}, \hat{V} \in \R^{m \times k}$ such that with probability at least $0.99$, 
    \begin{equation*}
        \norm{ A - \hat{U} \hat{V}^\top }_p^p \leq (1+ \eps) \cdot \OPT\,. 
    \end{equation*}
\end{theorem}

Our algorithm follows a similar approach to the rank $1$ case, but with some important differences. We introduce a new higher rank identifiability argument for LRA as well as new SDP variables $N_1, N_2$, which are used in the analysis to balance the positive and negative residuals. We give further explanation of variables $N_1, N_2$ when we show the feasibility of the constraint set below for an (approximately) optimal LRA solution. Our SDPs use the following constraint set (in which $A_s$ denotes the set of $s$ tuples):
\begin{align*}
\mathcal{A}(s) = \left\{ 
\begin{gathered}
    \forall \alpha \in A_{p}\,,b\in[p]\,: \qquad \sum_{i} \left(\left(N_1^\top (X_i+U_i)\right)^{\otimes b}\otimes \left( N_2^\top (X_i-U_i)\right)^{\otimes (p-b)}\right)_\alpha = s_{u, \alpha,b} \\
    \forall \alpha \in A_{p}\,,b\in[p]\,: \qquad \sum_j \left(\left(N_1^{-1}(Y_j-V_j)\right)^{\otimes b} \otimes \left(N_2^{-1}(Y_j + V_j)\right)^{\otimes (p-b)}\right)_{\alpha} = s_{v, \alpha,b} \\
    N_1^\top \left(N_1^\top\right)^{-1} = I_k,\,N_2^\top \left(N_2^\top\right)^{-1} = I_k \\
    \forall a\in[p-1], 0 \leq b \leq p-a:\\
    \sum_{i \in [n]} \left\Vert\left(N_1^\top (X_i+U_i)\right)^{\otimes b} \otimes \left( N_2^\top (X_i-U_i)\right)^{\otimes (p-a-b)}\right\Vert_{\frac{p}{p-a}}^{\frac{p}{p-a}} \leq k^{O_p(1)} \sqrt{\OPT} \\
    \sum_{j \in [m]} \left\Vert\left(N_1^{-1} (Y_j - V_j)\right)^{\otimes b} \otimes \left( N_2^{-1} (Y_j + V_j)\right)^{\otimes (p-a-b)}\right\Vert_{\frac{p}{p-a}}^{\frac{p}{p-a}} \leq k^{O_p(1)} \sqrt{\OPT} \\
\end{gathered}
\right\}
\end{align*}
\begin{remark}
    Note that the above constraint set also includes fractional powers of the form $\frac{p}{p-a}$. We note that this can be rewritten as a polynomial system of constraints via introducing auxiliary variables for $(N_2^\top (X-U))_{ij}^{\frac{1}{\prod_{a=0}^{p-1} (p-a)}}$, etc with only a polynomial blowup in the number of variables (for fixed $p$). We do not explicitly write this representation for ease of exposition and note that whenever fractional exponents appear, they can be represented within this framework. Furthermore, the above set of constraints also includes ``variables'' representing $N_1^{-1}, N_2^{-1}$ which are restricted by the constraints $N_1 N_1^{-1} = I_k$ and $N_2 N_2^{-1} = I_k$. 
\end{remark}
\begin{remark}
    Note that since all our variables are constrained to live on a grid, it is not possible to satisfy the equality constraints as written, even for the intended optimal solution. Thus, we will need to relax these constraints to instead require equality up to $2^{-O(kB)} \cdot \frac 1 {\poly_p(nm)}$ error. We will work with the true equality constraints below for ease of exposition. 
\end{remark}

\begin{mdframed}
  \begin{algorithm}[$\ell_p$ LRA Rank $k$ PTAS]
    \label{algo:lra-ptas-rank-k}\mbox{}
    \begin{description}
    \item[Input:] Matrix $A \in \mathbb{R}^{n \times m}$ and target accuracy $0<\epsilon<1$.
    
    \item[Operations:]\mbox{}
    \begin{enumerate}
        \item If $A$ is a rank-$k'$ matrix for $k' \leq k$ output the rank-$k'$ factorization $\widehat{U}, \widehat{V}$.
        \item Let $X_0,Y_0$ be an $O(\OPT)$ LRA solution for $A$ from Theorem 14 in~\cite{ban_etal_soda2019_ptas_lp_lra} and let the objective value be $\widetilde{\OPT}$.
        \item Let $X_0Y_0^\top = U \Sigma V^\top$ be the SVD factorization and let $X_1 = U \Sigma^{1/2}, Y_1 = V \Sigma^{1/2}$. Let $X, Y$ be $X_1 + H_1$ and $Y_1 + H_2$ where $H_1$, $H_2$ are random matrices with all entries drawn independently from $U(-2^{-\Omega_p(B)} \cdot \frac 1 {\poly_p(nmk)}  \cdot \widetilde{\OPT},  2^{-\Omega_p(B)} \cdot \frac 1 {\poly_p(nmk)}  \cdot \widetilde{\OPT}) + U(-2^{-\Omega_p(B)} \cdot \frac 1 {\poly_p(nmk)}  \cdot \widetilde{\OPT},  2^{-\Omega_p(B)} \cdot \frac 1 {\poly_p(nmk)}  \cdot \widetilde{\OPT})$.
        \item Let $\Delta = \poly_p(nmk) \cdot 2^{O_p(B)} \cdot \left( \frac 1 {\eps \cdot \widetilde{\OPT}}\right)^{\Theta_p(1)}$.
        \item Let $S$ be a $\left(\frac{\eps\cdot \widetilde{\OPT}}{knm \cdot 2^B}\right)^{O_p(1)}$-fineness grid of $[-\Delta, \Delta]$.
        \item For each $s$ such that $s_{u,\alpha,b}, s_{v,\alpha,b} \in S$ and set $\calT\subset [m]$ of size $(k/\eps)^{O_{p}(1)}$:
        \begin{enumerate}
            \item Let $\Sigma$ be a $\left(\frac{\eps\cdot \widetilde{\OPT}}{knm \cdot 2^B}\right)^{O_p(1)}$-fineness grid of $[-\Delta, \Delta]$.
            \item Let $\mu$ be a $(k/\eps)^{O_{p}(1)}$-degree sum-of-squares pseudo-distribution over the alphabet $\Sigma^{kn+km + 4k^2}$ with variables $U_i$, $V_j$, $N_1$, $N_2$, $N_1^{-1}$, $N_2^{-1}$ such that $\pE_{\mu}$ optimizes 
            \begin{align}
                &\min_{\pE_{\mu}} \quad \pE_{\mu} \sum_{i\in[n],j\in[m]} (A_{ij} - \langle U_i, V_j\rangle)^p \notag \\
                &\text{s.t.\ } \pE_{\mu} \text{ satisfies } \calA (s)\notag
            \end{align}
            \item Let $\widehat{N_1}, \widehat{N_2}, \widehat{N_1}^{-1}, \widehat{N_2}^{-1}$ be draws from their local distributions and condition the SDP pseudodistribution on the values of $\widehat{N_1}, \widehat{N_2}, \widehat{N_1}^{-1}, \widehat{N_2}^{-1}$.
            \item Let $\widehat{V}_\calT$ be a draw from the local distribution $\{V_{\calT}\}$. Let $\mu_\calT$ be the pseudodistribution obtained by conditioning on $\{V_j = \widehat{V}_j\}_{j \in \calT}$.
            \item Let $\widehat{U}, \widehat{V}$ be the solution where $\widehat{U}_i, \widehat{V}_j$ are sampled independently from their degree $k$ local distributions in $\mu_\calT$. 
        \end{enumerate}
    \end{enumerate}

    \item[Output:] The embedding $\widehat{U}, \widehat{V}$ with the lowest LRA objective value over all iterations of the loop.
    \end{description}
  \end{algorithm}
\end{mdframed}

\begin{remark}
    Although the above algorithm utilizes sum-of-squares pseudodistributions, we note that the Sherali-Adams linear programming hierarchy alongside global linear constraints would suffice. We use sum-of-squares pseudodistributions and SDPs for ease of exposition.
\end{remark}

\subsection{Discretization and SDP Constraints.}

In this section, we show that at least one of the loop iterations has an SDP objective value at most $(1+\eps)\cdot \OPT$. We first show that there exist $X,Y,N$ with real-valued entries that satisfy the SDP constraints, and then analyze the effect of discretization on the SDP objective value.

\paragraph{Higher Rank Identifiability.} We first address the higher rank identifiability. 

\begin{lemma}[Higher Rank LRA Identifiability]
\label{lem:high-rank-ident}
    Let $X,U \in \mathbb{R}^{n \times k}$ and $Y,V \in \mathbb{R}^{m \times k}$ be full column rank matrices. Then there exists an invertible matrix $M \in \mathbb{R}^{k \times k}$ such that
    \[ \norm{(U + XM)(V - Y\left(M^{-1}\right)^{\top})^\top}_p^p + \norm{(U - XM)(V + Y\left(M^{-1}\right)^{\top})^\top}_p^p \leq k^{O_p(1)} \norm{UV^\top - XY^\top}_p^p\,.\]
    Furthermore, if $X$, $U$, $\left(X^\top X\right)^{-1/2}$, and $\left(U^\top U\right)^{-1/2}$ have entries bounded in magnitude by $\beta$, then $M$ and $M^{-1}$ have entries bounded by $\poly(\beta \cdot nmk)$.
\end{lemma}

We first prove the statement for $p=2$ and then extend it to all even $p$ via sketching (similarly to the rank $1$ case). 

\begin{lemma}[Higher Rank $\ell_2$-LRA Identifiability]
\label{lem:high-rank-ell-2-ident}
    Let $X,U \in \mathbb{R}^{n \times k}$ and $Y,V \in \mathbb{R}^{m \times k}$ be full column rank matrices. Then there exists an invertible matrix $M \in \mathbb{R}^{k \times k}$ such that
    \[ \norm{(U + XM)(V - Y\left(M^{-1}\right)^\top)^\top}_F^2 + \norm{(U - XM)(V + Y\left(M^{-1}\right)^\top)^\top}_F^2 \leq 4 \norm{UV^\top - XY^\top}_F^2\,.\]
    Furthermore, if $X$, $U$, $\left(X^\top X\right)^{-1/2}$, and $\left(U^\top U\right)^{-1/2}$ have entries bounded in magnitude by $\beta$, then $M$ and $M^{-1}$ have entries bounded by $\poly(\beta \cdot nmk)$.
\end{lemma}

Note that with~\cref{lem:high-rank-ell-2-ident} we can easily prove~\cref{lem:high-rank-ident}.
\begin{proof}[Proof Sketch]
    Note that by sketching (\cref{fact:sketch}) it suffices to show that
    \[ \norm{(\widetilde{U} + \widetilde{X}M)(\widetilde{V} - \widetilde{Y}\left(M^{-1}\right)^\top)^\top}_p^p + \norm{(\widetilde{U} - \widetilde{X}M)(\widetilde{V} + \widetilde{Y}\left(M^{-1}\right)^\top)^\top}_p^p \leq k^{O_p(1)} \norm{\widetilde{U}\widetilde{V}^\top - \widetilde{X}\widetilde{Y}^\top}_p^p\,,\]
    for $\widetilde{U}, \widetilde{X} \in \mathbb{R}^{k^{O_p(1)} \times k}$ and $\widetilde{V}, \widetilde{Y} \in \mathbb{R}^{k^{O_p(1)} \times k}$. Since $\ell_p$ and $\ell_2$ norms are equivalent up to dimension-dependent factors, it thus also suffices to instead show that
    \[ \norm{(\widetilde{U} + \widetilde{X}M)(\widetilde{V} - \widetilde{Y}\left(M^{-1}\right)^\top)^\top}_F^2 + \norm{(\widetilde{U} - \widetilde{X}M)(\widetilde{V} + \widetilde{Y}\left(M^{-1}\right)^\top)^\top}_F^2 \leq O_{p,k}k^{O_p(1)} \norm{\widetilde{U}\widetilde{V}^\top - \widetilde{X}\widetilde{Y}^\top}_F^2\,,\]
    Applying~\cref{lem:high-rank-ell-2-ident} completes the proof of the identifiability inequality. To show that $M$ and $M^{-1}$ are bounded we note that the sketched matrices $\widetilde{X}$, $\widetilde{U}$, $\left(\widetilde{X}^\top \widetilde{X}\right)^{-1/2}$, and $\left(\widetilde{U}^\top \widetilde{U}\right)^{-1/2}$ still have bounded entries and thus applying~\cref{lem:high-rank-ell-2-ident} gives the desired bound on $M, M^{-1}$.
\end{proof}

We now return to the proof of~\cref{lem:high-rank-ell-2-ident}.
\begin{proof}[Proof of~\cref{lem:high-rank-ell-2-ident}]
    Note that we have that
    \[ (U + XM)(V - Y\left(M^{-1}\right)^\top)^\top + (U - XM)(V + Y\left(M^{-1}\right)^\top)^\top = 2 \left(UV^\top - XY^\top\right)\,.\]
    Thus, it suffices to show that there is a choice of $M$ such that $(U + XM)(V - Y\left(M^{-1}\right)^\top)^\top, (U - XM)(V + Y\left(M^{-1}\right)^\top)^\top$ are orthogonal. Let 
    \[ \widehat{X} \vcentcolon= X (X^\top X)^{-1/2} \quad\text{and}\quad \widehat{U} \vcentcolon= U(U^\top U)^{-1/2}\,.\]
    Furthermore, let $\widehat{X}^\top \widehat{U} = QH$ be a polar decomposition such that $Q$ is a unitary matrix and $H \succeq 0$. It suffices to take 
    \[ M \vcentcolon= (X^\top X)^{-1/2} Q (U^\top U)^{1/2}\,.\]
    Note that $M$ is in fact invertible since $X,U$ have full column rank. We now show that this choice of $M$ satisfies $\langle (U + XM)(V - Y\left(M^{-1}\right)^\top)^\top, (U - XM)(V + Y\left(M^{-1}\right)^\top)^\top \rangle_F = 0$. We have that
    \begin{align*}
        &\langle (U + XM)(V - Y\left(M^{-1}\right)^\top)^\top, (U - XM)(V + Y\left(M^{-1}\right)^\top)^\top \rangle_F \\
        &\qquad= \tr\left((U + XM)^\top (U - XM) (V + Y\left(M^{-1}\right)^\top)^\top (V - Y\left(M^{-1}\right)^\top) \right)\,,
    \end{align*} 
    and we will show that $(U + XM)^\top (U - XM) = 0$. We have that 
    \begin{align*}
        (U + XM)^\top (U - XM) &= U^\top U - M^\top X^\top X M + M^\top X^\top U - U^\top X M\,.
    \end{align*}
    Note that 
    \[ M^\top X^\top X M  = (U^\top U)^{1/2} Q^\top Q (U^\top U)^{1/2} = U^\top U\,,\]
    so we have that $(U + XM)^\top (U - XM) = M^\top X^\top U - U^\top X M$. Furthermore, for our choice of $M$, the matrix $M^\top X^\top U$ is symmetric, so this is $0$. We have that 
    \begin{align*}
        U^\top X M &= U^\top X (X^\top X)^{-1/2} Q (U^\top U)^{1/2} \\
        &= U^\top \widehat{X} Q (U^\top U)^{1/2} \\
        &= (U^\top U)^{1/2} \widehat{U}^\top \widehat{X} Q (U^\top U)^{1/2} \\
        &= (U^\top U)^{1/2} H (Q^\top Q) (U^\top U)^{1/2} \\
        &= (U^\top U)^{1/2} H (U^\top U)^{1/2} \,,
    \end{align*}
    using the definitions of $\widehat{X}, \widehat{U}$, and $M$ along with the fact that $Q$ is unitary. Finally, note that $H$ is symmetric, so we have that $U^\top X M$ is symmetric. Thus, we have that $(U + XM)^\top (U - XM) = 0$ and 
    \begin{align*}
        &\langle (U + XM)(V - Y\left(M^{-1}\right)^\top)^\top, (U - XM)(V + Y\left(M^{-1}\right)^\top)^\top \rangle_F \\
        &\qquad= \tr\left((U + XM)^\top (U - XM) (V + Y\left(M^{-1}\right)^\top)^\top (V - Y\left(M^{-1}\right)^\top) \right) = 0\,,
    \end{align*} 
    Since $(U + XM)(V - Y\left(M^{-1}\right)^\top)^\top, (U - XM)(V + Y\left(M^{-1}\right)^\top)^\top$ are orthogonal, we have that 
    \[ \norm{(U + XM)(V - Y\left(M^{-1}\right)^\top)^\top}_F^2 + \norm{(U - XM)(V + Y\left(M^{-1}\right)^\top)^\top}_F^2 \leq 4 \norm{UV^\top - XY^\top}_F^2\,,\]
    which completes the proof of the identifiability inequality.

    We now consider the explicit form above for $M$ and $M^{-1}$ to bound the entries. We have that
    \[ M = (X^\top X)^{-1/2} Q (U^\top U)^{1/2}\,,\]
    which has bounded entries by our assumption on the entries of $U$ and $(X^\top X)^{-1/2}$, along with the fact that $Q$ is a unitary matrix and thus also bounded. We now consider the explicit form for $M^{-1}$. We have that 
    \[ M^{-1} = (X^\top X)^{1/2} Q^\top (U^\top U)^{-1/2}\,,\]
    which has bounded entries by our assumption on the entries of $X$ and $(U^\top U)^{-1/2}$, along with the fact that $Q^\top$ is a unitary matrix and thus also bounded. This completes the second part of the proof.
\end{proof}

\paragraph{Balancing Residuals.}

We now show that there exist $N_1$, $N_2$ satisfying the SDP constraints. Intuitively, these matrices are attempting to balance the residuals such that the norms $\norm{(U + X)(V - Y)^\top}_p, \norm{(U - X)(V + Y)^\top}_p$ being large is witnessed by the norms of both matrices in the product.

\begin{lemma}[Pairwise $\ell_p$ Balancing]
\label{lem:pairwise-ell-p-balance}
    Let $p \in \mathbb{R}^{\geq 1}$ and let $F \in \mathbb{R}^{n \times k}$ and $G \in \mathbb{R}^{m \times k}$ have full column rank. Then there exists an invertible $N \in \mathbb{R}^{k \times k}$ such that 
    \[ \norm{F N}_p^p = \left\Vert G \left(N^{-1}\right)^\top\right\Vert_p^p \leq k \cdot \norm{FG^\top}_p^{p/2}\,.\]
\end{lemma}

In the proof of~\cref{lem:pairwise-ell-p-balance}, we will need the following fact:
\begin{fact}[Auerbach's Lemma]
\label{fact:auerbach}
  Let $V$ be a finite-dimensional normed vector space with norm $\| \cdot \|$.
  There exists $v_1,\ldots,v_n \in V$ which is a basis for $V$ and $v_1^*,\ldots,v_n^* \in V^*$ such that $v_i^*(v_j) = \delta_{ij}$, $\|v_i\| = 1$, and $\|v_i^*\|_* = 1$.
\end{fact}
\begin{proof}[Proof of~\cref{lem:pairwise-ell-p-balance}]
    Note that it suffices to show that there exists an invertible $N_0$ such that
    \[\norm{F N_0}_p^p \cdot \left\Vert G \left(N_0^{-1}\right)^\top\right\Vert_p^p \leq k^2 \cdot \norm{FG^\top}_p^{p} \,.\]
    Given such an $N_0$ we can pick $s$ such that $\norm{F sN_0}_p^p = \norm{G (sN_0)^{-1}}_p^p$ and then letting $N = sN_0$ proves the claim. Define the following norm on $\mathbb{R}^k$:
    \[ \norm{x}_F \vcentcolon= \norm{F x}_p\,.\]
    Note that since $F$ has full column rank, this is in fact a norm. Furthermore, by~\cref{fact:auerbach} we have that there exists a basis $v_1,\ldots, v_k$ and a dual basis $v_1^*,\ldots, v_k^*$ satisfying the above Auerbach basis properties for $\norm{\cdot}_F$. Let $N_0$ be the invertible matrix whose $t$-th column is $v_t$. Then we have that
    \[ \norm{F N_0}_p^p = \sum_{i = 1}^n \norm{F_i N_0}_p^p = \sum_{i = 1}^n \sum_{t = 1}^k (\langle F_i, v_t\rangle)^p = \sum_{t= 1}^k \norm{F v_t}_p^p = k\,.\]
    Now consider $\norm{G \left(N_0^{-1}\right)^\top}_p^p$. We will bound this by showing that $FN_0$ does not contract the $\ell_p$ norm of any vector $x$. Note that we have that for any vector $x$ and any $t \in [k]$,
    \[ \abs{x_t} = \left\vert v_t^* \left(\sum_{i=1}^k x_i v_i\right)\right\vert\]
    by the Auerbach basis properties since $v_t^*(v_t) = 1$ and $v_t^*(v_u) = 0$ for all $u \neq t$. Furthermore, by generalized Cauchy-Schwartz
    \[ \abs{x_t} = \left\vert v_t^* \left(\sum_{i=1}^k x_i v_i\right)\right\vert \leq \norm{v_t^*}_* \norm{\sum_{i=1}^k x_i v_i}_F = \norm{F N_0 x}_p\,.\]
    Thus, we have concluded that $\norm{x}_\infty^p \leq \norm{F N_0 x}_p^p$ and $\norm{x}_p^p \leq k\norm{F N_0 x}_p^p$. Let $\widetilde{G} = G \left(N_0^{-1}\right)^\top$, and we now apply the above fact with $\widetilde{g}_j$. Note that we have that
    \begin{align*}
        \norm{FG^\top}_p^p &= \norm{FN_0\widetilde{G}^\top}_p^p = \sum_{j = 1}^m \norm{FN_0 \widetilde{g}_j}_p^p \geq \frac 1 k \sum_{j = 1}^m \norm{\widetilde{g}_j}_p^p = \frac 1 k \left\Vert G \left(N_0^{-1}\right)^\top\right\Vert_p^p\,.
    \end{align*}
    Putting this all together, we have that
    \[ \norm{FG^\top}_p^p \geq \frac 1 k \norm{G \left(N_0^{-1}\right)^{\top}}_p^p \geq \frac 1 {k^2} \norm{F N_0}_p^p \cdot \left\Vert G \left(N_0^{-1}\right)^\top\right\Vert_p^p\,,\]
    which completes the proof.
\end{proof}

\paragraph{Discretization.}
We now analyze the effect of discretization on the SDP objective value. In particular, we prove the following lemma:
\begin{lemma}
\label{lem:rank-k-feasibility}
    Let $A \in \mathbb{R}^{n \times m}$ have entries with bit complexity $B$. Then with probability at least $0.999$ in~\cref{algo:lra-ptas-rank-k} there exists some $s$ with $s_{u,\alpha,b}, s_{v,\alpha,b} \in S$ such that the SDP with constraints $\mathcal{A}(s)$ and alphabet $\Sigma$ has objective value at most $(1+\eps)\cdot \OPT$.
\end{lemma}
In the proof of the lemma above, we will need to show that solutions that are close in $\ell_\infty$ norm have similar objective value, and that it is possible to refactor good solutions $X_0,Y_0$ such that they have bounded magnitude entries. We defer the proofs of these facts to~\cref{sec:appendix}.
\begin{proof}
    Consider $X, Y$ produced by~\cref{algo:lra-ptas-rank-k}. Note that $X_0, Y_0$ were originally an $O(\OPT)$ value solution, and by~\cref{lem:constant-factor-bounded-bit-complexity} we have that the factorization $X_1, Y_1$ given by SVD have entries with magnitude at most $\poly_p(nmk) \cdot 2^{O_p(B)}$. Consider $X_2, Y_2$ generated by adding random matrices to $X_1, Y_1$ with entries which are independent and sampled from $U(-2^{-\Omega_p(B)} \cdot \frac 1 {\poly_p(nmk)} \cdot \OPT, 2^{-\Omega_p(B)} \cdot \frac 1 {\poly_p(nmk)} \cdot \OPT)$. Note that $X_2, Y_2$ still have entries of magnitude at most $\poly_p(mnk) \cdot 2^{O_p(B)}$. In addition, by~\cref{lem:rounding-rank-k-disc-solutions} we have that since each entry changes by at most $2^{-\Omega_p(B)} \cdot \frac 1 {\poly_p(nmk)} \cdot \OPT$ and $\norm{X_1}_\infty,\norm{Y_1}_\infty$ are bounded, we have that $X_2, Y_2$ is still a $O(\OPT)$ value solution to the LRA problem on $A$. Finally, since we added $2^{-\Omega_p(B)} \cdot \frac 1 {\poly_p(nmk)} \cdot \OPT$ noise, we have that with probability $0.9999$ that the singular values of $X_2$ are lower bounded in magnitude by $2^{-\Omega_p(B)} \cdot \frac 1 {\poly_p(nmk)} \cdot \OPT$ and thus $(X_2^\top X_2)^{-1/2}$ has entries bounded in magnitude by $2^{O_p(B)} \cdot \poly_p(nmk) \cdot \left(1/\OPT\right)^{O(1)}$.

    Now consider some rank $k$ matrix $G$ associated with an optimal solution to the LRA problem on $A$. By~\cref{lem:constant-factor-bounded-bit-complexity} there exists a factorization of $G = U_0 V_0^\top$ such that $U_0,V_0$ have entries of magnitude at most $\poly_p(nmk) \cdot 2^{O_p(B)}$. Furthermore, if we add $U(-2^{-\Omega_p(B)} \cdot \frac 1 {\poly_p(nmk)} \cdot \eps \cdot \OPT, 2^{-\Omega_p(B)} \cdot \frac 1 {\poly_p(nmk)} \cdot \eps \cdot \OPT)$ noise to each entry of $U_0,V_0$ we similarly get matrices $U_1, V_1$ which are still $(1+\eps)\cdot \OPT$ value solutions (by~\cref{lem:rounding-rank-k-disc-solutions}), have bounded entries, and with probability $0.9999$ the singular values of $U_1$ lower bounded by $2^{-O_p(B)} \cdot \frac 1 {\poly_p(nmk)} \cdot {\eps \cdot \OPT}$ (and thus the entries of $(U_1^\top U_1)^{-1/2}$ are bounded by $2^{O_p(B)} \cdot \poly_p(nmk) \cdot \left(\frac 1 {\eps \cdot \OPT}\right)^{O(1)}$).

    We now apply~\cref{lem:high-rank-ident} to the matrices $X_2, Y_2$ and $U_1, V_1$. Note that we have bounds on the entries of $X_2$, $U_1$, $(X_2^\top X_2)^{-1/2}$, $(U_1^\top U_1)^{-1/2}$ and thus there exists an invertible $M$ such that
    \begin{align*}
        &\norm{(X_2 + U_1M)(Y_2 - V_1\left(M^{-1}\right)^{\top})^\top}_p^p + \norm{(X_2 - U_1M)(Y_2 + V_1M^{-1})^\top}_p^p \\
        &\qquad\leq k^{O_p(1)} \cdot \norm{U_1V_1^\top - X_2Y_2^\top}_p^p \\
        &\qquad\leq k^{O_p(1)} \cdot O_p\left(\norm{U_1V_1^\top - A}_p^p + \norm{A - X_2Y_2^\top}_p^p\right)\\
        &\qquad\leq k^{O_p(1)} \cdot O_p(\OPT) \,,
    \end{align*}
    and the entries of $M$ and $M^{-1}$ are both bounded by $\poly_p(nmk) \cdot 2^{O_p(B)} \cdot \left(\eps \cdot \OPT\right)^{-\Theta_p(1)}$. Thus, if we take $U = U_1 M$ and $V = V_1 (M^{-1})^\top$ then both $U,V$ have entries bounded by $\poly_p(nmk) \cdot 2^{O_p(B)} \cdot \left( \eps \cdot \OPT\right)^{-\Theta_p(1)}$. Furthermore, $U,V$ have the same solution value as $U_1, V_1$ so they are $(1+\eps)\cdot \OPT$ valued solutions. 

    We now consider adding $U(-2^{-\Omega_p(B)} \cdot \frac 1 {\poly_p(nmk)} \cdot \OPT, 2^{-\Omega_p(B)} \cdot \frac 1 {\poly_p(nmk)} \cdot \OPT)$ independent noise to each entry of $X_2,Y_2$ to get $X,Y$. Note that the distribution of $X,Y$ is identical to $X,Y$ in step 3 of~\cref{algo:lra-ptas-rank-k}. We can also take the amount of noise to be sufficiently small such that we have
    \[ \norm{(X+ U)(Y-V)^\top}_p^p + \norm{(X-U)(Y+V)^\top}_p^p \leq  k^{O_p(1)} \cdot O_p(\OPT)\,.\]
    Furthermore, note that since the noise was independent of $U,V$, we now have that all residuals $X-U$, $X+U$, $Y-V$, and $Y+V$ have full column rank. Thus, we can apply~\cref{lem:pairwise-ell-p-balance} to show that there exists invertible $N_1$ and $N_2$ such that
    \[ \norm{(X+U)N_1^\top}_p^p = \norm{(Y-V)N_1^{-1}}_p^p \leq k^{O_p(1)} \cdot \norm{(X+U)(Y-V)^\top}_p^{p/2} \leq k^{O_p(1)} \cdot O_p(\sqrt{\OPT})\,,\]
    and
    \[ \norm{(X-U)N_2^\top}_p^p = \norm{(Y+V)N_2^{-1}}_p^p \leq k^{O_p(1)} \cdot \norm{(X-U)(Y+V)^\top}_p^{p/2} \leq k^{O_p(1)} \cdot O_p(\sqrt{\OPT})\,,\]
    We now show that these $N_1$, $N_1^{-1}$, $N_2$, and $N_2^{-1}$ must have bounded entries. Note that since $\ell_p$ and $\ell_2$ norms are equivalent up to dimension-dependent factors, we have that $\norm{(X+U)N_1^\top}_F^2 \leq \poly(nk) \cdot \norm{(X+U)N_1^\top}_p^2 \leq  (nk)^{O_p(1)} \cdot O_p(\OPT^{1/p})$. Equivalently, we have that 
    \[ \tr\left(N_1 (X+U)^\top (X+U) N_1^\top\right) \leq (nk)^{O_p(1)} \cdot O_p(\OPT^{1/p})\,.\]
    Note that with high probability we have that the smallest magnitude singular value of $X+U$ is at least $2^{-O_p(B)} \cdot \frac 1 {\poly_p(nmk)} \cdot \OPT$ since we added noise to $X$ and thus this implies that $N_1$ must have singular values (and thus entries) bounded by $2^{O_p(B)} \cdot \poly_p(nmk) \cdot \left(1/\OPT\right)^{O(1)}$. Similarly, we can also bound $N_1^{-1}$, $N_2$, and $N_2^{-1}$ via the same argument. 

    To summarize, we have shown with high probability (via a union bound over the singular value bounds) that there exists $U$, $V$, $N_1$, and $N_2$ such that
    \begin{enumerate}
        \item $ \norm{(X+U)N_1^\top}_p^p, \norm{(Y-V)N_1^{-1}}_p^p, \norm{(X-U)N_2^\top}_p^p, \norm{(Y+V)N_2^{-1}}_p^p \leq k^{O_p(1)} \cdot O_p(\sqrt{\OPT})$.
        \item $U$ and $V$ have entries with magnitude bounded by $\poly_p(nmk) \cdot 2^{O_p(B)} \cdot \left( \frac 1 {\eps \cdot \OPT}\right)^{\Theta_p(1)}$.
        \item $N_1$, $N_1^{-1}$, $N_2$, $N_2^{-1}$ have entries with magnitude bounded by $2^{O_p(B)} \cdot \poly_p(nmk) \cdot \left(1/\OPT\right)^{O(1)}$.
    \end{enumerate}
    We now argue that we can round the entries of $U$, $V$, $N_1$, $N_1^{-1}$, $N_2$, $N_2^{-1}$ to the closest entry in $\Sigma$ while only losing constant factors in the inequalities in the first bullet above and preserving the solution value of $U,V$ up to an $\eps\cdot \OPT$ additive factor. Let these rounded values be $\widetilde{U}$, $\widetilde{V}$, $\widetilde{N}_1$, $\widetilde{N}_1^{-1}$, $\widetilde{N}_2$, $\widetilde{N}_2^{-1}$. We will claim that these values are feasible for the SDP constraints for some value of $s$. Note that the entries of these matrices are within the interval covered by $\Sigma$, so the error in each entry compared to the original matrices is at most the grid fineness, or $2^{-\Omega_p(B)} \cdot \frac 1 {\poly_p(nmk)} \cdot \left( \eps \cdot \OPT \right)^{\Theta_p(1)}$. Note that by~\cref{lem:rounding-rank-k-disc-solutions} we immediately have that $\widetilde{U}$ and $\widetilde{V}$ are still $(1+\eps)\cdot\OPT$ valued solutions.

    We now show that our rounded solutions $\widetilde{U}$, $\widetilde{V}$, $\widetilde{N}_1$, $\widetilde{N}_2$ still satisfy
    \[ \norm{(X+\widetilde{U})\widetilde{N}_1^\top}_p^p, \norm{(Y-\widetilde{V})\widetilde{N}_1^{-1}}_p^p, \norm{(X-\widetilde{U})\widetilde{N}_2^\top}_p^p, \norm{(Y+\widetilde{V})\widetilde{N}_2^{-1}}_p^p \leq k^{O_p(1)} \cdot O_p(\sqrt{\OPT})\,.\]
    We illustrate the argument for showing that $\norm{(X+\widetilde{U})\widetilde{N}_1^\top}_p^p \leq k^{O_p(1)} \cdot O_p(\sqrt{\OPT})$ and the rest follow similarly. Note that $\left\vert\norm{(X+U)\widetilde{N}_1^\top}_p^p - \norm{(X+U)N_1^\top}_p^p\right\vert \leq O_p(\sqrt{\OPT})$ since $\norm{N_1 - \widetilde{N}_1}_\infty \leq 2^{-\Omega_p(B)} \cdot \frac 1 {\poly_p(nmk)} \cdot \left( \eps \cdot \OPT \right)^{\Theta_p(1)}$ and the entries of $X, U$ are bounded by $\poly_p(nmk) \cdot 2^{O_p(B)} \cdot \left( \frac 1 {\eps \cdot \OPT}\right)^{\Theta_p(1)}$. Similarly, $\left\vert\norm{(X+U)\widetilde{N}_1^\top}_p^p - \norm{(X+\widetilde{U})\widetilde{N}_1^\top}_p^p\right\vert \leq O_p(\sqrt{\OPT)}$. Thus, we have that $\norm{(X+\widetilde{U})\widetilde{N}_1^\top}_p^p \leq k^{O_p(1)} \cdot O_p(\sqrt{\OPT})$.

    Furthermore, this implies our SDP inequality constraints are satisfied. We have that 
    \begin{align*}
        &\sum_{i \in [n]} \norm{\left(\widetilde{N}_1^\top (X_i+\widetilde{U}_i)\right)^{\otimes b} \otimes \left( \widetilde{N}_2^\top (X_i-\widetilde{U}_i)\right)^{\otimes (p-a-b)}}_{\frac{p}{p-a}}^{\frac{p}{p-a}} \\
        &\qquad= \sum_{i \in [n]} \norm{\widetilde{N}_1^\top (X_i+\widetilde{U}_i)}_{\frac{p}{p-a}}^{\frac {pb}{p-a}} \cdot \norm{\widetilde{N}_2^\top (X_i-\widetilde{U}_i)}_{\frac{p}{p-a}}^{\frac{p(p-a-b)}{p-a}} \\
        &\qquad\leq \poly_p(k) \sum_{i \in [n]} \norm{\widetilde{N}_1^\top (X_i+\widetilde{U}_i)}_{p}^{\frac {pb}{p-a}} \cdot \norm{\widetilde{N}_2^\top (X_i-\widetilde{U}_i)}_{p}^{\frac{p(p-a-b)}{p-a}}\,.
    \end{align*}
    By H\"olders with $\frac{p-a}{b}$ and $\frac{p-a}{p-a-b}$ we have that
    \begin{align*}
        &\sum_{i \in [n]} \norm{\left(\widetilde{N}_1^\top (X_i+\widetilde{U}_i)\right)^{\otimes b} \otimes \left( \widetilde{N}_2^\top (X_i-\widetilde{U}_i)\right)^{\otimes (p-a-b)}}_{\frac{p}{p-a}}^{\frac{p}{p-a}} \\
        &\qquad\leq \poly_p(k) \sum_{i \in [n]} \norm{\widetilde{N}_1^\top (X_i+\widetilde{U}_i)}_{p}^{\frac {pb}{p-a}} \cdot \norm{\widetilde{N}_2^\top (X_i-\widetilde{U}_i)}_{p}^{\frac{p(p-a-b)}{p-a}} \\
        &\qquad\leq \poly_p(k) \cdot \left( \sum_{i \in [n]} \norm{\widetilde{N}_1^\top (X_i+\widetilde{U}_i)}_{p}^p\right)^{\frac {b}{p-a}} \cdot \left(\sum_{i \in [n]} \norm{\widetilde{N}_2^\top (X_i-\widetilde{U}_i)}_{p}^p \right)^{\frac {p-a-b}{p-a}} \\
        &\qquad= \poly_p(k) \cdot \norm{(X+\widetilde{U})\widetilde{N}_1^\top}_{p}^{\frac {pb}{p-a}} \cdot \norm{(X-\widetilde{U})\widetilde{N}_2^\top}_{p}^{\frac {p(p-a-b)}{p-a}} \\
        &\qquad\leq \poly_p(k) \sqrt{\OPT}\,.
    \end{align*}
    Similarly, we also have that
    \begin{align*}
        \sum_{j \in [m]} \left\Vert\left(\widetilde{N}_1^{-1} (Y_j - \widetilde{V}_j)\right)^{\otimes b} \otimes \left( \widetilde{N}_2^{-1} (Y_j + \widetilde{V}_j)\right)^{\otimes (p-a-b)}\right\Vert_{\frac{p}{p-a}}^{\frac{p}{p-a}} \leq k^{O_p(1)} \sqrt{\OPT}
    \end{align*}
    Now we have $\widetilde{U}$, $\widetilde{V}$, $\widetilde{N}_1$, $\widetilde{N}_1^{-1}$, $\widetilde{N}_2$, $\widetilde{N}_2^{-1}$ such that $\widetilde{U}, \widetilde{V}$ have objective value $(1+\eps)\OPT$, the matrices satisfy all the SDP constraints apart from the constraints involving $s$, and all entries are in the discrete alphabet. We now consider the constraints involving $s$. Note that for all $\alpha$ and $b \in [p]$ we have that both expressions
    \begin{center}
        $\sum_{i} \left(\left(\widetilde{N}_1^\top (X_i+\widetilde{U}_i)\right)^{\otimes b}\otimes \left( \widetilde{N}_2^\top (X_i-\widetilde{U}_i)\right)^{\otimes (p-b)}\right)_\alpha$ \\
        $\sum_j \left(\left(\widetilde{N}_1^{-1}(Y_j-\widetilde{V}_j)\right)^{\otimes b} \otimes \left(\widetilde{N}_2^{-1}(Y_j + \widetilde{V}_j)\right)^{\otimes (p-b)}\right)_{\alpha}$
    \end{center}
    have values of magnitude at most $\poly_p(nmk) \cdot 2^{O_p(B)} \cdot \left( \frac 1 {\eps \cdot \OPT}\right)^{\Theta_p(1)}$ via our bounds on the magnitude of the entries of these matrices. Thus, for some value $s$ we will correctly guess the value of each of these expressions (up to sufficiently small slack). Thus, for some $s$ there is a feasible SDP solution with objective value at most $(1+\eps)\cdot \OPT$ as desired.
\end{proof}

Furthermore, we note that the size of our alphabet depends on the value of $\OPT$. Thus, in order to bound the size of the SDP and the runtime, we require a lower bound on $\OPT$ when it is non-zero. In particular, we note that our lower bound from the rank $1$ analysis cleanly extends to the rank $k$ setting.
\begin{lemma}
\label{lem:lb-opt-rank-k}
    Let $p$ be an even integer and let $A$ be a $\mathbb{R}^{n \times m}$ matrix with entries with bit complexity $B$. Let $\OPT$ be the optimal objective value of the rank $k$ $\ell_p$ LRA value, and let $\OPT \neq 0$. Then we have that 
    \[ \OPT \geq (nm)^{-O_p(k)} \cdot 2^{-O(kB)}\,.\]
\end{lemma}
\begin{proof}[Proof of~\cref{lem:lb-opt-rank-k}]
    Let $\OPT \neq 0$. Note that $\OPT = 0$ for the rank $k$ $\ell_p$ LRA problem if and only if $\OPT = 0$ for the rank $k$ $\ell_2$ LRA problem; in both cases this occurs if and only $A$ is rank at most $k$. Thus, we can assume that the optimal value for the $\ell_2$ LRA problem, is also non-zero. Furthermore, since $\ell_2$ and $\ell_p$ norms are equivalent up to multiplicative factor of $(nm)^{O_p(1)}$, it suffices to show the statement for $p=2$. 

    Let $\hat{A}$ be the best rank-$k$ approximation to $A$, and we want to show that $\norm{A - \hat{A}}_F^2 \geq 2^{-O(kB)}$ when $A$ is not rank at most $k$. Note that since $A$ is not rank at most $k$, there must exist a $(k+1) \times (k+1)$ minor of $A$ with non-zero determinant. Let the associated submatrix be $M$. Note that $\det(M) \neq 0$, so we must have that $\abs{\det(M)} \geq 2^{-O(kB)}$, since the entries of $M$ have bit complexity at most $B$. 

    Consider $\norm{M - \hat{M}}_F^2$, or the error in the best rank $k$ approximation to $M$. Note that this is just $\sigma_{k+1}(M)^2$ and we have that $\prod_{i=1}^{k+1} \abs{\sigma_i(M)} = \abs{\det(M)} \geq 2^{-O(kB)}$. Furthermore, note that $\abs{\sigma_i(M)} \leq \norm{M}_F$, which is at most $nm 2^{O(B)}$. Thus, we can conclude that $\abs{\sigma_{k+1}(M)} \geq (nm)^{-k} 2^{-O(kB)}$. Finally, to conclude we note that any rank $k$ approximation to $A$ gives a rank $k$ approximation to $M$ with $\norm{M - \hat{M}}_F^2 \leq \norm{A - \hat{A}}_F^2$. Thus, we have that $\norm{A - \hat{A}}_F^2 \geq (nm)^{-2k}\cdot  2^{-O(kB)}$.
\end{proof}

\subsection{Rounding Error}

We now turn to analyzing the rounding error. In this section, we will prove the following lemma:
\begin{lemma}
\label{lem:high-rank-rounding}
    Let $p$ be a positive even integer and let $A \in \mathbb{R}^{n \times m}$. Let $\OPT_{\mathrm{SDP}}$ be the objective value of the SDP in~\cref{algo:lra-ptas-rank-k} for any iteration of the loop. Then the rounding step in~\cref{algo:lra-ptas-rank-k} produces a solution with expected value at most $\OPT_{\mathrm{SDP}} + \eps \cdot \OPT$.
\end{lemma}

We will defer the proof to the end of this subsection and first prove some useful helper lemmas. To motivate the lemmas below, consider the following decomposition of the objective function:
\begin{align*}
    \sum_{ij} \left(A_{ij} - \langle U_i, V_j \rangle \right)^p &= \sum_{ij} \left(A_{ij} \pm X_i^\top Y_j - \langle U_i, V_j \rangle\right)^p \\
    &= \sum_{q=0}^p \binom{p}{q} \sum_{ij} \left(A_{ij} - X_i^\top Y_j\right)^{p-q} \left(\langle X_i, Y_j\rangle - \langle U_i, V_j \rangle\right)^q\,.
\end{align*}
We now note that 
\[ \langle X_i, Y_j\rangle - \langle U_i, V_j\rangle = \frac 1 2 \left[(X_i + U_i)^\top (Y_j - V_j) + (X_i - U_i)^\top (Y_j + V_j)\right]\,, \]
and therefore for any $q \in [p]$  we have that
\begin{align*}
    &\sum_{ij} \left(A_{ij} - X_i^\top Y_j\right)^{p-q} \left(\langle X_i, Y_j \rangle - \langle U_i, V_j \rangle\right)^q \\
    &\qquad= O(1) \cdot \sum_{a = 0}^q \sum_{ij} \left(A_{ij} - X_i^\top Y_j\right)^{p-q} \left((X_i + U_i)^\top (Y_j - V_j)\right)^a \left( (X_i - U_i)^\top (Y_j + V_j) \right)^{q-a} \\
    &\qquad= O(1)\cdot \sum_{a = 0}^q \sum_{ij} \left(A_{ij} - X_i^\top Y_j\right)^{p-q} \left\langle (X_i + U_i)^{\otimes a} \otimes (X_i - U_i)^{\otimes (q-a)}, (Y_j - V_j)^{\otimes a} \otimes (Y_j + V_j)^{\otimes (q-a)}\right\rangle\,.
\end{align*}

We will analyze the rounding error separately on each part of the objective function given by a similar decomposition. However, in order to ensure that the resulting variables are bounded by $\OPT$ factors, we will need to slightly modify the decomposition above.

\paragraph{Potential Alignment Lemma.} For all but the first part of the objective (degree $p$ in our variables $U_i,V_i$), we will utilize the following potential alignment lemma:
\begin{lemma}
\label{lem:lra-vec-pinning}
    Let $p$ be an even integer, $1 \leq q < p$ be an integer, $Z_i, W_j \in \mathbb{R}^s$ be vectors, and $B \in \mathbb{R}^{n \times m}$ be a matrix of coefficients. Let
    \[ \err \vcentcolon= \left\vert \sum_{i,j} B_{ij}^q \E_{\mu_{Z,W}} \langle Z_i, W_j \rangle - \sum_{ij} B_{ij}^q \E_{\mu_Z \otimes \mu_W} \langle Z_i, W_j \rangle \right\vert\,,\]
    where $\mu_{Z} \otimes \mu_{W}$ denotes the distribution where $Z_i, W_j$ are drawn independently from each other but with correlations within their respective coordinates. Let $\left((Z_i)_b\right)^{(0)}$ denote the distribution of $(Z_i)_b$ under the initial pseudodistribution. Then we have that any pseudodistribution satisfying $\norm{Z}_{\frac p {p-q}}^{\frac p {p-q}} \leq \beta_Z$ and $\norm{W}_{\frac p {p-q}}^{\frac p {p-q}} \leq \beta_W$ is $\left(\gamma, \frac{\gamma^{\frac{2(p-1)^2}{p-q}}}{s^{O_p(1)} \cdot \norm{B}_p^{\frac{2q(p-1)^2}{p-q}} \cdot \beta_Z^{\frac{(p-2)(2p-1)}{p}} \cdot \beta_W^{\frac{2 (p-1)^2}{p}}}\right)$ potential aligned for the potential
    \[ \overline{\Phi}(Z) = \sum_{i=1}^n \sum_{b=1}^s \left(\phi_{\left((Z_i)_b\right)^{(0)}}^{(p)}\right)^{\frac{p-1}{p-q}-1} \phi_{(Z_i)_b}^{(p)}\,.\]
\end{lemma}

\begin{remark}
    You can think of $Z_i$ and $W_j$ as the vectorized version of $(X_i + U_i)^{\otimes a} \otimes (X_i - U_i)^{\otimes (q-a)}$ and $(Y_j - V_j)^{\otimes a} \otimes (Y_j + V_j)^{\otimes (p-a)}$. These only depend on a single $U_i$ or $V_j$ (out of the indeterminates), and thus the error of independently sampling each vector with correlations within the coordinates does match the error of our rounding algorithm. In the final analysis, we will need to slightly modify this to ensure that we have bounds on $\beta_Z$ and $\beta_W$ which suffice to achieve $\eps \cdot \OPT$ rounding error.
\end{remark}

\begin{proof}
    The proof proceeds via an application of~\cref{lem:general-p-rounding-err}. Note that applying the lemma coordinate-wise and summing over coordinates, we have that for
    \[ \Phi(Z_b) = \sum_i \left(\phi_{\left((Z_i)_b\right)^{(0)}}^{(p)}\right)^{\frac{p-1}{p-q}-1} \phi_{(Z_i)_b}^{(p)}\,,\]
    there exists a distribution $\psi_b$ for each $b$ such that the error satisfies
    \begin{align*}
        \err &\leq O(1) \cdot \sum_{b = 1}^s \left(\E_{\{\mu\}} \left[ \sum_j \pE \left[\vert \left(W_j\right)_b\vert^{\frac p {p-q}}\right] \right]\right)^{\frac{p-q}{2(p-1)^2}} \\
        &\qquad\qquad\left( \E_{\{\mu\}} \left[ \sum_{ij} \left( \pE \left[\vert (Z_i)_b\vert^{\frac p {p-q}} \vert (W_j)_b\vert^{\frac p {p-q}}\right] + \pE \left[\vert (Z_i)_b\vert^{\frac p {p-q}}\right] \pE \left[\vert (W_j)_b\vert^{\frac p {p-q}}\right] \right) \right] \right)^{\frac{(p-q)(p-2)(2p-1)}{2p(p-1)^2}} \\
        &\qquad\qquad\norm{B}_p^q \left( \E_{\{\mu\}}\E_{j \sim \psi_b} \left( \Phi(Z_{b}) - \E_{W_j} \Phi(Z_b \vert W_j) \right)\right)^{\frac{p-q}{2(p-1)^2}}\,.
    \end{align*}
    By applying H\"olders with $\frac{2 (p-1)^2 p}{(2 - 3 p) q + p^2 (2 q-1)}$, $\frac{2(p-1)^2}{p-q}$, $\frac{2p(p-1)^2}{(p-q)(p-2)(2p-1)}$, and $\frac{2(p-1)^2}{p-q}$ we have that\footnote{Note that the first factor is at least $1$ when $q \geq 1$.} 
    \begin{align*}
        \err &\leq O(1) \cdot s^{\frac{(2 - 3 p) q + p^2 (2 q-1)}{2 (p-1)^2 p}} \cdot  \left(\sum_{b = 1}^s \E_{\{\mu\}} \left[ \sum_j \pE \left[\vert \left(W_j\right)_b\vert^{\frac p {p-q}}\right] \right]\right)^{\frac{p-q}{2(p-1)^2}} \\
        &\qquad\qquad\left( \sum_{b = 1}^s \left[\E_{\{\mu\}} \left[ \sum_{ij} \left( \pE \left[\vert (Z_i)_b\vert^{\frac p {p-q}} \vert (W_j)_b\vert^{\frac p {p-q}}\right] + \pE \left[\vert (Z_i)_b\vert^{\frac p {p-q}}\right] \pE \left[\vert (W_j)_b\vert^{\frac p {p-q}}\right] \right) \right] \right]\right)^{\frac{(p-q)(p-2)(2p-1)}{2p(p-1)^2}} \\
        &\qquad\qquad\norm{B}_p^q \left( \sum_{b = 1}^s \E_{\{\mu\}}\E_{j \sim \psi_b} \left( \Phi(Z_{b}) - \E_{W_j} \Phi(Z_b \vert W_j) \right)\right)^{\frac{p-q}{2(p-1)^2}}\,.
    \end{align*}
    Let $\overline{\psi}$ be the marginal distribution on $j$ obtained by choosing $b$ uniformly from $[s]$ and then sampling $j \sim \psi_b$. Since conditioning only decreases each coordinate's Bregman uncertainty in expectation, we have that
    \[
        \sum_{b=1}^s \E_{\{\mu\}}\E_{j \sim \psi_b} \left( \Phi(Z_b) - \E_{W_j} \Phi(Z_b \vert W_j) \right)
        \leq s \cdot \E_{\{\mu\}}\E_{j \sim \overline{\psi}} \left( \overline{\Phi}(Z) - \E_{W_j} \overline{\Phi}(Z \vert W_j) \right)\,.
    \]
    Therefore,
    \begin{align*}
        \err &\leq s^{O_p(1)} \cdot \left( \E_{\{\mu\}} \pE \norm{W}_{\frac p {p-q}}^{\frac p {p-q}} \right)^{\frac{p-q}{2(p-1)^2}} \left( \E_{\{\mu\}} \left(\pE \left[ \norm{Z}_{\frac p {p-q}}^{\frac p {p-q}} \norm{W}_{\frac p {p-q}}^{\frac p {p-q}} \right] + \pE \left[ \norm{Z}_{\frac p {p-q}}^{\frac p {p-q}} \right] \pE \left[ \norm{W}_{\frac p {p-q}}^{\frac p {p-q}} \right]\right)\right)^{\frac{(p-q)(p-2)(2p-1)}{2p(p-1)^2}} \\
        &\qquad \qquad \norm{B}_p^q \left( \E_{\{\mu\}}\E_{j \sim \overline{\psi}} \left( \overline{\Phi}(Z) - \E_{W_j} \overline{\Phi}(Z \vert W_j) \right)\right)^{\frac{p-q}{2(p-1)^2}}\,.
    \end{align*}
    Thus, whenever the error is at least $\gamma$ we have that 
    \begin{align*}
        &\E_{\{\mu\}}\E_{j \sim \overline{\psi}} \left( \overline{\Phi}(Z) - \E_{W_j} \overline{\Phi}(Z \vert W_j) \right)\geq \frac{\gamma^{\frac{2(p-1)^2}{p-q}}}{s^{O_p(q)}\cdot \norm{B}_p^{\frac{2q\cdot (p-1)^2}{p-q}} \cdot \left( \E_{\{\mu\}} \pE \norm{W}_{\frac p {p-q}}^{\frac p {p-q}} \right) \cdot \alpha^{\frac{(p-2)(2p-1)}{p}}}\,,
    \end{align*}
    where $\alpha = \E_{\{\mu\}} \left(\pE \left[ \norm{Z}_{\frac p {p-q}}^{\frac p {p-q}} \norm{W}_{\frac p {p-q}}^{\frac p {p-q}} \right] + \pE \left[ \norm{Z}_{\frac p {p-q}}^{\frac p {p-q}} \right] \pE \left[ \norm{W}_{\frac p {p-q}}^{\frac p {p-q}} \right]\right)$. Substituting the bounds on the norms of $Z,W$ from the pseudodistribution constraints yields the claim.
\end{proof}

\paragraph{Rounding Error Analysis.} We now return to analyzing the rounding error of our algorithm and a proof of~\cref{lem:high-rank-rounding}.

\begin{proof}[Proof of~\cref{lem:high-rank-rounding}]
    Note that we can decompose the objective function as follows:
    \[ \norm{A - UV^\top}_p^p = \norm{A - XY^\top + XY^\top - UV^\top}_p^p\,.\]
    Since for any invertible matrices $N_1, N_2$ we have that
    \begin{align*}
        XY^\top - UV^\top &= \frac 1 2 \left[(X+U)(Y-V)^\top + (X-U)(Y+V)^\top\right] \\
        &= \frac 1 2 \left[(X+U)N_1 N_1^{-1}(Y-V)^\top + (X-U)N_2 N_2^{-1} (Y+V)^\top\right]\,,
    \end{align*}
    we have that for the sampled indeterminates $N_1, N_2$ the objective function is equal to
    \begin{align*}
        \norm{A - UV^\top}_p^p &= \left\Vert A - XY^\top + (X+U)N_1 N_1^{-1}(Y-V)^\top + (X-U)N_2 N_2^{-1} (Y+V)^\top\right\Vert_p^p \\
        &= \sum_{a=0}^p \sum_{b=0}^{p-a} C_{a,b,p} \sum_{i,j} (A_{ij} - X_i^\top Y_j)^a \left\langle \left(Z_{a,b}\right)_i, \left(W_{a,b}\right)_j\right\rangle
    \end{align*}
    for some constants $C_{a,b,p}$ where 
    \begin{align*}
        \left(Z_{a,b}\right)_i &= \left(N_1^\top (X_i+U_i)\right)^{\otimes b} \otimes \left( N_2^\top (X_i-U_i)\right)^{\otimes (p-a-b)} \\
        \left(W_{a,b}\right)_j &= \left(N_1^{-1}(Y_j-V_j)\right)^{\otimes b} \otimes \left(N_2^{-1}(Y_j + V_j)\right)^{\otimes (p-a-b)}\,.
    \end{align*}
    Note that once we condition on the values of $N_1, N_2$ we have that the distribution of $\left(Z_{a,b}\right)_i$ for a single fixed $i$ are preserved by our independent sampling procedure, and similarly for $\left(W_{a,b}\right)_j$ for a single fixed $j$.
    
    Now, observe that for the part of the objective where $a=0$, we have that the rounding algorithm incurs no error since 
    \begin{align*}
        &\sum_{b=0}^p \sum_{i,j} \left\langle \left(N_1^\top (X_i+U_i)\right)^{\otimes b} \otimes \left( N_2^\top (X_i-U_i)\right)^{\otimes (p-b)}, \left(N_1^{-1}(Y_j-V_j)\right)^{\otimes b} \otimes \left(N_2^{-1}(Y_j + V_j\right)^{\otimes (p-b)} \right\rangle\\
        &= \sum_{\alpha} \sum_{b=0}^p \left(\sum_{i} \left(\left(N_1^\top (X_i+U_i)\right)^{\otimes b}\otimes \left( N_2^\top (X_i-U_i)\right)^{\otimes (p-b)}\right)_\alpha\right) \\
        &\qquad\qquad \left( \sum_j \left(\left(N_1^{-1}(Y_j-V_j)\right)^{\otimes b} \otimes \left(N_2^{-1}(Y_j + V_j\right)^{\otimes (p-b)}\right)_{\alpha}\right)\,.
    \end{align*}
    Since we guessed all $k^{O_p(1)}$ quantities $\sum_{i} \left(\left(N_1^\top (X_i+U_i)\right)^{\otimes (p-b)}\otimes \left( N_2^\top (X_i-U_i)\right)^{\otimes (p-b)}\right)_\alpha$ and $\sum_j \left(\left(N_1^{-1}(Y_j-V_j)\right)^{\otimes b} \otimes \left(N_2^{-1}(Y_j + V_j\right)^{\otimes (p-b)}\right)_{\alpha}$ this expression has a fixed value regardless of whether the vectors $U_i, V_j$ are sampled jointly or independently. Thus, the $a=0$ term incurs no rounding error. We turn our attention to the remaining terms. Since $\norm{f}_{\frac p {p-a}} \leq \poly_p(k) \cdot \norm{f}_p$ for vectors $f$ of dimension at most $\poly_p(k)$, we have that 
    \begin{align*}
        \norm{Z_{a,b}}_{\frac{p}{p-a}}^{\frac{p}{p-a}} &= \sum_{i \in [n]} \norm{\left(N_1^\top (X_i+U_i)\right)^{\otimes b} \otimes \left( N_2^\top (X_i-U_i)\right)^{\otimes (p-a-b)}}_{\frac{p}{p-a}}^{\frac{p}{p-a}}\,,
    \end{align*}
    and 
    \begin{align*}
        \norm{W_{a,b}}_{\frac{p}{p-a}}^{\frac{p}{p-a}} &= \sum_{j \in [n]} \left\Vert\left(N_1^{-1} (Y_j - V_j)\right)^{\otimes b} \otimes \left( N_2^{-1} (Y_j + V_j)\right)^{\otimes (p-a-b)}\right\Vert_{\frac{p}{p-a}}^{\frac{p}{p-a}}\,.
    \end{align*}
    Thus, by our constraints, the pseudodistribution satisfies
    \[ \norm{Z_{a,b}}_{\frac{p}{p-a}}^{\frac{p}{p-a}}, \norm{W_{a,b}}_{\frac{p}{p-a}}^{\frac{p}{p-a}} \leq k^{O_p(1)} \cdot \sqrt{\OPT}\,,\]
    Let $S_{i,\alpha} = (Z_{a,b})_{i,\alpha}$ and let $S_{i,\alpha}^{(0)}$ denote its distribution under the initial pseudodistribution. Thus, by~\cref{lem:lra-vec-pinning}\footnote{Note that conditioning on $V_j$ also conditions on the value of $W_j$ respectively it is a deterministic function of $V_j$ once we fix $N_1$ and $N_2$. Thus~\cref{lem:lra-vec-pinning} applies when conditioning on $V_j$.} we have that the polynomial
    \[ \sum_{i,j} (A_{ij} - X_i^\top Y_j)^a \left\langle \left(Z_{a,b}\right)_i, \left(W_{a,b}\right)_j\right\rangle \]
    is $\left(\eps \cdot \OPT, \frac{\eps^{O_p(1)}}{k^{O_p(1)}}\cdot \sqrt{\OPT} \right)$-potential aligned for the potential
    \[ \Phi(Z_{a,b}) = \sum_{i,\alpha} \left(\phi_{S_{i,\alpha}^{(0)}}^{(p)}\right)^{\frac{p-1}{p-a}-1} \phi_{S_{i,\alpha}}^{(p)}\,.\]
    Furthermore, by~\cref{fact:gen-initial-potential-bound} the potential is bounded initially by
    \[ \Phi(Z_{a,b}^{(0)}) = \sum_{i,\alpha} \left(\phi_{S_{i,\alpha}^{(0)}}^{(p)}\right)^{\frac{p-1}{p-a}} \leq \norm{Z_{a,b}}_{\frac{p}{p-a}}^{\frac{p}{p-a}} \leq k^{O_p(1)} \sqrt{\OPT} \,.\]
    Thus, by~\cref{lem:gcr_linear_combo} we have that after $O_p(1) \cdot \frac{k^{O_p(1)}}{\eps^{O_p(1)}}$ rounds of conditioning the total rounding error is at most $\eps \cdot \OPT$ as desired.
\end{proof}

\subsection{Proof of~\cref{thm:rank-k-lra-main}}

We now return to the proof of the main theorem of this section and put together all the pieces from above.
\begin{proof}[Proof of~\cref{thm:rank-k-lra-main}]
    We first start by showing that with probability $0.99$ the algorithm outputs an $(1+O(\eps))\cdot \OPT$ valued solution. Note that by~\cref{lem:rank-k-feasibility} that for some iteration of the loop with probability at least $0.999$ the SDP has optimal value at most $(1+O(\eps))\cdot \OPT$. Furthermore, when we round this solution by~\cref{lem:high-rank-rounding} we have that the expected value of the solution is $(1+O(\eps))\cdot \OPT$. Thus, by Markov and a union bound (over the randomness in the rounding algorithm and the randomness in the SDP optimal value), with probability at least $0.99$ we produce a solution of value at most $(1+O(\eps))\cdot \OPT$. 

    We now consider the runtime of the algorithm. The runtime is dominated by the cost of the loop in step $6$. There are $\left(\frac{kmn\cdot 2^{B}}{\eps \cdot \OPT}\right)^{O_p(1)}$ rounds of the loop. By our lower bound on $\OPT$ via~\cref{lem:lb-opt-rank-k} we have that this is at most $\left(\eps^{-1} \cdot kmn \cdot 2^{O(B)}\right)^{O_p(k)}$ iterations. We now consider the cost of a single step of the loop. This is dominated by the cost of finding an optimal SoS pseudodistribution over $\left(\frac{kmn\cdot 2^{B}}{\eps \cdot \OPT}\right)^{O_p(1)} \leq \left(\eps^{-1} \cdot kmn \cdot 2^{O(B)}\right)^{O_p(k)}$ variables with degree $(k/\eps)^{O_p(1)}$. By~\cref{fact:sos-separation-efficient} this can be done in time $\left(mn \cdot 2^{O(B)}\right)^{(k/\eps)^{O_p(1)}}$. Thus, the overall algorithm runs in time $\left(mn \cdot 2^{O(B)}\right)^{(k/\eps)^{O_p(1)}}$.
\end{proof}

\section{Additive Approximation for $p\rightarrow q$ Norms}

In this section, we prove the following theorem:
\begin{theorem}
    \label{thm:p_q_norm_main}
    Let $q \geq 2$ be an even integer with H\"older dual $q^*$, and let $p \in [q^*,2]$ have H\"older dual $p^* = \tfrac p {p-1}$ which is an even integer.
    For every $\eps > 0$, there is a $(nm)^{O\left(\eps^{-4q}\right)}$ time algorithm which takes a matrix $A \in \R^{n \times m}$ and with probability $0.99$ returns $x \in \R^m$ such that
    \[
    \frac{\|Ax\|_q^q}{\|x\|_p^q} \leq \|A\|_{p \rightarrow q}^q \leq \frac{\|Ax\|_q^q}{\|x\|_p^q} + \eps \cdot \sum_{i \leq n} \|A_i\|_{p^*}^q
    \]
    where $A_1,\ldots,A_n$ are the rows of $A$.
\end{theorem}

\begin{mdframed}
  \begin{algorithm}[$p\rightarrow q$ Norm Additive Approximation]
    \label{algo:p-to-q}\mbox{}
    \begin{description}
    \item[Input:] Matrix $A \in \mathbb{R}^{n \times m}$ and target accuracy $0<\epsilon<1$.
    
    \item[Operations:]\mbox{}
    \begin{enumerate}
        \item Let $\Sigma$ be a $1/(\poly(nm/\eps))$ grid of $[-1,1]$.
        \item Let $\mu$ be a $O\left(\eps^{-4q} \right)$-degree sum-of-squares pseudo-distribution over $\left(\Sigma\right)^{2n+2m}$ such that $\pE_{\mu}$ optimizes 
                    \begin{align}
                        &\max_{\pE_{\mu}} \quad \pE_{\mu} \sum_{i\in[n],j\in[m]} A_{ij} v_i^{q-1} u_j^{p^*-1} \notag \\
                        &\text{s.t.\ } \pE_{\mu} \text{ satisfies} \notag \\
                        &\quad \norm{v}_q^q \leq 1, \norm{u}_{p^*}^{p^*} \leq 1 \\
                        &\quad y_i = v_i^{q-1}, x_j = u_j^{p^*-1}  
                    \end{align}        
        \item For each set $\calT \subseteq [n]$ of size $O(\eps^{-4q})$
        \begin{enumerate}
            \item Let $\hat{v}_\calT$ be a draw from the local distribution $\{v_{\calT}\}$. Let $\mu_\calT$ be the pseudodistribution obtained by conditioning on $\{v_i = \hat{v}_i\}_{i \in \calT}$.
            \item Let $\hat{x}$ be $\pE x$ in the conditioned pseudodistribution.
        \end{enumerate}
    \end{enumerate}
    \item[Output:] The vector $\frac {\hat{x}}{\norm{\hat{x}}_p}$ maximizing $\frac{\norm{A\hat{x}}_q^q}{\norm{\hat{x}}_p^q}$ out of all iterations of the loop.
    \end{description}
  \end{algorithm}
\end{mdframed}

\begin{remark}
    Although the above algorithm utilizes sum-of-squares pseudodistributions, we note that the Sherali-Adams linear programming hierarchy alongside global linear constraints would suffice. We use sum-of-squares pseudodistributions and SDPs for ease of exposition.
\end{remark}
\paragraph{Discretization.}
We note that it suffices to let $\Sigma$ be a $1/(\poly(nm/\eps))$-fineness grid of $[-1,1]$. With this choice of $\Sigma$ it holds that 
\[ \max_{\norm{x}_p \leq 1, \norm{y}_{q^*} \leq 1} \langle y, Ax\rangle \leq \max_{\norm{x}_p \leq 1, \norm{y}_{q^*} \leq 1, x_i \in \Sigma, y_j \in \Sigma} \langle y, Ax\rangle + \eps \cdot \left(\sum_{i \leq n} \|A_i\|_{p^*}^q\right)^{1/q}\,.\]
We defer the proof to~\cref{sec:appendix}.

\paragraph{Potential Alignment of the Objective.} We now turn our attention to analyzing the rounding error via the potential alignment framework. We first simply consider the objective cost achieved by the rounded solution without considering the norms of the outputted $\hat{x}, \hat{y}$. Let $p \leq 2$ have H\"older dual $p^*$, which is an even integer. Recall that we are trying to round the following expression under the constraint that $\norm{u}_{p*}, \norm{v}_q \leq 1$:
\[ \sum_{i \in [n], j \in [m]} A_{ij} v_i^{q-1} u_j^{p^*-1}\,.\]
We will show that global correlation rounding preserves the objective value up to an additive factor of $\eps \left(\sum_{i \leq n} \|A_i\|_{p^*}^q\right)^{1/q}$.

\begin{lemma}
\label{lem:p-to-q-potential-aligned}
    Let $p^*,q$ be even integers and let $\mu$ be any sum-of-squares pseudodistribution satisfying the constraints $\norm{u}_{p^*}^{p^*}, \norm{v}_q^q \leq 1$. Then if $p^* \leq q$ then the objective
    \[ \sum_{i \in [n], j \in [m]} A_{ij} v_i^{q-1} u_j^{p^*-1} \]
    is $\left[\eps \left( \sum_i \norm{A_i}_{p^*}^q \right)^{\frac 1 q}, \eps^{4q}\right]$-potential aligned for the potential function $\sum_j \phi_{x_j}^{(p^*)}$.
\end{lemma}
\begin{proof}
    We consider the case where $p^* \leq q$, the other case is analogous. Let 
    \[ \err = \sum_{i \in [n], j \in [m]} A_{ij} \E_{\mu_{u} \otimes \mu_v} v_i^{q-1} u_j^{p^*-1} - \sum_{i \in [n], j \in [m]} A_{ij} \E_{\mu_{u,v}} v_i^{q-1} u_j^{p^*-1} \geq \eps \left( \sum_i \norm{A_i}_{p^*}^q \right)^{\frac 1 q}\,.\]
    We now show that this implies a $\poly(1/\eps)$ decrease in the Bregman uncertainty potential function $\sum_j \phi_{u_j^{p^*-1}}^{(p^*)}$. Let 
    \[ \delta(j \vert i) \vcentcolon= \E_{u_j,v_i} \left( B_{p^*}( u_j^{p^*-1}, \E u_j^{p^*-1}) - B_{p^*} ( u_j^{p^*-1},  \E [u_j^{p^*-1} \vert v_i])  \right)\,,\]
    and let
    \[ \alpha_{ij} \vcentcolon= \E [u_j^{p^*} v_i^q] + \E u_j^{p^*} \E v_i^q \,.\]
    Note by~\cref{lem:gen-p-q-norm-potential-decrease} we have that
    \[ \vert \Cov (v_i^{q-1}, u_j^{p^*-1}) \vert \leq \delta(j \vert i)^{\frac{{p^*}}{2q({p^*}-1)}} \left( \E v_i^q \right)^{\frac{(q-{p^*})({p^*}-2)+{p^*}q}{2{p^*}q({p^*}-1)}} \cdot \alpha_{ij}^{\frac{2 {p^*}q \left({p^*}-2\right) - \left(\left(p^*\right)^2 - 2q\right)}{2{p^*}q({p^*}-1)}}\,.\]
    Thus, we have that
    \begin{align*}
        \err &= \sum_{i \in [n], j \in [m]} A_{ij} \Cov (v_i^{q-1}, u_j^{p^*-1}) \\
        &\leq \sum_{i \in [n], j \in [m]} \vert A_{ij}\vert \cdot \vert \Cov (v_i^{q-1}, u_j^{p^*-1}) \vert \\
        &\leq \sum_{i \in [n], j \in [m]} \vert A_{ij}\vert \cdot \delta(j \vert i)^{\frac{{p^*}}{2q({p^*}-1)}} \cdot \left( \E v_i^q \right)^{\frac{(q-{p^*})({p^*}-2)+{p^*}q}{2{p^*}q({p^*}-1)}} \cdot \alpha_{ij}^{\frac{2 {p^*}q \left({p^*}-2\right) - \left(\left(p^*\right)^2 - 2q\right)}{2{p^*}q({p^*}-1)}} \\
        &= \sum_{i \in [n]} \left( \E v_i^q \right)^{\frac{(q-{p^*})({p^*}-2)+{p^*}q}{2{p^*}q({p^*}-1)}} \sum_{j \in [m]}  \vert A_{ij}\vert \cdot \delta(j \vert i)^{\frac{{p^*}}{2q({p^*}-1)}} \cdot \alpha_{ij}^{\frac{2 {p^*}q \left({p^*}-2\right) - \left(\left(p^*\right)^2 - 2q\right)}{2{p^*}q({p^*}-1)}}\,.
    \end{align*}
    We now apply H\"olders with $\frac{2{p^*}q({p^*}-1)}{(q-{p^*})({p^*}-2)+{p^*}q}$ and $\frac{2 {p^*} q \left({p^*}-1\right)}{2 {p^*}q \left({p^*}-2\right) + \left(\left(p^*\right)^2 - 2{p^*} + 2q\right)}$ to the outer sum to get that
    \begin{align*}
        \err \leq &\left( \sum_{i \in [n]} \E v_i^q \right)^{\frac{(q-{p^*})({p^*}-2)+{p^*}q}{2{p^*}q({p^*}-1)}} \\
        &\quad\left(\sum_{i \in [n]} \left( \sum_{j \in [m]}  \vert A_{ij}\vert \cdot \delta(j \vert i)^{\frac{{p^*}}{2q({p^*}-1)}} \cdot \alpha_{ij}^{\frac{2 {p^*}q \left({p^*}-2\right) - \left(\left(p^*\right)^2 - 2q\right)}{2{p^*}q({p^*}-1)}} \right)^{\frac{2 {p^*} q ({p^*}-1)}{2 {p^*}q ({p^*}-2) + ({p^*}^2 - 2{p^*} + 2q)}} \right)^{\frac{2 {p^*}q \left({p^*}-2\right) + \left(\left(p^*\right)^2 - 2{p^*} + 2q\right)}{2 {p^*} q \left({p^*}-1\right)}}\,.
    \end{align*}
    Furthermore, since we enforce that $\norm{v}_q^q \leq 1$ we have that $\sum_i \E v_i^q \leq 1$ and thus
    \[ \err \leq \left(\sum_{i \in [n]} \left( \sum_{j \in [m]}  \vert A_{ij}\vert \cdot \delta(j \vert i)^{\frac{{p^*}}{2q({p^*}-1)}} \cdot \alpha_{ij}^{\frac{2 {p^*}q \left({p^*}-2\right) - \left(\left(p^*\right)^2 - 2q\right)}{2{p^*}q({p^*}-1)}} \right)^{\frac{2 {p^*} q \left({p^*}-1\right)}{2 {p^*}q \left({p^*}-2\right) + \left(\left(p^*\right)^2 - 2{p^*} + 2q\right)}} \right)^{\frac{2 {p^*}q \left({p^*}-2\right) + \left(\left(p^*\right)^2 - 2{p^*} + 2q\right)}{2 {p^*} q \left({p^*}-1\right)}}\,.\]
    We now apply H\"olders on the inner sum over $j$ with $p^*$, $\frac{2 q \left(p^*-1\right)}{p^*}$, and $\frac{2{p^*}q({p^*}-1)}{2 {p^*}q \left({p^*}-2\right) - \left(\left(p^*\right)^2 - 2q\right)}$. This gives us that
    \begin{align*}
        \err \leq \Biggl( \sum_{i \in [n]} &\left( \sum_{j \in [m]} A_{ij}^{p^*} \right)^{\frac{2 q \left({p^*}-1\right)}{2 {p^*}q \left({p^*}-2\right) + \left(\left(p^*\right)^2 - 2{p^*} + 2q\right)}} \left( \sum_{j \in [m]} \delta(j \vert i) \right)^{\frac{{\left(p^*\right)^2} }{2 {p^*}q \left({p^*}-2\right) + \left(\left(p^*\right)^2 - 2{p^*} + 2q\right)}} \\
        &\quad \left( \sum_{j \in [m]} \alpha_{ij}\right)^{\frac{2 p^* q (p^*-2) - ((p^*)^2 - 2q)}{2 p^* q (p^*-2) + ((p^*)^2 - 2 p^* + 2 q)}}\Biggr)^{\frac{2 {p^*}q \left({p^*}-2\right) + \left(\left(p^*\right)^2 - 2{p^*} + 2q\right)}{2 {p^*} q \left({p^*}-1\right)}}\,.
    \end{align*}
    We now apply H\"olders again on the outer sum with $\frac{2 p^* q (p^*-2) + ((p^*)^2 - 2 p^* + 2 q)}{2 p^* (p^*-1)}$ and $\frac{2 p^* q (p^*-2) + ((p^*)^2 - 2 p^* + 2 q)}{2 p^* q (p^*-2) - ((p^*)^2 - 2q)}$ to get that
    \begin{align*}
        \err \leq &\left( \sum_{i \in [n]} \left( \sum_{j \in [m]} A_{ij}^{p^*} \right)^{\frac q {p^*}} \left( \sum_{j \in [m]} \delta(j \vert i) \right)^{\frac {p^*}{2(p^*-1)}}\right)^{\frac 1 q} \left( \sum_{i \in [n], j \in [m]} \alpha_{ij} \right)^{\frac{2 p^* q (p^*-2) - ((p^*)^2 - 2q)}{2 {p^*} q \left({p^*}-1\right)}}\,.
    \end{align*}
    Note that via our constraints, we have that
    \[ \sum_{i \in [n], j \in [m]} \alpha_{ij} = \sum_{i \in [n], j \in [m]} \E [u_j^{p^*} v_i^q] + \E u_j^{p^*} \E v_i^q = \E \left[ \norm{u}_{p^*}^{p^*} \norm{v}_q^q \right] + \E \norm{u}_{p^*}^{p^*} \E \norm{v}_q^q \leq 2 \,.\]
    Therefore,
    \begin{align*}
        \err \leq  O(1) \cdot \left( \sum_{i \in [n]} \left( \sum_{j \in [m]} A_{ij}^{p^*} \right)^{\frac q {p^*}} \left( \sum_{j \in [m]} \delta(j \vert i) \right)^{\frac {p^*}{2(p^*-1)}}\right)^{\frac 1 q}\,.
    \end{align*}
    Let $A_i$ be the $i$th row of $A$ and let the distribution $\psi$ be defined such that $\psi(i) \propto \norm{A_i}_{p^*}^q$. Then we have that 
    \[\left( \sum_{i \in [n]} \left( \sum_{j \in [m]} A_{ij}^{p^*} \right)^{\frac q {p^*}} \left( \sum_{j \in [m]} \delta(j \vert i) \right)^{\frac {p^*}{2(p^*-1)}}\right)^{\frac 1 q} = \left(\sum_i \norm{A_i}_{p^*}^q\right)^{\frac 1 q} \left(\E_{i \sim \psi} \left( \sum_{j \in [m]} \delta(j \vert i) \right)^{\frac {p^*}{2(p^*-1)}}\right)^{\frac 1 q}\,.\]
    Note that $\frac {p^*}{2(p^*-1)} \leq 1$ and thus by convexity we have that
    \[ \E_{i \sim \psi} \left( \sum_{j \in [m]} \delta(j \vert i) \right)^{\frac {p^*}{2(p^*-1)}} \leq \left(\E_{i \sim \psi} \left( \sum_{j \in [m]} \delta(j \vert i) \right)\right)^{\frac {p^*}{2(p^*-1)}}\,.\]
    When $\err \geq \eps \left( \sum_i \norm{A_i}_{p^*}^q \right)^{\frac 1 q}$ we therefore have that
    \[ \eps \left( \sum_i \norm{A_i}_{p^*}^q \right)^{\frac 1 q} \leq \left(\sum_i \norm{A_i}_{p^*}^q\right)^{\frac 1 q} \left( \E_{i \sim \psi} \sum_{j \in [m]} \delta(j \vert i) \right)^{\frac {p^*}{2q(p^*-1)}}\,,\]
    and thus
    \[ \E_{i \sim \psi} \sum_{j \in [m]} \delta(j \vert i) \geq \eps^{\frac{2q(p^*-1)}{p^*}} \geq \eps^{4q}\,.\]
\end{proof}

\paragraph{Final Approximation Guarantee.}
We now argue that the output vector $\hat{x}$ at the end of the algorithm in fact witnesses a lower bound on the $p\rightarrow q$ norm with additive error $\eps \cdot \sum_{i \leq n} \|A_i\|_{p^*}^q$. 

\begin{proof}[Proof of~\cref{thm:p_q_norm_main}]
    Note that by~\cref{lem:p-q-disc} we have that some iteration of our loop has SDP value at least $\norm{A}_{p \rightarrow q} - \eps \cdot \left(\sum_{i \leq n} \|A_i\|_{p^*}^q\right)^{1/q}$. Consider the $\hat{x}$ and $\hat{y} = \pE y$ produced by this iteration of the loop. Note that since we are rounding a bilinear form, outputting the expectation of $x,y$ produces the same objective value as the expected objective when rounding each coordinate independently. Note that by~\cref{lem:p-to-q-potential-aligned} and the fact that the potential function is at most $1$ by norm bounds, we have that
    \[ \E_{\{\mu\}} \langle \hat{y}, A \hat{x}\rangle \geq \norm{A}_{p \rightarrow q} - \eps \cdot \left(\sum_{i \leq n} \|A_i\|_{p^*}^q\right)^{1/q}\,,\]
    where the expectation on the LHS is over the randomness in the conditioning. Furthermore, note that we always have that $\norm{\hat{x}}_p, \norm{\hat{y}}_{q^*} \leq 1$ by Jensen's since for $\hat{x}$
    \[ \norm{\hat{x}}_p^p = \sum_j \|\pE x_j\|^p \leq \sum_j \pE \|x_j\|^p = \pE \norm{x}_p^p \leq 1\,,\]
    and similarly for $\hat{y}$ we also have that $ \norm{\hat{y}}_{q^*} \leq 1$. Thus, $\langle \hat{y}, A \hat{x}\rangle \leq \norm{A}_{p \rightarrow q}$ with probability $1$, and thus by Markov's we have with probability $0.99$ that
    \[ \langle \hat{y}, A \hat{x}\rangle \geq \norm{A}_{p \rightarrow q} - 100\eps \cdot \left(\sum_{i \leq n} \|A_i\|_{p^*}^q\right)^{1/q}\,.\]
    Note that since $\hat{x}, \hat{y}$ have the relevant norms at most $1$, renormalizing only increases the objective function. Let these renormalized vectors be $\tilde{x}, \tilde{y}$. We thus have that 
    \[ \langle \tilde{y}, A \tilde{x}\rangle \geq \norm{A}_{p \rightarrow q} - 100\eps \cdot \left(\sum_{i \leq n} \|A_i\|_{p^*}^q\right)^{1/q}\]
    for $\tilde{x}, \tilde{y}$ satisfying $\norm{\tilde{x}}_p = 1$ and $\norm{\tilde{y}}_{q^*} = 1$. Therefore, 
    \[ \norm{A \tilde{x}}_q \geq \norm{A}_{p \rightarrow q} - 100\eps \cdot \left(\sum_{i \leq n} \|A_i\|_{p^*}^q\right)^{1/q}\,,\]
    and taking $q$th powers along with the fact that $\norm{A}_{p \rightarrow q} \leq \left(\sum_{i \leq n} \|A_i\|_{p^*}^q\right)^{1/q}$ (via H\"olders) completes the proof of the approximation guarantee. For the runtime analysis, we note that each iteration of the loop takes time $(nm)^{O\left(\eps^{-4q}\right)}$ by~\cref{fact:sos-separation-efficient} and there are $n^{O(\eps^{-4q})}$ iterations of the loop.
\end{proof}

\subsection{Comparison to Prior Work}
\label{sec:prior-work-pq-additive}
We compare the additive error $\sum_{i \leq n} \|A_i\|_{p^*}^q$ to the additive error obtained in the algorithms of \cite{barak2012hypercontractivity} and \cite{brandao2015replacing}.

\paragraph{The $p=2$ case}
When $p=2$, \cite{brandao2015replacing}, improving over \cite{barak2012hypercontractivity}, give a polynomial time algorithm for each fixed $\eps$ with additive error $\eps \cdot \|A\|_{2 \rightarrow 2}^{2} \cdot \|A\|_{2 \rightarrow \infty}^2$.
This is strictly incomparable to $\eps \cdot \sum_i \|A_i\|_2^q$ -- smaller if the row-norms of $A$ are all similar, and larger if $A$ has one very large row making $\|A\|_{2 \rightarrow \infty}$ very large.

\paragraph{The $p < 2$ case}
\cite{brandao2015replacing}'s Theorem 15 also provides a nontrivial additive approximation for $\|A\|_{p \rightarrow q}$ when $p < 2$, the only such algorithm we are aware of.
Specializing Theorem~15 of \cite{brandao2015replacing} to the map $A : \ell_p^m \to \ell_q^n$, the relevant quantity in their guarantee is the factorization norm of $A$ through $\ell_1$,
\[
\gamma_1(A) \vcentcolon= \norm{A}_{\ell_p \rightarrow \ell_1 \rightarrow \ell_q}
    = \inf_{A = RL} \norm{L}_{\ell_p \rightarrow \ell_1}\norm{R}_{\ell_1 \rightarrow \ell_q}.
\]
Their theorem, given such a factorization explicitly, gives additive error $\eps \cdot \gamma_1(A)$ for $\norm{A}_{p \rightarrow q}$, and hence, after adjusting constants depending only on $q$, additive error $\eps \cdot \gamma_1(A)^q$ for $\norm{A}_{p \rightarrow q}^q$.
Thus the comparison is between our error term $\eps \cdot \sum_{i \leq n} \norm{A_i}_{p^*}^q$ and the factorization-based quantity $\eps \cdot \gamma_1(A)^q$.

These two bounds are not ordered in general. If $A$ admits a genuinely good factorization through $\ell_1$, namely if
\[
\gamma_1(A)^q \ll \sum_{i \leq n} \norm{A_i}_{p^*}^q,
\]
then \cite{brandao2015replacing} can be substantially better. On the other hand, the trivial row decomposition
\[
Ax = \sum_{i \leq n} \langle A_i, x \rangle e_i
\]
immediately gives
\[
\gamma_1(A) \leq \norm{A}_{p \rightarrow 1} \leq \sum_{i \leq n} \norm{A_i}_{p^*}.
\]
Hence, whenever this trivial factorization is close to optimal, Theorem~15 yields error at most
\[
\eps \cdot \left( \sum_{i \leq n} \norm{A_i}_{p^*} \right)^q,
\]
whereas our error is $\eps \cdot \sum_{i \leq n} \norm{A_i}_{p^*}^q$. Since
\[
\sum_{i \leq n} \norm{A_i}_{p^*}^q \leq \left( \sum_{i \leq n} \norm{A_i}_{p^*} \right)^q,
\]
and the gap can be as large as $n^{q-1}$ when the row norms are comparable, our guarantee is better on instances whose row $p^*$-mass is spread across many rows and for which the trivial row factorization is near-optimal, while \cite{brandao2015replacing} is better on instances with a small $\ell_1$-factorization norm.

\section*{AI Disclosure:}
We used GPT 5.4 Pro as described in the introduction to this paper.
The tool materially affected Section~\ref{sec:higher-rank-lra}.
The authors verified the correctness and originality of all content including references.

\printbibliography

\appendix
\newpage
\section{Appendix}
\label{sec:appendix}

\subsection{Deferred Discretization Lemmas}
We give a proof of the deferred discretization lemmas from the previous sections below.
\paragraph{Rank $1$ LRA Discretization.}
\begin{proof}[Proof of~\cref{lem:discretization-error-bound} (LRA Discretization)]
    We note that by our bounds on $\norm{u}, \norm{v}$ that all entries of $u,v$ lie within the interval $\left[ - O(1) \cdot \norm{A}_p^{1/2}, O(1) \cdot \norm{A}_p^{1/2}\right]$. Thus, for all entries $u_i,v_j$ we have that $\vert u_i - \tilde{u}_i \vert \leq \gamma$ and $\vert v_j - \tilde{v}_j \vert \leq \gamma$. Note that we have that
    \begin{align*}
        \sum_{i \in [n], j \in [m]} \left(A_{ij} - u_i v_j\right)^p &= \sum_{i \in [n], j \in [m]} \left(A_{ij} \pm \tilde{u}_i \tilde{v}_j - u_i v_j\right)^p \\
        &=\sum_{i \in [n], j \in [m]} \left(A_{ij} - \tilde{u}_i \tilde{v}_j\right)^p + \sum_{q = 1}^p \binom{p}{q}\sum_{i \in [n], j \in [m]} \left(\tilde{u}_i \tilde{v}_j - u_i v_j\right)^q \left(A_{ij} - \tilde{u}_i \tilde{v}_j\right)^{p-q}\,.
    \end{align*}
    Thus, it suffices to give a bound on the second sum. By applying H\"olders with $\frac q q$ and $\frac {p-q} p$ we have that 
    \[ \sum_{i \in [n], j \in [m]} \left(\tilde{u}_i \tilde{v}_j - u_i v_j\right)^q \left(A_{ij} - \tilde{u}_i \tilde{v}_j\right)^{p-q} \leq \left(\sum_{i \in [n], j \in [m]} \left(\tilde{u}_i \tilde{v}_j - u_i v_j\right)^p\right)^{\frac q p} \left( \sum_{i \in [n], j \in [m]} \left(A_{ij} - \tilde{u}_i \tilde{v}_j\right)^{p}\right)^{\frac {p-q} p}\,.\]
    Note that 
    \[ \sum_{i \in [n], j \in [m]} \left(A_{ij} - \tilde{u}_i \tilde{v}_j\right)^{p} \leq O_p(1) \left( \norm{A}_p^p + \norm{\tilde{u} \tilde{v}^\top}_p^p\right) \leq O_p(1) \cdot \norm{A}_p^p\,.\]
    Thus, plugging this all in, we have that
    \begin{align*}
         \sum_{i \in [n], j \in [m]} \left(A_{ij} - u_i v_j\right)^p &=\sum_{i \in [n], j \in [m]} \left(A_{ij} - \tilde{u}_i \tilde{v}_j\right)^p + \sum_{q = 1}^p \binom{p}{q}\sum_{i \in [n], j \in [m]} \left(\tilde{u}_i \tilde{v}_j - u_i v_j\right)^q \left(A_{ij} - \tilde{u}_i \tilde{v}_j\right)^{p-q} \\
         &\leq \sum_{i \in [n], j \in [m]} \left(A_{ij} - \tilde{u}_i \tilde{v}_j\right)^p + O_p(1) \cdot \sum_{q = 1}^p \left(\sum_{i \in [n], j \in [m]} \left(\tilde{u}_i \tilde{v}_j - u_i v_j\right)^p\right)^{\frac q p} \norm{A}_p^{p-q}\,.
    \end{align*}
    Using that $\vert u_i - \tilde{u}_i \vert \leq \gamma$ and $\vert v_j - \tilde{v}_j \vert \leq \gamma$ along with that $\abs{u_i}, \abs{v_j} \leq O(1) \cdot \norm{A}_p^{1/2}$ we have that for all $i,j$,
    \[ \abs{\tilde{u}_i \tilde{v}_j - u_i v_j} \leq O(1) \cdot \gamma \cdot \norm{A}_p^{1/2}\,.\]
    Thus,
    \begin{align*}
         \sum_{i \in [n], j \in [m]} \left(A_{ij} - u_i v_j\right)^p &=\sum_{i \in [n], j \in [m]} \left(A_{ij} - \tilde{u}_i \tilde{v}_j\right)^p + \sum_{q = 1}^p \binom{p}{q}\sum_{i \in [n], j \in [m]} \left(\tilde{u}_i \tilde{v}_j - u_i v_j\right)^q \left(A_{ij} - \tilde{u}_i \tilde{v}_j\right)^{p-q} \\
         &\leq \sum_{i \in [n], j \in [m]} \left(A_{ij} - \tilde{u}_i \tilde{v}_j\right)^p + O_p(1) \cdot \sum_{q = 1}^p \left(nm \cdot \gamma \cdot \norm{A}_p^{1/2} \right)^{q} \norm{A}_p^{p-q}\,.
    \end{align*}
    For all $1 \leq q \leq p$ we have that
    \[ \left(nm \cdot \gamma \cdot \norm{A}_p^{1/2} \right)^{q} \norm{A}_p^{p-q} \leq \left(nm \cdot \gamma \cdot \norm{A}_p^{1/2} \right) \norm{A}_p^{p-1} + \left(nm \cdot \gamma \cdot \norm{A}_p^{1/2} \right)^{p-1} \norm{A}_p\,.\]
    When $\gamma \leq \norm{A}_p^{1/2}$ we have that $\left(nm \cdot \gamma \cdot \norm{A}_p^{1/2} \right) \norm{A}_p^{p-1} \leq \left(nm \cdot \gamma \cdot \norm{A}_p^{1/2} \right)^{p-1} \norm{A}_p$ and thus
    \begin{align*}
        \sum_{i \in [n], j \in [m]} \left(A_{ij} - u_i v_j\right)^p &\leq \sum_{i \in [n], j \in [m]} \left(A_{ij} - \tilde{u}_i \tilde{v}_j\right)^p + O_p(1) \cdot \sum_{q = 1}^p \left(nm \cdot \gamma \cdot \norm{A}_p^{1/2} \right)^{q} \norm{A}_p^{p-q} \\
        &\leq \sum_{i \in [n], j \in [m]} \left(A_{ij} - \tilde{u}_i \tilde{v}_j\right)^p + O_p(1) \cdot nm \cdot \gamma \cdot \norm{A}_p^{\frac{2p - 1}{2}}\,.
    \end{align*}
\end{proof}

\paragraph{Rank $k$ LRA Discretization.} We also show that the above lemma can be extended to the rank $k$ setting, which will be required for the higher rank analysis.

\begin{lemma}
\label{lem:rounding-rank-k-disc-solutions}
    Let $X \in \mathbb{R}^{n \times k}, Y \in \mathbb{R}^{m \times k}$ be solutions to the LRA problem on $A \in \mathbb{R}^{n \times m}$. Let $\widetilde{X}, \widetilde{Y}$ be rounded LRA solutions such that $\norm{X - \widetilde{X}}_{\infty}, \norm{Y - \widetilde{Y}}_\infty \leq \gamma$. Then if $\alpha = \norm{X}_\infty + \norm{Y}_\infty + \gamma$ we have that
    \[ \norm{A - \widetilde{X} \widetilde{Y}^\top}_p^p \leq \norm{A - XY^\top}_p^p + O_p(1) \cdot nm \cdot \left(\norm{A - XY^\top}_p^{p-1} \cdot k \cdot \gamma \cdot \alpha + k^p \cdot \gamma^p \cdot \alpha^p \right)\,.\]
    In particular, if $A, X, Y$ have entries with magnitude at most $2^B$ and $\gamma \leq 2^{-\Omega_p(B)} \cdot \frac{1}{\poly(nm)} \cdot k^{-p} \cdot \gamma'$ then we have that
    \[ \norm{A - \widetilde{X} \widetilde{Y}^\top}_p^p \leq \norm{A - XY^\top}_p^p + \gamma'\,.\]
\end{lemma}
\begin{proof}[Proof of~\cref{lem:rounding-rank-k-disc-solutions}]
    Let $\alpha = \norm{X}_\infty + \norm{Y}_\infty + \gamma$. Note that we have that 
    \begin{align*}
        \norm{A - \widetilde{X} \widetilde{Y}^\top}_p^p &= \sum_{i \in [n], j \in [m]} \left(A_{ij} - \langle \widetilde{X}_i, \widetilde{Y}_j\rangle\right)^p \\
        &= \sum_{i \in [n], j \in [m]} \left(A_{ij} \pm \langle X_i, Y_j \rangle - \langle \widetilde{X}_i, \widetilde{Y}_j\rangle\right)^p \\
        &= \norm{A - XY^\top}_p^p + \sum_{q=1}^p \binom{p}{q} \sum_{i \in [n], j \in [m]} \left( A_{ij} - \langle X_i, Y_j \rangle\right)^{p-q} \left(\langle X_i, Y_j \rangle - \langle \widetilde{X}_i, \widetilde{Y}_j\rangle \right)^q \\
        &\leq \norm{A - XY^\top}_p^p + \sum_{q=1}^p \binom{p}{q} \sum_{i \in [n], j \in [m]} \left\vert A_{ij} - \langle X_i, Y_j \rangle\right\vert^{p-q} \left(k \cdot \gamma \cdot \alpha\right)^q \,.
    \end{align*}
    Applying H\"older with $\frac{p}{p-q}$ and $\frac{p}{q}$ we have that
    \begin{align*}
        \norm{A - \widetilde{X} \widetilde{Y}^\top}_p^p &\leq \norm{A - XY^\top}_p^p + \sum_{q=1}^p \binom{p}{q} \sum_{i \in [n], j \in [m]} \left\vert A_{ij} - \langle X_i, Y_j \rangle\right\vert^{p-q} \left(k \cdot \gamma \cdot \alpha \right)^q \\
        &\leq \norm{A - XY^\top}_p^p + \sum_{q=1}^p \binom{p}{q}\left(\sum_{i \in [n], j \in [m]}\left\vert A_{ij} - \langle X_i, Y_j \rangle\right\vert^p\right)^{\frac {p-q}{p}} \left( \sum_{i \in [n], j \in [m]} \left(k \cdot \gamma \cdot \alpha\right)^p \right)^{\frac q p} \\
        &= \norm{A - XY^\top}_p^p + \sum_{q=1}^p \binom{p}{q} \norm{A - XY^\top}_p^{p-q} \cdot \left(nm\right)^{\frac q p} k^q \gamma^q \cdot \alpha^q \,.
    \end{align*}
    We now note that for all $1 \leq q \leq p$ we have that 
    \begin{align*}
        \norm{A - XY^\top}_p^{p-q} \cdot \left(nm\right)^{\frac q p} k^q \gamma^q \cdot \alpha^q &\leq \norm{A - XY^\top}_p^{p-1} \cdot \left(nm\right)^{\frac 1 p} k \gamma \alpha + \left(nm\right) k^p \gamma^p\alpha^p \\
        &\leq nm \cdot \left(\norm{A - XY^\top}_p^{p-1} \cdot \cdot k \cdot \gamma \alpha + k^p \cdot \gamma^p \cdot \alpha^p\right)\,.
    \end{align*}
    Therefore, 
    \begin{align*}
        \norm{A - \widetilde{X} \widetilde{Y}^\top}_p^p &\leq \norm{A - XY^\top}_p^p + \sum_{q=1}^p \binom{p}{q} \norm{A - XY^\top}_p^{p-q} \cdot \left(nm\right)^{\frac q p} k^q \gamma^q \cdot \alpha^q\\
        &\leq \norm{A - XY^\top}_p^p + O_p(1) \cdot nm \cdot \left(\norm{A - XY^\top}_p^{p-1} \cdot k \cdot \gamma \cdot \alpha + k^p \cdot \gamma^p \cdot \alpha^p \right)\,.
    \end{align*}
    Now note that if $A,X,Y$ all have bit complexity $B$ then we have that $\norm{A - XY^\top}_p^{p-1} \leq nm \cdot 2^{O_p(B)}$ and thus 
    \[ O_p(1) \cdot nm \cdot \left(\norm{A - XY^\top}_p^{p-1} \cdot k\gamma \alpha+ k^p \gamma^p \alpha^p\right) \leq 2^{O_p(B)} \cdot (nm)^2 \cdot \left( k \cdot \gamma + k^p \cdot \gamma^p\right)\,.\]
    Thus, if we take $\gamma \leq 2^{-\Omega_p(B)} \cdot \frac{1}{\poly(nm)} \cdot k^{-p} \cdot \gamma'$ then we have that 
    \[ \norm{A - \widetilde{X} \widetilde{Y}^\top}_p^p \leq \norm{A - XY^\top}_p^p + \gamma'\,.\]
\end{proof}

Similarly to the rank $1$ case, we must also show that our constant factor approximation solutions $X,Y$ have bounded magnitude entries.
\begin{lemma}
\label{lem:constant-factor-bounded-bit-complexity}
    Let $A \in \mathbb{R}^{n \times m}$ have bit complexity $B$ and let $G \in \mathbb{R}^{n \times m}$ satisfy $\rank(G) = k$ and 
    \[ \norm{A - G}_p^p \leq \norm{A}_p^p\,,\]
    Let $G = U\Sigma V^\top$ be the SVD factorization of $G$. Then $X = U \Sigma^{1/2}$ and $Y = V \Sigma^{1/2}$ satisfy $\norm{X}_\infty, \norm{Y}_\infty \leq \poly_p(nmk) \cdot 2^{O_p(B)}$. 
\end{lemma}
\begin{proof}
    Note that since $\norm{U_i}_2^2, \norm{V_j}_2^2 \leq 1$ it suffices to bound the singular values or $\norm{G}_F^2$. Observe that since $A$ has bounded bit complexity, we have that 
    \[ \norm{A - G}_p^p \leq \norm{A}_p^p \leq nm \cdot 2^{pB}\,.\]
    Furthermore, since $\ell_p$ and $\ell_2$ norms are equivalent up to polynomial in dimension factors, we also have that
    \[ \norm{A - G}_F^2 \leq \poly(nm) \cdot 2^{pB}\,.\]
    Thus, we have that 
    \[ \norm{G}_F^2 \leq O\left(\norm{A - G}_F^2 + \norm{A}_F^2\right) \leq \poly(nm) \cdot 2^{pB}\,.\]
    Therefore, all singular values are most $\poly(nm) \cdot 2^{pB}$ and therefore $\norm{X}_\infty, \norm{Y}_\infty \leq \poly_p(nmk) \cdot 2^{O_p(B)}$.
\end{proof}

\paragraph{$p\rightarrow q$ Discretization.} We now also show that our discretization is sufficient for $p\rightarrow q$ norms. 

\begin{lemma}
\label{lem:p-q-disc}
    Let $p \leq 2 \leq q$ and let $p^*$ be the H\"older dual of $p$. Let $A \in \mathbb{R}^{n \times m}$ and let $\Sigma$ be a $1/(\poly(nm/\eps))$-fineness grid of $[-1,1]$. Then we have that
    \[ \max_{\norm{x}_p \leq 1, \norm{y}_{q^*} \leq 1} \langle y, Ax\rangle \leq \max_{\norm{x}_p \leq 1, \norm{y}_{q^*} \leq 1, x_i \in \Sigma, y_j \in \Sigma} \langle y, Ax\rangle + \eps \cdot \left(\sum_{i \leq n} \|A_i\|_{p^*}^q\right)^{1/q}\,.\]
\end{lemma}
\begin{proof}
    Let $x,y$ be the vectors maximizing 
    \[ \max_{\norm{x}_p \leq 1, \norm{y}_{q^*} \leq 1} \langle y, Ax\rangle\,. \]
    Note that since $\norm{x}_p \leq 1, \norm{y}_{q^*} \leq 1$ we have that all $y_i, x_j \in [-1,1]$. Thus, if we let $\bar{x}, \bar{y}$ be $x,y$ with each entry being rounded to the nearest element in $\Sigma$ then we have that $\norm{x-\bar{x}}_\infty, \norm{y-\bar{y}}_\infty \leq 1/(\poly(nm/\eps))$. 
    \begin{align*}
        \langle y, Ax \rangle - \langle \bar{y}, A \bar{x}\rangle &= \langle y - \bar{y}, A x\rangle + \langle \bar{y}, A (x - \bar{x})\rangle\,.
    \end{align*}
    Note that for any $u \in \mathbb{R}^n,v \in \mathbb{R}^m$we have that 
    \[ u^\top A v \leq \norm{u}_\infty \norm{v}_\infty \sum_{i \in [n]} \norm{A_i}_1 \,.\]
    Applying this above we have that
    \[ \langle y, Ax \rangle - \langle \bar{y}, A \bar{x}\rangle \leq \left(\norm{y-\bar{y}}_\infty + \norm{x-\bar{x}}_\infty\right) \cdot \sum_{i \in [n]} \norm{A_i}_1\,.\]
    We now bound $\sum_{i \in [n]} \norm{A_i}_1$. Note that $\norm{A_i}_1 \leq m \norm{A_i}_{p^*}$ so we have that 
    \[ \sum_{i \in [n]} \norm{A_i}_1 \leq m \sum_{i \in [n]} \norm{A_i}_{p^*}\,,\]
    and $q \geq 2$ so by H\"olders
    \[ m \sum_{i \in [n]} \norm{A_i}_{p^*} \leq mn^{1-1/q} \left(\sum_{i \in [n]} \norm{A_i}_{p^*}^q\right)^{1/q} \leq mn \left(\sum_{i \in [n]} \norm{A_i}_{p^*}^q\right)^{1/q}\,.\]
    Finally letting $\norm{y-\bar{y}}_\infty,\norm{x-\bar{x}}_\infty \leq \eps/(nm)$ yields that
    \[ \langle y, Ax \rangle - \langle \bar{y}, A \bar{x}\rangle \leq \eps \left(\sum_{i \in [n]} \norm{A_i}_{p^*}^q\right)^{1/q}\,.\]
\end{proof}

\subsection{Pinning and Mirror Descent}
\label{sec:mirror-descent-overview}
In this section, we give an overview of the heavy-tailed pinning lemma and our new bridge between pinning lemmas and mirror descent.
We start by reviewing the proof of the variance pinning lemma (Theorem~\ref{thm:pinning-intro-variance}), and re-interpreting that proof as an instance of (online) gradient descent.
Then by replacing gradient descent with its generalization, online mirror descent, we will be able to prove Theorem~\ref{thm:pinning-intro-heavy}.
As discussed earlier, after we developed the connection between pinning lemmas and mirror descent, GPT 5.4 Pro found a proof of Theorem~\ref{thm:pinning-intro-heavy} via a relatively simple reduction to Theorem~\ref{thm:pinning-intro-variance}.
Nonetheless, we present the more sophisticated proof because it presents a clear picture of the landscape of pinning lemmas.

\paragraph{The Variance Pinning Lemma.}
Theorem~\ref{thm:pinning-intro-variance} is proved via a nonnegative potential function:
\[
  \Phi(s) = \E_{S \subseteq [n], |S| = s} \E_{i \sim [n]} \Var(X_i \, | \, X_S) = \frac 1 n \E_{S \subseteq [n], |S| = s} \E \| X - \E[ X \, | \, X_S] \|^2 \, .
\]
$\Phi(s)$ provides a measure of the ``uncertainty about $X$'' remaining after conditioning on $s$ random coordinates of $X$.
Since $\Phi(0) = \E_i \Var(X_i)$, to prove Theorem~\ref{thm:pinning-intro-variance} it is enough to lower bound the decrement in potential function $\Phi(s) - \Phi(s+1)$ in terms of $\E_{S} \E_{ij} |\Cov(X_i,X_j \, | \, X_S)|$; this decrement cannot be large for too many steps and thus eventually $\E_{S} \E_{ij} |\Cov(X_i,X_j \, | \, X_S)$ must be small.
The heavy lifting is done by the following simple fact; we give a particularly suggestive proof below.
\begin{fact}[Conditioning reduces Variance]
\label{fact:var-reduction-techoverview}
Given scalar random variables $u, v$ drawn from the joint distribution $\calD$, the decrease in variance of $u$ when conditioning on $v$ can be lower bounded as follows:
\begin{equation*}
    \Var(u) - \E_v \Var(u \vert v) \geq \Omega(1) \cdot \frac{\Cov^2(u, v)}{\Var(v)}\,.
\end{equation*}
\end{fact}
\begin{proof}
    Assume without loss of generality that $\Cov(u,v) > 0$.
    Note that since the conditional expectation of $u$ given $v$ minimizes mean squared error among all functions which attempt to predict the value of $u$ given knowledge of $v$, we have that 
    \[ \Var(u) - \E_v \Var(u \vert v) \geq \E \left(u - \E u\right)^2 - \min_{f \, : \, \R \rightarrow \R} \E_v \E_u \left(u - f(v)\right)^2\,,\]
    To prove the fact, we consider a particular predictor $f(v)$, which we can view as a gradient descent update to the predictor $\E u$, interpreting $v$ as a ``gradient''.
    In particular, let $f(v) = \E u + \eta (v - \E v)$ for some step size $\eta > 0$ we will choose later.
    Expanding, we have that
    \begin{align*}
        \Var(u) - \E_v \Var(u \vert v) &\geq \E \left(u - \E u\right)^2 - \E_v \E_u \left(u - \E u - \eta (v - \E v)\right)^2 \\
        &= \E \left(u - \E u\right)^2 - \left[ \E \left(u - \E u\right)^2 - 2 \E_v \E_u \eta \left(v - \E v \right)\left(u - \E u\right) + \E_v \eta^2 (v - \E v)^2\right] \\
        &\geq 2 \eta \Cov(u,v) - \eta^2 \Var(v)\,.
    \end{align*}
    Setting $\eta = \frac{1}{100} \cdot \frac{\Cov(u,v)}{\Var(v)}$ completes the proof.
\end{proof}
Given Fact~\ref{fact:var-reduction-techoverview}, a simple Cauchy-Schwarz argument to lower bound $\Phi(s) - \Phi(s+1)$ finishes the proof of Theorem~\ref{thm:pinning-intro-variance}.
Since all the arguments above look only at $t+2$ coordinates of $X$ at a time, they work unchanged for pseudodistributions of degree at least $t+2$.
So now we can move on to the heavy-tailed pinning lemma.
The main idea to obtain the heavy-tailed pinning lemma can be summed up as:
\begin{quote}
    \emph{The above argument proves that conditioning reduces variance via the witness of taking a gradient step with ``gradient'' $X_j$. Gradient descent is a special case of mirror descent with quadratic regularization $\|x\|^2$. Replace this with the regularization function $\|x\|_{1+1/p}^{1+1/p}$.}
\end{quote}
We refer the reader to~\cite{bubeck2015convex} for background material on mirror descent.
This pattern fits a long line of work in theoretical computer science using mirror descent with non-Euclidean regularizers as a proof technique -- e.g., in spectral sparsification \cite{allen2015spectral} and discrepancy theory \cite{DadushJiangReis2022}.
We now carry it out in more detail.

\paragraph{Generalizing the Potential Function via Bregman Divergence.}
For this overview, we will focus on the case of Theorem~\ref{thm:pinning-intro-heavy} where we get a bound on the pairwise correlations which depends on the $4/3$ moments of $X$.
We start by introducing a new potential function, to replace $\Phi(s)$.
Following the potential functions used to analyze mirror descent in the context of online convex optimization, we use Bregman divergences: 

\begin{definition}[Bregman Divergence]
    For a differentiable function $f$, the Bregman divergence associated with $f$ is
    \[ B_f(x, y) = f(x) - f(y) - \langle \nabla f(y), x - y \rangle\,.\]
\end{definition}

We now define the ``Bregman uncertainty'' of a random variable, which will be our measure of its uncertainty, generalizing variance. 

\begin{definition}[Bregman Uncertainty]
\label{def:breg-uncertainty}
    Let $p \in \mathbb{N}^{> 0}$ and let $B_p$ be the Bregman divergence associated with $f(x) = \frac{p-1} {p} \vert x\vert^{\frac {p} {p-1}}$. Then the $p$-Bregman uncertainty of a random variable $u$ is 
    \[\min_c \E_u B_p(u, c)\,.\]
    In fact, the minimum is always obtained at $c = \E u$, so an equivalent formulation is 
    \[
    \E_u B_p(u, \E u) = \frac{p-1} {p} \Paren{\E |x|^{\frac{p+1}p} - |\E x|^{\frac{p+1}{p}}} \, ,
    \]
    which also conveniently shows that $\E_u B_p(u,\E u) \geq 0$.
\end{definition}
When $p=2$, this is precisely the variance.
Directly generalizing from the variance case, our potential function will be
\[
\Phi_p(s) = \E_{S \subseteq [n], |S| = s} \E_{i \sim [n]} \E B_p(X_i, \E [X_i \, | \, X_S]) \geq 0 \, .
\]
It turns out that $\E[X_i \, | \, X_S]$ minimizes $\E B_p (X_i, f(X_S))$ among all functions depending only on $X_S$.
We can think of it as the ``best prediction'' of $X_i$ given only the random variable $X_S$.

\paragraph{Conditioning Reduces Bregman Uncertainty.}
Focusing on the $4/3$ case again, the heart of the argument is the following analogue of Fact~\ref{fact:var-reduction-techoverview}, which will be enough to show a lower bound on $\Phi_4(s) - \Phi_4(s+1)$ in terms of $\E_{S \subseteq [n], |S| = s} \E_{ij} |\Cov(X_i, X_j \, | \, X_S)|$.
\begin{lemma}
\label{lem:bregman-reduction-tech-overview}
    For all jointly distributed random variables $u,v$,
    \[ \E_{u, v} B_4(u, \E u) - B_4(u, \E [u \vert v]) \geq \Omega \left(\frac{\Cov^6 (u, v)}{\left( \E |v|^{4/3} \right)\left( \E_{u,v} [|u|^{4/3} |v|^{4/3}] + \E_u [|u|^{4/3}] \E_v [|v|^{4/3}] \right)^{7/2}}\right) \,.\]
\end{lemma}
Lemma~\ref{lem:bregman-reduction-tech-overview} captures a reduction in the Bregman uncertainty associated with $u$ when conditioning on $v$ for which $|\Cov(u,v)|$ is significant and the $4/3$ moments of $u$ and $v$ are not too big, exactly what we need.
To prove Lemma~\ref{lem:bregman-reduction-tech-overview}, roughly speaking, we construct an explicit function $f(v)$ which is designed to ``predict'' $u$ given only $v$, and prove a lower bound on $B_4(u,\E u) - B_4(u, f(v))$.
The latter then implies a lower bound on $B_4(u, \E u) - B_4(u, \E [u | v])$ because the divergence is minimized at $B_4(u, \E[u | v])$.
At a high level, the idea is to take $f(v)$ to be a mirror-descent step starting at $\E u$ with gradient $v$.
The formula for such $f$ is inscrutable, but can be derived from the mirror descent update rules:
\[
f(v)
=
\operatorname{sign}\!\Big(\operatorname{sign}(\E u)\,|\E u|^{1/3}-\eta\,v\Big)\,
\Big|\operatorname{sign}(\E u)\,|\E u|^{1/3}-\eta\,v\Big|^{3} \, ,
\]
where $\eta$ is an appropriately-chosen step size.
In fact, even this $f$ will not quite work -- for technical reasons, we need to instead use a mirror descent step starting at $v^{2/3} \E u$ with gradient $v^{1/3}$; an explanation why is beyond the scope of this overview.

A major challenge in the proof of Lemma~\ref{lem:bregman-reduction-tech-overview} is that $x \mapsto |x|^{4/3}$ is not strongly convex, so we cannot use the standard analysis of mirror descent for strongly convex regularizers.
Instead we give a ``bare hands'' proof.
To prove Theorem~\ref{thm:pinning-intro-heavy} (the $4/3$ case) from Lemma~\ref{lem:bregman-reduction-tech-overview}, we can use H\"older's inequality, generalizing the argument which proves Theorem~\ref{thm:pinning-intro-variance} from Fact~\ref{fact:var-reduction-techoverview}.

\paragraph{A Unified Proof for Pinning Lemmas.}
Our mirror descent proof strategy for the heavy-tailed pinning lemma can also be adapted to recover the other known pinning lemma, which applies to discrete random variables, by choosing an entropy regularizer.
We give the details in~\cref{sec:discrete-pinning}.
The mirror-descent unification of pinning lemmas is summarized in Table~\ref{tab:pinning-lemmas}.

\begin{table}[h!]
\centering
\begin{tabular}{ccc}
\hline
Regularizer & Class of Random Variables & Reference \\ 
\hline
$x^2$ & finite variance & \cite{Barak2011} \\ 
entropy & discrete/bounded & \cite{Montanari2008,Raghavendra2012}\\ 
$|x|^{1+1/p}$ & heavy-tailed & this work \\ 
Von Neumann entropy & quantum mixed states & \cite{brandao2013product} \\
\hline
\end{tabular}
\caption{Pinning lemmas, unified via mirror descent.}
\label{tab:pinning-lemmas}
\end{table}

\subsection{Discrete Pinning Lemma from Mirror Descent}
\label{sec:discrete-pinning}

We now describe how to obtain discrete pinning lemmas from mirror descent. In particular, we sketch a proof of the following:
\begin{lemma}
\label{lem:discrete-pinning}
    For every $t \in \mathbb{N}$ there is $s \leq t$ such that if $S \subseteq [n]$ is a random set of size $s$ and $(X_1,\ldots,X_n)$ are any jointly distributed discrete random variables, 
    \[
      \E_S \E_{i,j\sim [n]} \tv(\{X_i, X_j \vert X_S\}, \{X_i \vert X_S\} \otimes \{X_j \vert X_S\}) \leq \left( \frac{\E_i H(X_i)}{t}\right)^2 \,.
    \]
\end{lemma}

The above lemma is proved by utilizing the following non-negative potential function:
\[ \Phi(s) = \E_{S \subseteq [n], \vert S \vert = s} \E_{i \sim [n]} H(X_i \vert X_S)\,.\]

In order to prove~\cref{lem:discrete-pinning} it is enough to lower bound the decrement in potential function $\Phi(s) - \Phi(s+1)$. This is captured by the following lemma:
\begin{lemma}
\label{lem:tv-decrease}
    Let $x,y$ be jointly distributed random variables over $[m]$. Then we have that
    \[ \tv^2(\{x,y\}, \{x\} \otimes \{y\}) \leq O(1) \cdot \left(H(x) - H(x \vert y)\right)\,.\]
\end{lemma}

\cref{lem:tv-decrease} is a simple consequence of Pinsker's inequality and the definition of mutual information. From~\cref{lem:tv-decrease} we have that~\cref{lem:discrete-pinning} easily follows. Below we give a proof of~\cref{lem:tv-decrease} using mirror descent.
\begin{proof}
    We consider the following setup for our mirror descent game: let $x,y$ be sampled and let the player guess the hidden value of $x$. Guesses take the form of distributions over $[m]$ or vectors in $\mathbb{R}^m$ on the simplex, where the correct guess is $\mathbf{1}_x$.

    We study how a player can update their guess for $x$ when given $y$, when measuring error with respect to the Bregman divergence associated with $f(z) = -\sum_{i \in [m]} z_i \log z_i$. Note that the Bregman divergence between two distributions is the KL divergence between them. When given no other information other than the distribution of $x$, the guess which minimizes the Bregman divergence to the value $\mathbf{1}_x$ in expectation is $\E \mathbf{1}_x = \{x\}$. Furthermore, note that this value is 
    \[ \E_x B(\mathbf{1}_x, \{x\}) = \E_x \log \left(\frac{1}{p(x)}\right) = H(x)\,.\]
    We now explicitly construct a function $f: \mathbb{R} \rightarrow \mathbb{R}^m$ which updates the original guess for $x$ (which was $\{x\}$) when given the value of $y$. This function $f$ is a single step of (projected) mirror descent with the negative entropy regularizer. Note that negative entropy is $1$ strongly convex with respect to $\ell_1$. Therefore, for any cost function $g_y$, we have by~\cref{lem:mirror-descent-progress}
    \[ \eta \langle g_y, \mathbf{1}_x - \{x\}\rangle \leq B(\mathbf{1}_x, \{x\}) - B(\mathbf{1}_x, z) + \eta^2 \Vert g_y\Vert_\infty^2\,,\]
    where $z$ is the result of taking a mirror descent step with gradient $g$ and step size $\eta$ from $\{x\}$ and then Bregman projecting the result back to the simplex.

    We now describe how to pick $g_y$ based on the value of $g$. Note that by the variational characterization of $\tv$ there exists some function $h$ such that $\norm{h}_\infty \leq 1$ and
    \[ \tv(\{x,y\}, \{x\} \otimes \{y\}) = \E_{\{x,y\}} h(x,y) - \E_{\{x\} \otimes \{y\}} h(x,y)\,.\]
    Given such a function $h$ we let $g_y = h(\cdot, y)$. Plugging this into the above inequality and taking an expectation over the joint distribution of $x,y$ we have that 
    \[ \eta \E_{x,y} \langle h(\cdot, y), \mathbf{1}_x - \{x\}\rangle \leq \E_x B(\mathbf{1}_x, \{x\}) - \E_{x,y} B(\mathbf{1}_x, z) + \eta^2 \E_{x,y} \Vert h(\cdot, y)\Vert_\infty^2 \,.\]
    Note that we have that 
    \[ \E_{x,y} \langle h(\cdot, y), \mathbf{1}_x\rangle = \E_{x,y} h(x,y)\,,\]
    while 
    \[ \E_{x,y} \langle h(\cdot, y), \{x\}\rangle = \E_{\{x\} \otimes \{y\}} h(x,y)\,,\]
    so the left hand side of this inequality is $\eta \tv(\{x,y\}, \{x\} \otimes \{y\})$. Furthermore, since $h$ is bounded in $\ell_\infty$ norm by $1$ we have that 
    \[\eta \tv(\{x,y\}, \{x\} \otimes \{y\}) \leq \E_x B(\mathbf{1}_x, \{x\}) - \E_{x,y} B(\mathbf{1}_x, z) + \eta^2\,.\]
    We now note that the value minimizing the expected Bregman divergence to $\mathbf{1}_x$ out of all functions of $y$ is simply $\{x \vert y\}$. Combining with the fact that this expected Bregman divergence is entropy (as noted above) we have that 
    \[ \eta \tv(\{x,y\}, \{x\} \otimes \{y\}) \leq H(x) - H(x \vert y) + \eta^2\,.\]
    Setting $\eta = \frac{1}{2} \tv(\{x,y\}, \{x\} \otimes \{y\})$ completes the proof.
\end{proof}

\subsection{Impossibility of Heavy-Tailed Pinning Lemma Without Pairwise Upper Bound}
\label{sec:pairs-are-necessary}

We show that the following ``dream statement'' is false.

\begin{conjecture}[False]
For every $p \in \N$ and $t \in \N$ there is $s \leq t$ such that if $S \subseteq [n]$ is a random set of size $s$ and $(X_1,\ldots,X_n)$ are any jointly distributed random variables,
\[
  \E_S \E_{i,j\sim [n]} | \Cov(X_i,X_j \, | \, X_S) | \leq \frac 1 {t^{\Omega(1)}} \cdot
	\Brac{ \E_{X} \Paren{ \E_{i \sim [n]} |X_i|^{\frac {p+1}{p}} }}^{\frac {2p} {p+1}}
\]
\end{conjecture}
The only difference between the above conjecture and the heavy-tailed pinning lemma is whether the exponent $2$ lies inside or outside the final brackets on the right-hand side -- this makes a crucial difference!

Consider the following family of counterexamples, for the $4/3$ case (they generalize easily to other cases).
We divide the coordinates into $n/2$ pairs.
For each pair, say $X_{i}$ and $X_{i+1}$, we independently draw a random variable $Y_i$ with infinite first moment but finite $2/3$ moment.
Let $X_i, X_{i+1}$ both be equal to the square root of $|Y_i|$.
Then even after pinning $o(n)$ coordinates, the left-hand side $\E_S \E_{i,j \sim [n]} |\Cov(X_i, X_j \, | \, X_S)$ is infinite, because $\Cov(X_i, X_{i+1}) = \E |Y_i| - (\E \sqrt{|Y_i|})^2$, which is infinite.
But the right-hand side depends only on the $4/3$ moments of the $X_i$ variables, which are $2/3$ moments of $Y_i$, so the right-hand side is finite.

\subsection{Truncation-based Proof of Heavy-Tailed Pinning Lemma}
\label{sec:truncation-proof}
We present here the proof of the heavy-tailed pinning lemma which avoids the mirror-descent technology, instead going by a more direct reduction to the variance pinning lemma.
This proof was discovered by GPT, in the following dialogue (note that at this point the human interlocuter does actually posess a statement and proof of the heavy tailed pinning lemma, but does not admit this to the LLM, to see if the LLM can produce the statement on its own):

\subsubsection{Dialogue}

\textbf{Human:}

\begin{lstlisting}
the pinning lemma shows that if you have an n dimensional random variable, and you condition on poly(1/eps) indices, you get the following guarantee:

E_S E_{ij} |Cov(Xi, Xj | X_S)| leq eps E_i Var(Xi)

where S is a random set of indices of size S.

now, i want to know if there is a similar lemma which would be true for X where the variance might not exist, or might be very large. instead, i want the right hand side to depend only on the q-th moments of Xi, where q is something less than 2, like 4/3. I still want the LHS to be the average pairwise covariance. It would be ok if the the i neq j pairs are not included in the LHS. It would also be ok if the RHS depended on the pairwise products E Xi^q Xj^q. Even though the variance might be very large, the i = j term should show up only like 1/n of all the i,j pairs so having this kind of dependence on the RHS would be ok.

I don't know exactly what such a statement would look like. Could you discover and prove such a statement for me, or tell me why it can't exist, if there is a roadblock? The basic idea is that I want the average pairwise covariance to become small after conditioning on poly(1/eps) coordinates, but I want the RHS bound to be about qth moments rather than variances ie allowing the Xi to be heavy tailed.
\end{lstlisting}

\textbf{GPT:} (describes a counterexample showing that the RHS cannot depend only on $\E_i |X_i|^q$ when $ q < 2$, and shows that if we truncate/clip the $X_i$ varaibles then this issue disappears but the resulting pinning lemma will only apply to the clipped random variables)

\textbf{Human:}

\begin{lstlisting}
ok, suppose that i am also allowed to depend on E_ij |Xi Xj|^q but i don't want to go around losing some extraneous factors of n -- that thing will give me some bound on E_i |Xi|^2q but that bound will be dumb by a factor of n
\end{lstlisting}

\textbf{GPT:} (now produces the proof presented below)

\subsubsection{Proof of Heavy-Tailed Pinning Lemma via Truncation}

\begin{proposition}[Heavy-tailed pinning via truncation]
Assume \(1<q<2\), and let \(X_1,\dots,X_n\) be real-valued random variables. Define
\[
M_q
:=
\mathbb E_{i,j}\,\mathbb E|X_iX_j|^q
=
\mathbb E\!\left(\mathbb E_i |X_i|^q\right)^2.
\]
Assume \(M_q<\infty\). Let \(\pi\) be a uniformly random permutation of \([n]\), independent of \(X\), and define
\[
\mathcal F_t := \sigma\bigl(\pi(1),\dots,\pi(t),\,X_{\pi(1)},\dots,X_{\pi(t)}\bigr),
\qquad 0\le t\le n.
\]
Then for every \(m\in\{0,\dots,n-1\}\), there exists \(t\in\{0,\dots,m\}\) such that
\[
\mathbb E\Big[\,\mathbb E_{i,j}\big|\operatorname{Cov}(X_i,X_j\mid \mathcal F_t)\big|\,\Big]
\;\le\;
7\,(m+1)^{-(q-1)/2}\,M_q^{1/q}.
\]
Consequently, if
\[
m+1 \ge \left(\frac{7}{\varepsilon}\right)^{2/(q-1)},
\]
then for some \(t\le m\),
\[
\mathbb E\Big[\,\mathbb E_{i,j}\big|\operatorname{Cov}(X_i,X_j\mid \mathcal F_t)\big|\,\Big]
\;\le\;
\varepsilon\,M_q^{1/q}.
\]
\end{proposition}

\begin{proof}
If \(M_q=0\), then \(X_i=0\) a.s.\ for every \(i\), and there is nothing to prove. So assume \(M_q>0\).

Set
\[
A := \mathbb E_i |X_i|^q.
\]
Then
\[
M_q=\mathbb E A^2.
\]

Fix a truncation level \(L>0\), and write
\[
Y_i := \tau_L(X_i):=\operatorname{sgn}(X_i)\min\{|X_i|,L\},
\qquad
R_i := X_i-Y_i.
\]
Thus \(X_i=Y_i+R_i\), with \(Y_i\) the clipped part and \(R_i\) the tail.

Apply the variance pinning lemma to $Y_i$. Since \(Y_i\) is a measurable function of \(X_i\), there exists \(t\in\{0,\dots,m\}\) such that
\[
\mathbb E\Big[\,\mathbb E_{i,j}\big|\operatorname{Cov}(Y_i,Y_j\mid \mathcal F_t)\big|\,\Big]
\le
(m+1)^{-1/2}\,\mathbb E_i \operatorname{Var}(Y_i).
\]
Now
\[
Y_i^2 \le L^{2-q}|X_i|^q,
\]
hence
\[
\mathbb E_i \operatorname{Var}(Y_i)
\le
\mathbb E_i \mathbb E Y_i^2
\le
L^{2-q}\,\mathbb E A
\le
L^{2-q}(\mathbb E A^2)^{1/2}
=
L^{2-q}M_q^{1/2}.
\]
Therefore
\[
\mathbb E\Big[\,\mathbb E_{i,j}\big|\operatorname{Cov}(Y_i,Y_j\mid \mathcal F_t)\big|\,\Big]
\le
(m+1)^{-1/2}L^{2-q}M_q^{1/2}.
\]

Next define
\[
U := \mathbb E_i |R_i|,
\qquad
V := \mathbb E_i |Y_i|.
\]
Since
\[
|R_i|
=
|X_i|\,\mathbf 1_{\{|X_i|>L\}}
\le
L^{1-q}|X_i|^q,
\]
we get
\[
U \le L^{1-q}A,
\qquad\text{hence}\qquad
\mathbb E U^2 \le L^{2-2q}M_q.
\]
Also, \(|Y_i|\le |X_i|\), so by the power-mean inequality,
\[
V
\le
\mathbb E_i |X_i|
\le
\left(\mathbb E_i |X_i|^q\right)^{1/q}
=
A^{1/q}.
\]
Therefore
\[
\mathbb E V^2
\le
\mathbb E A^{2/q}
\le
(\mathbb E A^2)^{1/q}
=
M_q^{1/q},
\]
since \(2/q\le 2\).

Now decompose the covariance:
\[
\operatorname{Cov}(X_i,X_j\mid \mathcal F_t)
=
\operatorname{Cov}(Y_i,Y_j\mid \mathcal F_t)
+
\operatorname{Cov}(R_i,Y_j\mid \mathcal F_t)
+
\operatorname{Cov}(Y_i,R_j\mid \mathcal F_t)
+
\operatorname{Cov}(R_i,R_j\mid \mathcal F_t).
\]

For the mixed term,
\begin{align*}
\mathbb E_{i,j}\big|\operatorname{Cov}(R_i,Y_j\mid \mathcal F_t)\big|
&\le
\mathbb E_{i,j}\mathbb E\big[|R_iY_j|\mid \mathcal F_t\big]
+
\mathbb E_i\mathbb E\big[|R_i|\mid \mathcal F_t\big]\,
\mathbb E_j\mathbb E\big[|Y_j|\mid \mathcal F_t\big] \\
&=
\mathbb E[UV\mid \mathcal F_t]
+
\mathbb E[U\mid \mathcal F_t]\,
\mathbb E[V\mid \mathcal F_t].
\end{align*}
Taking expectations and using Cauchy--Schwarz (together with Jensen for the second term),
\begin{align*}
\mathbb E\Big[\,\mathbb E_{i,j}\big|\operatorname{Cov}(R_i,Y_j\mid \mathcal F_t)\big|\,\Big]
&\le
\mathbb E[UV]
+
\mathbb E\!\big[\mathbb E(U\mid \mathcal F_t)\mathbb E(V\mid \mathcal F_t)\big] \\
&\le
2(\mathbb E U^2)^{1/2}(\mathbb E V^2)^{1/2} \\
&\le
2L^{1-q}M_q^{(q+1)/(2q)}.
\end{align*}
The same bound holds with \(R_i,Y_j\) replaced by \(Y_i,R_j\).

For the tail-tail term,
\begin{align*}
\mathbb E_{i,j}\big|\operatorname{Cov}(R_i,R_j\mid \mathcal F_t)\big|
&\le
\mathbb E_{i,j}\mathbb E\big[|R_iR_j|\mid \mathcal F_t\big]
+
\mathbb E_i\mathbb E\big[|R_i|\mid \mathcal F_t\big]\,
\mathbb E_j\mathbb E\big[|R_j|\mid \mathcal F_t\big] \\
&=
\mathbb E[U^2\mid \mathcal F_t]
+
\mathbb E[U\mid \mathcal F_t]^2.
\end{align*}
Taking expectations and using Jensen,
\[
\mathbb E\Big[\,\mathbb E_{i,j}\big|\operatorname{Cov}(R_i,R_j\mid \mathcal F_t)\big|\,\Big]
\le
2\mathbb E U^2
\le
2L^{2-2q}M_q.
\]

Combining the last three displays,
\[
\mathbb E\Big[\,\mathbb E_{i,j}\big|\operatorname{Cov}(X_i,X_j\mid \mathcal F_t)\big|\,\Big]
\le
(m+1)^{-1/2}L^{2-q}M_q^{1/2}
+
4L^{1-q}M_q^{(q+1)/(2q)}
+
2L^{2-2q}M_q.
\]

Choose
\[
L := M_q^{1/(2q)}(m+1)^{1/2}.
\]
Then the first two terms become
\[
(m+1)^{-(q-1)/2}M_q^{1/q}
\quad\text{and}\quad
4(m+1)^{-(q-1)/2}M_q^{1/q},
\]
while the last term becomes
\[
2(m+1)^{-(q-1)}M_q^{1/q}
\le
2(m+1)^{-(q-1)/2}M_q^{1/q}.
\]
Hence
\[
\mathbb E\Big[\,\mathbb E_{i,j}\big|\operatorname{Cov}(X_i,X_j\mid \mathcal F_t)\big|\,\Big]
\le
7\,(m+1)^{-(q-1)/2}M_q^{1/q}.
\]
This proves the proposition.
\end{proof}

\end{document}